\documentclass[a4paper,UKenglish,dvipsnames]{article}

\usepackage[utf8]{inputenc}
\usepackage{fullpage}
\usepackage{amssymb}
\usepackage{amsmath}
\usepackage{amsthm}

\usepackage[authoryear, square]{natbib}

\usepackage{color,tikz}
\usepackage{wrapfig}
 \usepackage{pict2e}
\usepackage{graphicx} %
\usepackage{wrapfig}
\usepackage[textsize=scriptsize, textwidth=1.5cm, obeyFinal]{todonotes}
  
\usetikzlibrary{external}
\usepackage{listings}
\usepackage{comment}
\usepackage{booktabs}
\newcommand*{\otoprule}{\midrule[\heavyrulewidth]}
\usepackage{multirow}
\usepackage{makecell}

\usepackage{mathtools}

\usepackage{tcolorbox}

\usepackage[retainorgcmds]{IEEEtrantools}
\usepackage{pifont}
\newcommand*{\cmark}{\ding{51}}
\newcommand*{\xmark}{\ding{55}}

\usepackage{stmaryrd}
\usepackage{paralist}

\usepackage{ltl}
\usepackage{tikz}
\usetikzlibrary{automata, arrows, arrows.meta, positioning, shapes}
\usetikzlibrary{decorations,decorations.markings,calc, decorations.text, matrix,decorations.pathmorphing,patterns,backgrounds,shapes.misc,arrows.meta,positioning}

\tikzset{
	state/.style=
    {circle, draw, align=center, auto, initial text={}}, 
	>=stealth,
	loopright/.style={loop,looseness=5,out=35, in=-35},
	loopleft/.style={loop,looseness=5,out=215, in=145},
	loopabove/.style={loop,looseness=5,out=125, in=55},
	loopbelow/.style={loop,looseness=5,out=-55, in=-125}
}

\usetikzlibrary{fit} %
\usetikzlibrary{backgrounds} %
\usetikzlibrary{patterns}
\usepgflibrary{shapes}
\usetikzlibrary{shadows}
\tikzstyle{cirre}=[draw=green!60!red, fill=green!20!white,circle,minimum size=1.4em,inner sep=0em]                                                                                                      
\tikzstyle{cir1re}=[draw=red!80!violet, fill=red!20!white,circle,minimum size=1.4em,inner sep=0em]  

\tikzstyle{cir}=[draw=violet, fill=red!20!white,circle,minimum size=1.4em,inner sep=0em]                                                                                                      
 \tikzstyle{dia}=[draw=green!80!red, fill=green!20!white, diamond,minimum size=1.4em,inner sep=0.1em]                                                                                                      
\tikzstyle{cir1}=[draw=violet, fill=violet!20!white,circle,minimum size=1.4em,inner sep=0em]

\tikzstyle{background}=[rectangle,fill=gray!10, inner sep=0.1cm, rounded corners=0mm]
\tikzstyle{loc}=[draw,rectangle,minimum size=1.4em,inner sep=0em]
\tikzstyle{trans}=[-latex, rounded corners]
\tikzstyle{trans2}=[-latex, dashed, rounded corners]
\tikzset{snake it/.style={decorate, decoration=snake}}%

\tikzstyle{inv}=[]
\tikzstyle{varpass}=[]

\usepackage{microtype}
\AtBeginEnvironment{verbatim}{\microtypesetup{activate=false}}

\usepackage{hyperref}
\usepackage{caption}

\usepackage[capitalise,nameinlink]{cleveref}

\crefname{section}{Sect.}{Sect.}
\Crefname{section}{Section}{Sections}
\crefname{listing}{List.}{List.}
\crefname{listing}{Listing}{Listings}
\Crefname{listing}{Listing}{Listings}

\def\myshift#1{\raisebox{1.5ex}}
\usepackage{fontawesome5}
\newcommand{\until}{\mathbin{\mathbf{{U}}}}
\newcommand{\release}{\mathbin{\mathbf{{R}}}}
\newcommand{\since}{\mathbin{\mathbf{{S}}}}
\newcommand{\trigger}{\mathbin{\mathbf{{T}}}}

\newcommand{\tctl}{\textmd{\textup{\textsf{TCTL}}}}
\newcommand{\mtl}{\textmd{\textup{\textsf{MTL}}}}
\newcommand{\ltl}{\textmd{\textup{\textsf{LTL}}}}
\newcommand{\stl}{\textmd{\textup{\textsf{STL}}}}
\newcommand{\mitlpp}{\textmd{\textup{\textsf{MITPPL}}}}
\newcommand{\tptl}{\textmd{\textup{\textsf{TPTL}}}}
\newcommand{\uptl}{\textmd{\textup{\textsf{UPTL}}}}

\newcommand{\mightyppl}{\textmd{\textup{\textsc{MightyPPL}}}}
\newcommand{\mightyl}{\textmd{\textup{\textsc{MightyL}}}}
\newcommand{\mitl}{\textmd{\textup{\textsf{MITL}}}}
\newcommand{\mtlpp}{\textmd{\textup{\textsf{MTLPPL}}}}

\newcommand{\qtwomlo}{\mathsf{Q2MLO}}

\newcommand*{\gta}[1]{\mathsf{GTA}\text{#1}}
\newcommand*{\ta}[1]{\mathsf{TA}\text{#1}}
\newcommand*{\dc}[1]{} %
\newcommand*{\simplify}[1]{\overline{#1}}
\newcommand*{\newneg}[1]{\neg{\overline{#1}}}

\newcommand{\expspace}{\mathsf{EXPSPACE}}
\newcommand{\pspace}{\mathsf{PSPACE}}

\newcommand*{\AP}{\mathsf{AP}}

\newcommand{\F}{\mathcal{F}}

\newcommand{\diamondminus}{%
  \sbox0{$\lozenge$}%
  \usebox0\kern-.5\wd0\clap{\raisebox{.1ex}{\scalebox{.7}[1]{$-$}}}\kern.5\wd0%
}
\DeclareMathOperator{\past}{\once}

\DeclareMathOperator{\Boxminus}{\oset{\leftarrow}{\LTLsquare}}
\DeclareMathOperator{\nm}{\oset[-1pt]{\leftarrow}{\LTLcircle}}
\DeclareMathOperator{\nex}{\LTLcircle}

\newcommand*\Transitions{\Delta}
\newcommand{\R}{\mathbb{R}_{\geq 0}}
\newcommand{\N}{\mathbb{N}}

\newcommand{\I}{\mathbb{I}}

\renewcommand{\max}{\mathsf{max}}

\newcommand{\Guards}{\mathcal{G}}

\newcommand*\sem[1]{\ensuremath{\llbracket#1\rrbracket}}

\definecolor{saffron}{rgb}{1.0,0.49,0.0}

\makeatletter
\newcommand{\oset}[3][0ex]{%
  \mathrel{\mathop{#3}\limits^{
    \vbox to#1{\kern-2\ex@
    \hbox{$\scriptstyle#2$}\vss}}}}
\makeatother

\newcommand{\Pnkern}{%
  \mkern-2mu
}
\DeclareMathOperator{\eventually}{\LTLdiamond}
\DeclareMathOperator{\fut}{\LTLdiamond}
\DeclareMathOperator{\once}{\oset[-1pt]{\leftarrow}{\LTLdiamond}}
\DeclareMathOperator{\globally}{\LTLsquare}

\DeclareMathOperator{\PnF}{\mathbf{P \Pnkern n}}
\DeclareMathOperator{\PnO}{\oset[-1pt]{\leftarrow}{\mathbf{P \Pnkern n}}}
\DeclareMathOperator{\dualPnF}{\vphantom{\mathbf{P \Pnkern n}}\smash{{\mathbf{P \Pnkern n}}^{\mathrlap{\sim}}}}
\DeclareMathOperator{\dualPnO}{\vphantom{\mathbf{P \Pnkern n}}\smash{\oset[-1pt]{\leftarrow}{\mathbf{P \Pnkern n}}}^{\mathrlap{\sim}}}

\newtheorem{lemma}{Lemma}
\newtheorem{theorem}{Theorem}
\newtheorem{remark}{Remark}

\newtheorem{example}{Example}

\title{$\mightyppl$: Verification of  \mitl{} with Past and Pnueli Modalities}

\author{%
Hsi-Ming Ho (University of Sussex, Brighton, United Kingdom)\\%
Shankara Narayanan Krishna (Indian Institute of Technology Bombay, Mumbai, India)\\%
Khushraj Madnani (Indian Institute of Technology Guwahati, Guwahati, India)\\%
Rupak Majumdar (MPI-SWS, Kaiserslautern, Germany)\\%
Paritosh Pandya (Indian Institute of Technology Bombay, Mumbai, India)}
\date{}

\begin{document}

\maketitle

\begin{abstract}
\emph{Metric Interval Temporal Logic} ($\mitl{}$) is a popular formalism for specifying properties of reactive systems with timing constraints.
Existing approaches to using $\mitl{}$ in verification tasks, however, have notable drawbacks: they either support only limited fragments of the logic or allow for only incomplete verification.
This paper introduces $\mightyppl{}$, a new tool for translating 
formulae in 
\emph{Metric Interval Temporal Logic with Past and Pnueli modalities} ($\mitlpp{}$) over the pointwise semantics 
into timed automata. $\mightyppl{}$
 enables \emph{satisfiability} and \emph{model checking} of a  
much more expressive 
 specification logic over both finite and infinite words 
  and incorporates a number of performance optimisations, including a novel symbolic encoding of transitions and a symmetry reduction technique that leads to an exponential improvement in the number of reachable discrete states. For a given $\mitlpp{}$ formula, $\mightyppl{}$ can generate either a network of timed automata or a single timed automaton that is language-equivalent and compatible with multiple verification back-ends, including  \textsc{Uppaal},  \textsc{TChecker}, and \textsc{LTSmin}, which supports multi-core model checking.
  We evaluate the performance of the toolchain  
across various case studies and configuration options. 
\end{abstract}

\section{Introduction}

Real-time logics provide a formal framework for specifying and reasoning about time-dependent behaviours of reactive systems (see, e.g.,~\citep{AluHen92, AH93, H98, bouyer2009model, bouyer2017timed}).
\emph{Metric Interval Temporal Logic} ($\mitl{}$)~\citep{AFH96} is a  prominent real-time logic which  extends \emph{Linear Temporal Logic} ($\ltl{}$)~\citep{4567924} with constructs that associate \emph{time intervals} with temporal operators. 
For example, the property `each request must be followed by an acknowledgement within $5$ seconds' can be written as
\[
\globally (\textit{req} \implies \eventually_{[0, 5]} \textit{ack}).
\]
Thanks to their familiar $\ltl$-like syntax that appeals to practitioners,
$\mitl{}$ and its signal-based variant \emph{Signal Temporal Logic} ($\stl{}$)~\citep{maler2004monitoring} are now widely used in the design and analysis of cyber-physical systems (CPSs) in various safety-critical application domains such as automotives~\citep{jin2014powertrain, deshmukh2017robust}, robotics~\citep{raman2014model}, medical monitoring~\citep{roohi2018parameter}, smart grids~\citep{beg2018signal} and so on. 

As in conventional discrete-time settings, a fundamental challenge in the algorithmic verification of real-time systems is to balance \emph{expressiveness} (the ability to specify sophisticated properties) with \emph{decidability} and \emph{tractability} (the feasibility of verification tasks). 
For $\mitl{}$, \emph{satisfiability} and \emph{model checking} are both decidable with reasonable computational complexity ($\expspace{}$- or $\pspace{}$-complete, depending on the timing constraints allowed in formulae~\citep{AFH96, raskin1997state, H98}). 
However, despite numerous advances in \emph{incomplete} verification methods such as \emph{monitoring}~\citep{donze2013efficient, nivckovic2020rtamt} and \emph{falsification}~\citep{annpureddy2011s, abbas2013probabilistic, waga2020falsification},
\emph{complete} verification methods are often overlooked in practice due to lack of tool support. 
We also note that $\mitl{}$ itself has limited expressive power as it cannot express some natural properties. For example, the `counting' property `$p$ must occur at least twice within the next $10$ time units' can be expressed in the decidable real-time logic $\qtwomlo{}$~\citep{Hirshfeld2004}  but not in $\mitl{}$~\citep{AH93, Rabinovich2007, bouyer2010expressiveness, concur11}.

\paragraph{Pnueli modalities.}
Research has shown that $\mitl{}$ can be extended with a more general form of counting while preserving decidability~\citep{rabinovichY}. This is achieved with \emph{Pnueli modalities}, $\PnF_I(\varphi_1, \dots, \varphi_k)$, which specifies that a sequence of events in an interval $I$ must satisfy $\varphi_1, \dots, \varphi_k$
in the given order. To highlight their practical utility, we now consider an example.

\begin{example}[Adapted from~\citep{concur25}]\label{ex:food}
 Consider a city with three eateries: a pizzeria (serving pizzas, denoted by $P$), a burger joint (serving burgers, $B$), and a cafe (serving coffees, $C$).
Two locations, $L1$ and $L2$, represent the origin points for customer orders. 
 The delivery driver starts at a designated initial location $L0$.
 See~\cref{fig:icon-delivery-map} for a map where $L0$ is the leftmost location, $L2$ is the rightmost location, and $L1$ is the bottom-right location. 
We write $K{:}L$ (where $K \in \{P, B, C\}$ and $L \in \{L1, L2\}$) for an order from location $L$ for item $K$. For example, $P{:}L1$ denotes that $L1$ has ordered a pizza. The property `once $K{:}L$ occurs, the delivery driver should  pick up the item $K$ and deliver it to $L$ within the next $15$ minutes' can be written as
\[
\globally \big(
K{:}L \implies \PnF_{[0,15]}(K, L) \big)  \;.
\]

\begin{figure}[!htbp]
\centering
\scalebox{0.8}{\begin{tikzpicture}[->, auto, node distance=1.5cm, line width=0.3mm,
  location/.style={circle, draw=black, minimum size=11mm, thick, fill=blue!10, font=\footnotesize},
  eatery/.style={circle, draw=black, minimum size=11mm, thick, fill=orange!20, font=\footnotesize},
  font=\small]

\node (q0) [location, initial, initial where=left, initial distance=0.3cm, initial text={}] at (-5,0) {\faMotorcycle};
\node (p)  [eatery] at (-2.5,2) {\faPizzaSlice};
\node (c)  [eatery] at (1,2)    {\faCoffee};
\node (b)  [eatery] at (-2,-2)  {\faHamburger};
\node (l1) [location] at (1.5,-2) {\faHome};
\node (l2) [location] at (4,0)    {\faHome};

\draw (q0) to node[midway, above left] {5} (p);
\draw (q0) to node[midway, below left] {7} (b);
\draw (p)  to node[midway, above] {3} (c);
\draw (p)  to node[midway, left] {3} (b);
\draw (b)  to node[midway, below] {1} (c);
\draw (b)  to node[midway, below right] {2} (l1);
\draw (c)  to node[midway, above right] {3} (l2);
\draw (l2) to node[midway, right] {1} (l1);
\draw (l1) to[bend left=45] node[midway, below right] {3} (q0);

\end{tikzpicture}}
\caption{City map with eateries and other locations. Icons denote eateries (\faPizzaSlice, \faHamburger, \faCoffee), driver's initial location (\faMotorcycle), and customer locations (\faHome). Edges represent shortest travel times.} %
\label{fig:icon-delivery-map}
\end{figure}
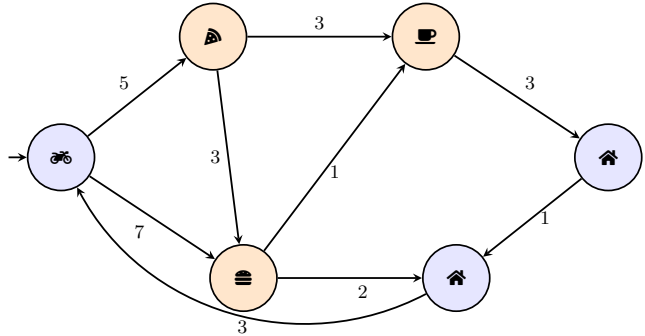

\end{example}

\paragraph{Going from Logics to Automata.}
In this paper, we consider the satisfiability and model-checking problems for a very expressive extension of $\mitl{}$ called \emph{Metric Interval Temporal Logic with Past and Pnueli Modalities} ($\mitlpp{}$), originally proposed in~\citep{rabin, Rabinovich2007} and recently studied in~\citep{concur25}. $\mitlpp{}$ is \emph{expressively complete} for many other proposals of strong decidable real-time logics, including $\qtwomlo{}$~\citep{Hirshfeld2004} and unilateral $\tptl{}$~\citep{concur23}. 
For example, the property `each acknowledgement is 
preceded by a request at most $5$ seconds ago
and followed by a three-step locking process in the next $10$ seconds' can be written using past and Pnueli modalities as 
\[
\globally \big(\textit{ack} \implies \once_{[0, 5]} \textit{req} \land \PnF_{[0, 10]}(\textit{lock}_1, \textit{lock}_2, \textit{lock}_3)\big).
\]
We propose a construction from $\mitlpp{}$ formulae into language-equivalent \emph{timed automata} ($\ta{s}$) over \emph{timed words}, which can be analysed by a number of existing $\ta{}$-based tools, including \textsc{LTSmin}~\citep{KanLaa15}, \textsc{Uppaal}~\citep{behrmann2007uppaal} and \textsc{TChecker}~\citep{TChecker}.\footnote{In the literature, this is commonly referred to as the \emph{pointwise} or \emph{event-based} semantics~\citep{AD94, ouaknine2007decidability} where behaviours are timed words (time-stamped event sequences).}
The fully compositional construction is based on translating subformulae into  \emph{tester automata}~\citep{kesten1998algorithmic, pnueli2008merits}, which can be combined into a single monolithic $\ta{}$ with the standard product construction, if required.
In contrast to incomplete methods like~\citep{bersani2016tool, bae2019bounded}, automata-theoretic approaches to satisfiability and model checking of (linear-time) temporal logic specifications (also used by tools like LTL2BA~\citep{gastin-oddoux} and Spot~\citep{duret2022spot} for $\ltl{}$) guarantees soundness and completeness. While handling timed logics, we also need to take care of time intervals in modalities. We discuss approaches for these next.

\paragraph{General time intervals.}
 In $\mitl{}$, modalities (temporal operators) decorated with \emph{unilateral} time intervals (i.e.~time intervals of the form $\langle l, \infty)$ or $[0, u \rangle$) can be translated into tester automata without too much difficulty~\citep{raskin1997state, geilen2003improved}. However, 
  for practical applications, general time intervals such as $\langle l, u \rangle$ are very handy (see e.g.,~\citep{bae2019bounded}). 
  The main challenge is that the automata constructions for general intervals are much more complicated and difficult to implement, due to the intricate nature of handling a potentially unbounded number of \emph{obligations} (assertions about the future or the past).
   For example, if an automaton \emph{guesses} that $\eventually_{[10, 20]} q$ is satisfied at points $t_1$ and $t_2$ where $t_1 < t_2$, then the corresponding \emph{obligations} are that there must be $q$-events
 in both $[t_1 + 10, t_1 + 20]$ and $[t_2 + 10, t_2 + 20]$.
  From a theoretical perspective, there is also a strict expressiveness gap between general interval modalities and unilateral interval modalities---%
  $\mitl{}$ modalities with general time intervals  
 cannot be expressed in $\mitl{}$ using only modalities with $\langle l, \infty)$ and $[0, u \rangle$ over timed words~\citep{icalp-raskin, raskin-thesis}. %
For this reason, specific constructions for modalities with $\langle l, u \rangle$ were proposed in~\citep{brihaye2013mitl, brihaye2014mitl} and later improved and  implemented in the tool $\mightyl{}$~\citep{DBLP:conf/cav/BrihayeGHM17}.
The main technical idea of $\mightyl{}$ is to maintain an abstraction of \emph{overlapping} obligations using only a bounded number of clocks.
In the previous example if (say) $5 < t_2 - t_1 < 6$, then the  automaton chooses, non-deterministically, that the two obligations will be satisfied by either (i) a single $q$-event in $[t_2 + 10, t_1 + 20]$; or (ii) two individual $q$-events in $[t_1 + 10, t_2 + 10)$ and $(t_1 + 20, t_2 + 20]$. 
Representing this abstraction, however, leads to an exponential blow-up in the discrete state space. Experimental results suggest that more than a minimal use of general time intervals is infeasible with $\mightyl{}$.

\paragraph{Obligations and sequentialisation.} 
In this paper, we take a cleaner and more modular approach to handle $\mitl{}$ modalities with $\langle l, u \rangle$. 
Rather than encoding how overlapping obligations will be satisfied, we keep track of a bounded set of \emph{non-overlapping} obligations. 
More specifically, inspired by a surprising connection between Pnueli modalities with unilateral time intervals and $\mitl{}$ modalities with general time intervals $\langle l, u \rangle$, we propose a better abstraction where (i) obligations are merged so they do not overlap with others, and (ii) all obligations are \emph{of the same form} and will be satisfied in the same manner. 
This enables us to break down the task further and use a dedicated (smaller) component automaton for each obligation, instead of manipulating all obligations in a single monolithic tester automaton.
While delegating the validation of each obligation to a component automaton simplifies the construction considerably---for example, we do not need operations like shifting or renaming of clocks~\citep{maler2005real, akshay2024mitl}---naively composing these automata via their product can cause a state-space explosion (exponential in the number of components).
Inspired by similar ideas in the context of concurrent software verification~\citep{lal2009reducing, fischer2013cseq, chaki2011time}, we \emph{sequentialise} the 
obligations, 
i.e.~enforce that all incoming obligations are handled by these component automata in a specific order.
This ensures that the
size (number of reachable locations) of \emph{the product of all the component automata} (which, apart from some additional atomic propositions, is language-equivalent to a monolithic tester automaton for the same subformula) is 
\emph{polynomial} in the magnitude of the constants $\langle l, u \rangle$.

\paragraph{Implementation.}
 We have implemented the described construction in the tool $\mightyppl{}$, which supports multiple output formats and features a symbolic back-end.  It can be used as a stand-alone tool or in a tool chain to verify 
real-time systems (modeled as $\ta{s}$) against $\mitlpp$ specifications.
This is the very first tool, to the best of our knowledge, for satisfiability and model checking of $\mitl{}$ extended with past and Pnueli modalities. 
In addition to the optimisations found in $\mightyl{}$ and a more modular architecture with broader back-end compatibility, $\mightyppl{}$ incorporates the following implementation techniques for better efficiency: \begin{inparaenum} 
\item Instead of representing the truth values of atomic propositions explicitly, which leads to an exponential blow-up in the number of transitions or locations, $\mightyppl{}$ uses \emph{symbolic values} to synchronise transitions labelled by Boolean formulae.
\item The product of tester automata can be constructed directly by the tool, generating a single monolithic timed automaton for the entire formula; to minimise the number of locations, only the forward- and backward-reachable parts of the state space are constructed.
\end{inparaenum}   
We provide a comprehensive experimental evaluation of $\mightyppl{}$ and show that it achieves significant performance gains over $\mightyl{}$, in some cases by more than two orders of magnitude.

\paragraph{Related work.} 
The idea of adding timing constraints into temporal logics dates back to the early 1990s ($\mtl{}$~\citep{koymans} and $\tctl{}$~\citep{alur1993model}).
Early research mainly focusses on \emph{discrete-time} settings, as the standard verification problems in more natural and general \emph{dense-time} settings are mostly undecidable, due to the fact that `punctual' constraints can be used to encode computations of Turing machines~\citep{AH93} (a notable exception is the decidability of finite-word satisfiability for the future fragment of $\mtl{}$~\citep{ouaknine2007decidability}). 
$\mitl{}$ is the most notable fragment of $\mtl{}$ that is decidable in a dense-time setting, based on the idea of prohibiting the punctual time intervals. 
But as we mentioned earlier, there are natural `counting' properties that are not expressible in $\mitl{}$.
Also, as opposed to $\ltl{}$, $\mitl{}$ with past modalities is strictly more expressive than the future fragment of $\mitl{}$~\citep{bouyer2010expressiveness}.
$\qtwomlo{}$~\citep{Hirshfeld2004} and $\mitl{}$ with counting modalities~\citep{hunter} support these features, but as far as we are aware, only incomplete approaches based on bounded model checking are available for these logics (e.g.,~\citep{bersani2016tool}). The logic $\mitlpp{}$ that we consider in this paper subsumes these logics.
Another orthogonal development is to extend $\mitl{}$ with predicates over real-valued variables, obtaining \emph{Signal Temporal Logic} ($\stl{}$~\citep{maler2004monitoring}), which has gained considerable interest in the past decade. 
$\stl{}$ is often paired with a quantitative or robust interpretation that measures the degree of satisfaction~\citep{fainekos2006robustness,donze2013efficient}, making it well-suited for runtime monitoring and falsification.

Another parallel line of development in the modelling and verification of real-time systems is based on the formal analysis of timed automata. Zone-based abstractions (see~\citep{bouyer2022zone} for a recent survey) have been implemented in practical tools from the early 1990s, most notably \textsc{Uppaal}~\citep{Behrmann2006}, which has become the de-facto standard and as a result, its model file format is also supported by many other tools such as \textsc{LTSmin} (which supports model checking of timed automata against LTL formulae, and also multi-core model checking). More recently, \textsc{TChecker}~\citep{TChecker}
has emerged as a modern and flexible open-source alternative for the analysis of timed automata, based on advanced zone abstraction techniques. 
Compared with tools for model-checking (untimed) reactive systems like SPIN~\citep{DBLP:journals/tse/Holzmann97} or NuSMV~\citep{cimatti2002nusmv}, which supports $\ltl{}$, most $\ta{}$-based tools only support very limited real-time specifications. For example, \textsc{Uppaal} supports only a limited fragment of $\tctl{}$.

The successful idea of automata-theoretic model checking~\citep{DBLP:conf/banff/Vardi95} has led to robust tools and algorithms for $\ltl{}$. To adopt the idea to timed settings, we need a reliable approach that translates $\mitl{}$ into timed automata. The earliest translation from $\mitl$ to $\ta{s}$ 
is described in~\citep{AFH96}. Since it is very involved, there have been other proposals as well, e.g.,~\citep{maler2006mitl, nivckovic2010mtl, d2021clock}.
But these constructions are based on the \emph{continuous} or \emph{state-based} semantics (where behaviours are \emph{signals}), which is not directly compatible with the back-ends such as \textsc{Uppaal} or \textsc{TChecker}.
Moreover, to the best of our knowledge, these constructions have never been implemented. 
$\mightyl{}$~\citep{DBLP:conf/cav/BrihayeGHM17} is the only implementation of logic-to-$\ta{}$ construction that we are aware of. It is also based on timed words, but it only supports the future fragment of $\mitl{}$,  encodes the truth assignments to atomic propositions explicitly, 
and as we have mentioned earlier there is an exponential blow-up in the discrete state space when a  general time interval is used,
which severely limits its practical performance when it is used with tools like Uppaal and LTSmin.
 A more recent proposal is~\citep{akshay2024mitl}, which is similar in spirit to~\citep{maler2006mitl} but formulated in the framework of \emph{generalized timed automata} over timed words~\citep{akshay2023unified}. This construction, however, makes heavy use of the machinery of \emph{future clocks} and thus not compatible with standard $\ta{}$-based tools. 
Our tool $\mightyppl{}$ outputs standard `vanilla' $\ta{s}$~\citep{AD94} that work with all the $\ta{}$-based tools seamlessly.

\paragraph{Summary.} Our contributions are as follows: \begin{inparaenum}
\item We propose a fully compositional construction from $\mitlpp{}$, one of the most expressive decidable real-time logics, to timed automata. The construction supports both the future and past variants of $\mitl{}$ modalities and Pnueli modalities in a uniform manner. 
\item We propose a new abstraction for handling obligations in the tester automata for $\mitl{}$ modalities associated with intervals of the form $\langle l, u \rangle$. In particular, the obligations can be handled by a bounded number of \emph{identical} simple component automata, and a novel sequentialisation technique is applied to achieve an exponential improvement in the number of reachable locations over the state-of-the-art approach implemented in \textsc{MightyL}.
\item We present a complete implementation of the discussed construction and present comprehensive emprical results, which shows the efficiency of our approach.
\end{inparaenum}

The rest of this paper is organised as follows.
\cref{sec:prelim} provides necessary background on the syntax and semantics of $\mitlpp{}$.
\cref{sec:past.and.pnueli} explains the compositional construction with a focus on the tester automata construction for past modalities.
\cref{sec:general}  explains the construction for $\mitl{}$ modalities with general $\langle l , r \rangle$ intervals and the gadgets used for sequentialisation.
\cref{sec:exp} presents experimental results over a wide range of benchmarks, including both existing ones from the literature and several new ones we propose.
Finally, \cref{sec:cn} presents conclusions and discusses future directions.

\section{Preliminaries}
\label{sec:prelim}

Let $\R$ and $\N$ respectively represent the set of non-negative reals and naturals (including 0). 
Let $\langle$ denote left open `$($' or left closed `$[$', and  $\rangle$  denote right open `$)$' or right closed `$]$'. Let $\I$ 
denote the set of all intervals $\langle l, u\rangle$ for $l \leq  u$, $l \in \N$, $u \in \N \cup \{\infty\}$,
and $\I_0$ the set of all intervals $[0, u \rangle$ for $0 \leq  u$, $u \in \N$. %
Let $\AP$ be a finite set of atomic propositions, and let $\Sigma_{\AP}=2^{\AP}$ be the finite alphabet that contains all the subsets of $\AP$. 
An infinite (resp.~finite) \emph{timed word} $\rho$ over $\Sigma_{\AP}$ 
is an infinite (resp.~finite) sequence of \emph{events} (pairs of letters and \emph{timestamps}) $\rho=(\sigma_1, \tau_1)(\sigma_2, \tau_2) \dots$ where $\sigma_i \in \Sigma_{\AP}, \tau_i \in \R$, $\tau_1 = 0$, and $\tau_i \le \tau_{i+1}$ for all \emph{positions} $i> 0$. For example, 
$(\{p,q\},0)(\emptyset,1.1)(\{p\},2.1)(\{q\},2.1)$ is a finite timed word over the set of atomic propositions $\AP=\{p,q\}$.
The set of all infinite (resp.~finite) timed words over $\Sigma$ is denoted $T\Sigma^{\omega}$ (resp.~$T\Sigma^{\ast}$).

An infinite timed word is called \emph{Zeno} if
the sequence $(\tau_i)_{i \geq 0}$ converges, and non-Zeno otherwise. We restrict ourselves to non-Zeno infinite timed words (henceforth simply referred to as `timed words'), which is the usual convention, as Zeno words allow infinite actions  within a finite duration, which does not model a natural behaviour.

    \begin{remark}
Since the semantics of the real-time logics we discuss in this paper depend solely on the relative distances between timestamps, it is without loss of generality to assume that every timed word begins with $\tau_1 = 0$. This normalisation, also employed in~\citep{Wilke}, simplifies the presentation.
    \end{remark}

\subsection{Metric Temporal Logic with Past and Pnueli modalities, $\mtlpp$.}
\label{sec:mtl}
Logic $\mtlpp{}$ is an extension of the classical \emph{Metric Temporal Logic} ($\mtl{}$)~\citep{koymans} with past and Pnueli modalities~\citep{Rabinovich}. Formulae of  $\mtlpp$ over a set of atomic propositions $\AP$  are defined  as follows:
\[
\varphi:=\top~\mid~p~\mid~\neg \varphi~\mid~\varphi_1 \wedge \varphi_2~\mid~\varphi_1 \until_I \varphi_2~\mid~\varphi_1 \since_I \varphi_2
\mid \PnF_J(\varphi_1, \ldots, \varphi_k) \mid \PnO_J(\varphi_1, \ldots, \varphi_k)
\]
where $p \in \AP$, $I$ is an interval in $\I$, and $J$ is an interval in $\I_0$.
The other Boolean operators are defined as usual:
$\bot \equiv \neg\top$,
$\varphi_1\lor\varphi_2 \equiv \neg(\neg\varphi_1\land\neg\varphi_2)$,
and $(\varphi_1 \implies \varphi_2) \equiv \lnot\varphi_1 \lor \varphi_2$.
Given a timed word $\rho=(\sigma_1,\tau_1)(\sigma_2,\tau_2)\dots$ over $\Sigma_{\AP}$ and a \emph{position} $i \in \N_{>0}$,  
we define the \emph{pointwise semantics} of $\mtlpp$ formulae inductively 
as follows:
\begin{itemize}
 \item $\rho, i  \models  \top$; \\
 \item $\rho, i  \models  p$ \text{iff } $p \in \sigma_i$; \\
 \item $\rho, i  \models  \neg \varphi ~  \text{iff }  \rho, i \not\models \varphi$; \\
\item  $\rho, i  \models  \varphi_1 \wedge \varphi_2 ~ \text{iff }  \rho, i \models \varphi_1 \text{ and } \rho, i \models \varphi_2$; \\
  \item $\rho, i  \models  \varphi_1 \until_I \varphi_2 $ iff   $\exists j > i$  s.t.\  $\tau_j - \tau_i \in I,$ $\rho, j \models \varphi_2,$  and 
                 $\forall i < k < j, \rho, k  \models \varphi_1$; \\
      \item $\rho, i  \models  \varphi_1 \since_I \varphi_2 $ iff    $\exists  j < i$  s.t.\  $\tau_i - \tau_j \in I$, $\rho, j \models \varphi_2,$ and 
                 $\forall j < k < i, \rho, k  \models \varphi_1$;\\
\item  $\rho, j\models \PnF_J(\varphi_1, \ldots \varphi_k)$ iff $\exists{i_k{>}i_{k-1} {>}{\ldots}{>}i_1{>}j} $
s.t. $\forall 1\le n \le k, {\tau_{i_n}{-}\tau_j{\in}J}$, $\rho, i_n{\models}\varphi_n$; and \\
\item  $\rho, j \models \smash{\PnO_J}(\varphi_1, \ldots \varphi_k)$ 
iff $\exists{i_k{<}i_{k-1} {<}{\ldots}{<}i_1{<}j}$ s.t. $\forall 1\le n \le k, 
{\tau_{j}{-}\tau_{i_n}{\in}J}$, and  $\rho, i_n{\models}\varphi_n$.
\end{itemize}
We define the \emph{timed language} of $\varphi$ as $\sem{\varphi} = \{\rho | \rho, 1 \models \varphi\}$. 
The usual derived operators  $\eventually_I$ (eventually), $\globally_I$ (globally),  $\past_I$ (past),  $\Boxminus_I$ (globally in the past), 
$\nex$ (next) and  $\nm$ (previous) are defined in terms of $\until$ and $\since$ as follows: 
$$\fut_I \varphi \equiv \top \until_I \varphi, \Box_I \varphi \equiv \neg \fut_I \neg \varphi, \past_I \varphi \equiv \top \since_I \varphi, \Boxminus_I \varphi \equiv \neg \past_I \neg \varphi, \nex \varphi \equiv \bot \until_{[0,\infty)} \varphi, \nm \varphi \equiv \bot \since_{[0,\infty)} \varphi.$$ 
We also define some additional dual operators as follows. 
\begin{IEEEeqnarray*}{rClrCl}
\varphi_1 \release_I \varphi_2 & \equiv & \neg ((\neg \varphi_1) \until_I (\neg \varphi_2)), &
\varphi_1 \trigger_I \varphi_2 & \equiv & \neg ((\neg \varphi_1) \since_I (\neg \varphi_2)), \\ 
\dualPnF_{J}(\varphi_1, \varphi_2, \dots ) & \equiv & \neg \PnF_J( \neg \varphi_1, \neg 
\varphi_2,\dots), \quad & \dualPnO_{J}(\varphi_1, \varphi_2, \dots ) & \equiv & \neg \PnO_J( \neg \varphi_1, \neg 
\varphi_2,\dots)
\end{IEEEeqnarray*}
 We omit the subscript when the intervals are $[0,\infty)$. For examples, $\varphi \until_{[0,\infty)} \psi$ and $\varphi \since_{[0,\infty)} \psi$ are written as $\varphi \until \psi$ and $\varphi \since \psi$, respectively.

The subclass \emph{Metric Interval Temporal Logic with Past and Pnueli modalities}, written $\mitlpp$, consists of all $\mtlpp$ formulae where all intervals
$I$ are non-singular (i.e., of the form $\langle l, u \rangle$, where $l<u$). 
\emph{Metric Temporal Logic} ($\mtl{}$) \citep{koymans} and \emph{Metric Interval Temporal Logic} ($\mitl{}$) \citep{AFH96} 
can be seen as fragments of $\mtlpp$ (resp., $\mitlpp$) without Pnueli modalities $\PnF_J$ and $\PnO_J$. 
The \emph{unilateral} fragment of all these logics consist of the fragment in which
all intervals $I$ are either of the form $[0, u\rangle$ or $\langle l, \infty)$.

\begin{example}
To illustrate  Pnueli modalities, we have $\rho,1 \models \PnF_{[0,2)}(p,q,r)$ for
\[
\rho{=}(\{p\},0)(\{p\},0.5)(\emptyset, 0.9)(\{q,r\},1.1)(\{p,q,r\},1.8)\dots
\] 
since  $\tau_5-\tau_1, \tau_4-\tau_1, \tau_2-\tau_1  \in (0,2)$ and 
$\rho, 2 \models p$, $\rho,4 \models q$ and $\rho, 5 \models r$. However, $\rho', 1 \not \models \varphi$ 
for
\[
\rho'{=}(\{p\},0)(\{r\},0.1)(\{q\},1.1)(\{p,q\},1.9)(\emptyset, 2)\dots \;.
\]
\end{example}
\begin{remark}
    Notice that we work with the `strict' semantics for the temporal operators which is more expressive than the `non-strict' semantics.  In the non-strict semantics, $\varphi_1 \until_I \varphi_2$, when asserted at a position $i$ of $\rho$, checks if  there exists $j \geq i$ 
    where $\varphi_2$ holds, and in case $j > i$, checks whether $\varphi_1$ holds at all the intermediate positions $i \leq k < j$. 
    The strict choice simplifies the presentation, as strict $\until$ and $\since$ can express the next ($\nex$) and previous ($\nm$) operator as shown above, unlike their non-strict versions. 
\end{remark}

\subsection{Timed Automata ($\ta{}$)}
We give a concise definition of timed automata, focussing on \emph{generalised B\"uchi acceptance} for technical convenience; one can reduce such automata to 
classical timed B\"uchi automata~\citep{AD94} via a  standard construction~\citep{Courcoubetis1992}.
Let $X$ be a finite set of \emph{clocks} (variables taking values from $\R$).
A \emph{valuation} $\nu$ for $X$ maps each clock $x \in X$ to a value in $\R$.
We denote by $\mathbf{0}$ the valuation that maps every clock to $0$.
The set $\Guards(X)$ of \emph{clock constraints} $g$ over $X$ is generated
by the grammar $$g:= \top\mid g\land g \mid x\bowtie c, ~\text{where}~
{\bowtie}\in \{{\leq},{<},{\geq},{>}\},~ x\in X,~\text{and}~ c\in\N$$

We write $x \in g$ if $x \bowtie c$ appears as a conjunct of $g$, and in this case $g(x)$ for the interval that corresponds to the clock constraints on $x$.
The satisfaction relation $\nu \models g$ is
defined in the usual way. For instance, for a valuation $\nu$ where $\nu(x)=1.1, \nu(y)=2$, $g = (x < 2 \wedge y = 2)$ and $g' = (x \geq 5 \wedge x < 6)$, we have $x, y \in g$, $g(x) = [0, 2)$, $g(y) = [2, 2]$, $\nu \models g$. Likewise, $y, z \notin g'$, $g'(x) = [5, 6)$, and $\nu \not \models g'$.
For $t\in\R$, we let $\nu +t$ be the valuation defined by $(\nu +t)(x) = \nu (x)+t$ for all $x\in X$. For $\lambda \subseteq X$, we let $v [\lambda \leftarrow 0]$ be the valuation defined by $(\nu[\lambda \leftarrow 0])(x) = 0$ if $x\in \lambda$, and
$(\nu[\lambda \leftarrow 0])(x) = \nu(x)$ otherwise.
A \emph{timed automaton} ($\ta{}$) over a finite alphabet $\Sigma$ is a tuple
$\mathcal{A} = \langle \Sigma, S, s_0, X, \Transitions, \F \rangle$ where $S$ is a finite set of
locations, $s_0 \in S$ is the initial location, $X$ is a finite set of clocks,
$\Transitions \subseteq S \times \Sigma \times \Guards(X) \times
2^X \times S$ is the transition relation,
and $\F=\{F_1,\dots,F_n\}$, with $F_i\subseteq S$ for all $i$, $1\leq i \leq n$,
is a \emph{generalised B\"uchi acceptance condition}, i.e.~a set of sets of final locations.

A \emph{state} of $\mathcal{A}$ is a pair $(s, \nu)$ consisting of a location $s \in S$, and a valuation $\nu$ for $X$.
A \emph{run} $r$ of $\mathcal{A}$ on a timed word $(\sigma_1,\tau_1)(\sigma_2,\tau_2)\cdots\in T\Sigma^\omega$ is an alternating
sequence of states and transitions
\[
r = (s_0,\nu_0) \xrightarrow{(s_0,\sigma_{1},g_1,\lambda_1,s_{1})}(s_1,\nu_1)
\xrightarrow{(s_1,\sigma_{2},g_2,\lambda_2,s_{2})}
\cdots
\]
where 
\begin{inparaenum}
\item[(i)]
$\nu_0=\mathbf{0}$, and 
\item[(ii)] for each $i\geq 0$,
there is a transition $(s_i,\sigma_{i+1},g_{i+1},\lambda_{i+1},s_{i+1}) \in \Delta$
such that %
$\nu_i +(\tau_{i+1}-\tau_i)\models g_{i+1}$ (let $\tau_0=0$) and
$\nu_{i+1} =(\nu_i +(\tau_{i+1}-\tau_i))[\lambda_{i+1} \leftarrow 0]$ (note that since $\tau_1 = 0$, we necessarily have $\nu_1 = \mathbf{0}$). 
\end{inparaenum}
For each position $i > 0$ of $r$, we write $r(i) = (s_i, \nu_i)$. 
 A run of $\mathcal{A}$ is \emph{accepting} iff the set of locations it visits infinitely often contains at least
one location from each $F_i$, $1\leq i\leq n$. A timed word is \emph{accepted} by $\mathcal{A}$ iff $\mathcal{A}$ has an accepting run on it.
We denote by $\sem{\mathcal{A}}$ (the timed language of $\mathcal{A}$) the set of all timed words accepted by $\mathcal{A}$. 
The class of languages accepted by $\ta{s}$ is called \emph{timed regular 
languages}. 
We also define finite-word acceptance in the usual way (where $\mathcal{F}$ is a singleton $\{ F \}$, and a run is accepting if it ends up in a location in $F$), and accordingly the class of \emph{finite-word timed regular languages}.

For two $\ta{s}$, $\mathcal{A}^1 = \langle \Sigma,S^1,s_0^1,X^1,\Transitions^1,\F^1 \rangle$ and
$\mathcal{A}^2 = \langle \Sigma,S^2,s_0^2,X^2,\Transitions^2,\F^2 \rangle$ over a
common alphabet $\Sigma$, the synchronous product 
$\mathcal{A}^1 \times \mathcal{A}^2$ is defined as the $\ta{}$ $\langle \Sigma,S,s_0,X,\Transitions,\F \rangle$
where 
\begin{enumerate}
\item[(i)]  $S=S^1 \times S^2$, $s_0 = (s_0^1, s_0^2)$, and $X = X^1 \cup X^2$;
\item[(ii)]  $((s^1_1,s^2_1),\sigma,g,\lambda,(s^1_2,s^2_2))\in\Transitions$
iff there exists $(s^1_1,\sigma,g^1,\lambda^1,s^1_2)\in\Transitions^1$
and $(s^2_1,\sigma,g^2,\lambda^2,s^2_2)\in\Transitions^2$ such that
$g=g^1\land g^2$ and $\lambda = \lambda^1 \cup \lambda^2$; and
\item[(iii)]  let $\F^1={\{F_1^1,\dots,F_n^1\}}$, $\F^2={\{F_1^2,\dots,F_m^2\}}$, then
$\F = \{F_1^1 \times S^2,\dots,F_n^1\times S^2, S^1 \times
F_1^2,\dots,S^1\times F_{m}^2\}$.
\end{enumerate}
Note that we have $\sem{\mathcal{A}^1 \times \mathcal{A}^2} {=} \sem{\mathcal{A}^1} \cap \sem{\mathcal{A}^2}$.
This generalises to the product of more than two $\ta{s}$.

\subsection{Compositional Translation from Temporal Logics to Automata}
Temporal logics such as $\ltl$ ($\mitl{}$) provide a concise and declarative way to specify (timing) requirements. 
However, for algorithmic verification tasks like model checking or satisfiability, it is often more effective to work with operational models such as (timed) automata. This motivates translations from logical specifications to automata-based representations, enabling the use of well-established automata-theoretic techniques and tools such as \textsc{Uppaal}~\citep{behrmann2007uppaal}, \textsc{TChecker}~\citep{TChecker}, etc. 
 We now briefly recall the methodology developed in the state of the art translation from timed logics to timed automata  
  \citep{DBLP:conf/cav/BrihayeGHM17} where,  future $\mitl{}$ (with only $\until_I$ modalities) is translated  to a network of $\ta{s}$. 
This is reminiscent of the \emph{stratification} method~\citep{manna1989completing, burch1992symbolic, clarke1994another, kesten1998algorithmic} for untimed temporal logics, also pioneered by Pnueli.

\paragraph{Triggers and Tester Automata.}\label{para:triggers}

Given a formula $\varphi$ over a set of atomic propositions $\AP$, we assume without loss of generality that $\varphi$ is written in \emph{negation normal form}---that is, negation is applied only to atomic propositions. In the case where $\varphi$ is a future $\mitl{}$ formula, this means that only the Boolean connectives $\wedge$, $\vee$ and future temporal modalities $\eventually_I$, $\globally_I$, $\until_I$, and $\release_I$ are used.

Let $\Psi$ denote the set of temporal subformulae of $\varphi$ whose outermost operator is a temporal modality $\triangledown_I$—specifically, either $\globally_I$, $\eventually_I$, $\until_I$, or $\release_I$. For instance,  if  $\varphi=\Big(p \until_{[0, 2)}\big((q \until_{[3, 5]} r) \wedge s\big)\Big)$, then $\Psi = \{\varphi, \kappa\}$ where $\kappa=(q \until_{[3, 5]} r)$.
For each such temporal subformula $\psi \in \Psi$, we introduce a fresh atomic proposition $p_\psi$, referred to as the \emph{trigger} for $\psi$, and let $\AP_\Psi = \{ p_\psi \mid \psi \in \Psi \}$. In this case, $\AP_{\Psi}=\{p_{\varphi}, p_{\kappa}\}$. 
For every subformula $\phi$ of $\varphi$, define $\overline{\phi}$ as the formula obtained by replacing each of its \emph{top-level} temporal subformulae $\psi$ (i.e.~$\psi$ is not a strict subformula of $\psi'$ where $\psi' \in \Psi$) with its corresponding trigger $p_\psi$. For example,  $\overline{\varphi} = p_\varphi$, $\overline{\big((q \until_{[3, 5]} r) \wedge s \big)} = p_\kappa \land s$, and $\overline{\kappa} = p_\kappa$.
Intuitively, this operation abstracts $\phi$ into its \emph{propositional skeleton} $\overline{\phi}$.
Now we construct a new formula $\varphi'$ over the extended atomic propositions $\AP \cup \AP_\Psi$ that is \emph{equi-satisfiable} with the original formula $\varphi$, i.e.~$\varphi$ is satisfiable if and only if $\varphi'$ is satisfiable. The formula $\varphi'$ is defined as the conjunction of:
\begin{itemize}
  \item the propositional skeleton of $\varphi$, i.e. $\overline{\varphi}$, and
  \item one formula $\globally (p_\psi \implies \triangledown_I(\overline{\phi_1}, \dots, \overline{\phi_n}))$ for each $\psi = \triangledown_I(\phi_1, \dots, \phi_n) \in \Psi$.
\end{itemize}
In our example $\varphi'$ is $p_{\varphi} \wedge \globally(p_\varphi \implies p \until_{[0, 2)} (p_{\kappa} \wedge s)) \wedge  \globally(p_{\kappa} \implies q \until_{[3, 5]} r)$.
In general,
\[
\varphi' := \overline{\varphi} \land \bigwedge_{\{ \psi = \triangledown_I(\phi_1, \dots, \phi_n) \in \Psi \}} \globally (p_\psi \implies \triangledown_I(\overline{\phi_1}, \dots, \overline{\phi_n})) \;.
\]
For each $\globally(p_\psi \implies \triangledown_I(\dots))$, we construct a corresponding \emph{tester $\ta{}$} $\mathcal{C}_\psi$ that accepts exactly those behaviors satisfying the constraint. That is,
\[
\sem{\mathcal{C}_\psi} = \sem{\globally(p_\psi \implies \triangledown_I(\overline{\phi_1}, \dots, \overline{\phi_n}))} \;.
\]
 Intuitively, these are called tester automata because $\mathcal{C}_\psi$ tests at every position where $p_\psi$ holds, whether the subformula $\triangledown_I(\dots)$ is satisfied or not.\footnote{In the terminology of~\citep{pnueli2008merits}, these are called \emph{positive} testers.} 
  In addition, we construct a simple automaton $\mathcal{C}_{\overline{\varphi}}$ which enforces that the purely propositional formula $\overline{\varphi}$ holds at the initial time instant. Finally, the automaton for $\varphi'$ is obtained as the synchronous product:
$\mathcal{C}_{\varphi'} := \mathcal{C}_{\overline{\varphi}} \times (\times_{\psi \in \Psi} \mathcal{C}_\psi)$, 
which accepts the same language as $\sem{\varphi'}$, and hence is equisatisfiable to the original formula $\varphi$.

\section{Past $\mitl{}$ and Pnueli Modalities}
\label{sec:past.and.pnueli}
In this section, 
we take the approach that \emph{past and future are reflections of each other} and extend the construction of \citep{DBLP:conf/cav/BrihayeGHM17} to handle past and Pnueli modalities. 
In fact, it is possible to accommodate any temporal operator $\triangledown_I$, as long as for each subformula $\triangledown_I(\phi_1, \dots, \phi_n)$ of the entire formula $\varphi$ (assumed to be in negation normal form), we can construct a corresponding tester $\ta{}$ $\mathcal{C}_\psi$ such that $\sem{\mathcal{C}_\psi} = \sem{\globally (p_\psi \implies \triangledown_I(\overline{\phi_1}, \dots, \overline{\phi_n}))}$. In what follows, we focus on tester $\ta{s}$ for past $\mitl{}$ modalities and Pnueli modalities.

\subsection{Tester Automata for Past $\mitl{}$ Modalities with Unilateral Intervals}

Intuitively, tester $\ta{s}$ for past $\mitl{}$ modalities just need to `remember' what happened up until the current point and `output' accordingly; for past $\mitl{}$ modalities with \emph{unilateral} intervals this can be simple (see, e.g.,~\citep{ferrere2019real}). For instance, for $p \since_{[0, 2]} q$, the tester $\ta{}$ resets a clock $x$ whenever it reads a $q$-event, and $p \since_{[0, 2]} q$ will continue to hold until either $\neg p$ holds or $x$ exceeds $2$, whichever happens earlier. In order to deal with more sophisticated past modalities in a principled manner, here we make use of a folklore result in the theory of timed automata, namely that the \emph{class of (finite-word) timed regular languages are closed under reversal}.
 This result, essentially based on the simple idea of `reversing the arrows' in the corresponding $\ta{s}$, enables us to construct the tester $\ta{s}$ for past modalities in a principled manner.
Formally, given a finite timed word $\rho=(\sigma_1, \tau_1)(\sigma_2, \tau_2) \dots (\sigma_{n-1}, \tau_{n-1})(\sigma_n, \tau_n)$ where $\tau_1=0$, 
the \emph{reverse} of $\rho$, denoted $\rho^R$, is the finite timed  word $\rho^R=(\sigma_n,0)(\sigma_{n-1}, \tau_n-\tau_{n-1}) \dots 
(\sigma_2, \tau_n-\tau_2)(\sigma_1, \tau_n-\tau_1)$. Given a finite-word timed language $L \subseteq T\Sigma^*$,
we write $L^R=\{\sigma^R \mid \sigma \in L\}$ for the reverse finite-word timed language of $L$.

 \smallskip 
 Let $\varphi$ be any $\mitlpp{}$ formula, and $\Psi$ be a set of all its temporal subformulae.
Consider a temporal subformula $\psi$ of the form $\oset[-1pt]{\leftarrow}{\triangledown}_I\!\!(\phi_1, \ldots, \phi_n)$, where $\oset[-1pt]{\leftarrow}{\triangledown}_I$ denotes a past modality such as $\past_I$, $\Boxminus_I$, $\PnO_I$, $\dualPnO_I$, $\since_I$, or $\trigger_I$. Let $(\sigma_1', \tau_1), (\sigma_2', \tau_2), \dots, (\sigma_i', \tau_i), \dots$ be an infinite timed word over the alphabet $\Sigma_{\AP \cup \AP_{\Psi}}$. %
As $\oset[-1pt]{\leftarrow}{\triangledown}_I$ is a past modality, if $\oset[-1pt]{\leftarrow}{\triangledown}_I\!\!(\overline{\phi_1}, \dots, \overline{\phi_n})$ holds at position $i$, then $\oset[-1pt]{\rightarrow}{\triangledown}_I\!\!(\overline{\phi_1}, \dots, \overline{\phi_n})$, where
$\oset[-1pt]{\rightarrow}{\triangledown}_I$ is the \emph{finite-word} future
counterpart of $\oset[-1pt]{\leftarrow}{\triangledown}_I$, 
must be satisfied by 
the finite timed word over $\Sigma_{\AP \cup \AP_{\Psi}}$ obtained by reversing
$(\sigma_1', \tau_1) \dots (\sigma_i', \tau_i)$,
i.e.~the prefix read thus far, and vice versa.
In other words, the
tester $\ta{}$ for $\oset[-1pt]{\leftarrow}{\triangledown}_I\!\!(\overline{\phi_1}, \dots, \overline{\phi_n})$
can be obtained by `reversing' a $\ta{}$
that accepts the finite-word timed language of 
$\globally (p_\psi \implies \oset[-1pt]{\rightarrow}{\triangledown}_I\!\!(\overline{\phi_1}, \dots, \overline{\phi_n}))$.\footnote{For the past modalities $\oset[-1pt]{\leftarrow}{\triangledown}_I$ considered in this paper, the tester $\ta{s}$ for $\oset[-1pt]{\rightarrow}{\triangledown}_I$
(the \emph{finite-word} future
counterpart of $\oset[-1pt]{\leftarrow}{\triangledown}_I$)
can be obtained from the infinite-word ones with some trivial modifications in acceptance conditions.}
More generally, we have the following lemma.

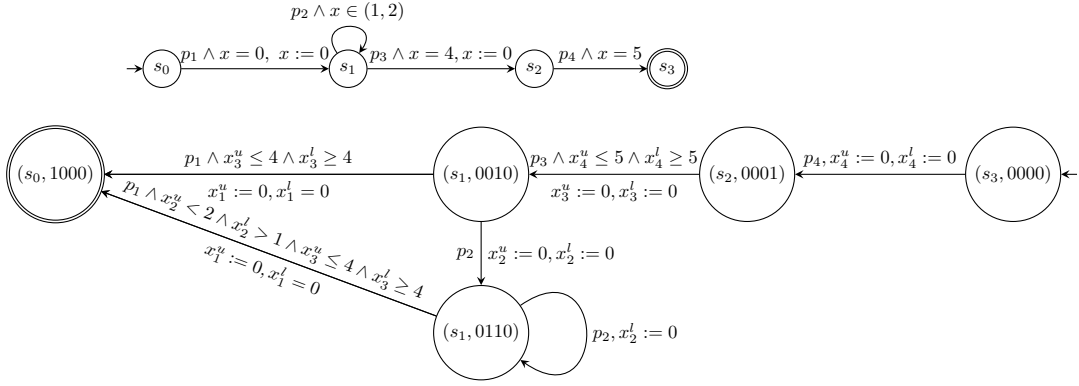
\begin{figure}[!htbp]
\centering
  \begin{tikzpicture}[->, transform shape, scale=0.7, node distance=4cm]
   \node (q0) [state,  initial, initial where=left, initial distance=0.3cm, initial text={}] at (-5,0) {$s_0$};
    \node (q1) [state] at (-1.5,0) {$s_1$};
    \node (q2) [state] at (2,0) {$s_2$};
    \node (q3) [state, accepting] at (4.5,0) {$s_3$};
    \draw   (q0) edge[] node[above]{$p_1 \land x=0,~ x:=0$}  (q1);
    \draw   (q1) edge[loopabove] node[above]{$p_2 \land x \in (1,2)$}  (q1);
    \draw   (q1) edge[] node[above]{$p_3 \land x=4, x:=0$}  (q2);
\draw   (q2) edge[] node[above]{$p_4 \land x=5$}  (q3);

    \node (q0) [state, accepting] at (-7,-2) {$(s_0,1000)$};
    \node (q1) [state] at (1,-2) {$(s_1,0010)$};
    \node (q11) [state] at (1,-5) {$(s_1,0110)$};
        \node (q2) [state] at (6,-2) {$(s_2,0001)$};
    \node (q3) [state, initial, initial where=right, initial distance=0.3cm, initial text={}] at (11,-2) {$(s_3,0000)$};
    \draw   (q1) edge[] node[above, sloped]{$p_1 \land x^{u}_3\le 4 \land x^{l}_3 \ge 4$} (q0);
    
    \draw   (q1) edge[] node[below, sloped]{$x^u_1 := 0, x^l_1 = 0$} (q0);
    
    \draw   (q11) edge[] node[above, sloped]{$p_1 \land x^{u}_2 < 2 \land x^{l}_2>1 \land x^{u}_3\le 4 \land x^{l}_3 \ge 4$} (q0);
    
    \draw   (q11) edge[] node[below, sloped]{$x^u_1 := 0, x^l_1 = 0$} (q0);
    
    \draw   (q11) edge[loopright] node[right]{$p_2, x^l_2:=0$}  (q11);
    \draw   (q1) edge[] node[left]{{$p_2$}} node[right]{{$x^u_2 := 0, x^{l}_2 := 0$}}  (q11);
    
    \draw (q2) edge[]
  node[above] {$p_3 \land x^u_{4} \leq 5 \land x^l_4 \geq 5$}
  node[below] {$x^u_3 := 0,x^{l}_3 := 0$}
  (q1);
\draw   (q3) edge[] node[above]{$p_4, x^{u}_4:=0, x^l_4:=0$}  (q2);
 \end{tikzpicture}
  \caption{A $\ta{}$ with finite-word acceptance condition (above) and the corresponding `reverse' $\ta{}$ (below). The $\ta{}$ above has a single clock $x$ and four transitions (numbered $1$ to $4$ from left to right), thus %
  we have the clocks $x_i^u, x_i^{\ell}$ for $1 \leq i \leq 4$ and each location has a bit-vector of length $4$.}
  \label{eg-reverse}
\end{figure}

\begin{lemma}\label{lem:reversal}
Finite-word timed regular languages are closed under reversal.
\end{lemma}

\begin{proof}
 
Given a $\ta{}$ $\mathcal{A} = \langle \Sigma, S, s_0, X, \Transitions, \F \rangle$ where $\F = \{ F \}$ with $F \subseteq S$ (without loss of generality, we assume that every transition from $s_0$ resets all clocks and $|F| = 1$), 
we construct a $\ta{}$ $\mathcal{A}^\textit{R} = \langle \Sigma, S^\textit{R}, s_0^\textit{R}, X^\textit{R}, \Transitions^\textit{R}, \F^\textit{R} \rangle$
such that $\sem{\mathcal{A}^\textit{R}} = \sem{\mathcal{A}}^\textit{R}$, where 
$\sem{\mathcal{A}}^\textit{R}$ is the reverse finite-word timed language of 
$\sem{\mathcal{A}}$. 
The key idea is to swap the roles of the initial and final locations
and remember 
which clock constraints have been `hit' in the discrete state space. Specifically, each location in $S^\textit{R}$ is of the form $(s, \mathbf{b})$ where $s \in S$ and $\mathbf{b}$ is a zero-initialised two-dimensional bit-array indexed by $x \in X$ and $\delta \in \Transitions$.
That is, for each $x \in X$ and each transition $\delta \in \Delta$, we have a bit corresponding to the pair $(x,\delta)$. 

For each $\delta = (s, \sigma, g, \lambda, t) \in \Transitions$, we have
$2^{\lvert X \rvert \cdot \lvert \Delta \rvert}$ reversed transitions $\delta^\textit{R}_\mathbf{b} = ((t, \mathbf{b}), \sigma, g^R, \lambda^R, (s, \mathbf{b'})) \in \Transitions^\textit{R}$ 
(one for each possible $\mathbf{b}$)
where each $\delta^\textit{R}_\mathbf{b}$ sets $\mathbf{b'}[x, \delta] = 1$ for all $x \in X$ with 
$x \in g$.
For each $x \in X$ with $x \in g$, 
we introduce 
in 
$X^\textit{R}$ 
two new clocks
$x_\delta^\textit{u}$, $x_\delta^\textit{l}$
(for the upper- and lower-bound of $g(x)$, respectively):
$x_\delta^\textit{u}$ is reset when
one of the $\delta^\textit{R}_\mathbf{b}$'s is first taken,
and $x_\delta^\textit{l}$ is reset
every time when one of the $\delta^\textit{R}_\mathbf{b}$'s 
is taken. 
Finally, for each $x \in X$ with $x \in \lambda$, $\delta^R_\mathbf{b}$ checks the `accumulated' clock constraint indicated by $\mathbf{b}$, sets $\mathbf{b'}[x, \delta'] = 0$ for all $\delta' \in \Transitions$ where $\delta' \neq \delta$, and also sets $\mathbf{b'}[x, \delta] = 0$ if $x \notin g$. Having defined the reverse TA, before providing the formal proof that this accepts the reverse language, we look at an example. 
\smallskip 

\noindent{\bf{On an Example}}.
See \cref{eg-reverse} for an illustrative example. The given $\ta{}$ has a single clock $x$ and 4 transitions. Therefore, $\mathbf{b}$
is a row vector of dimension $1 \times 4$.  Let us number the transitions in $\ta{}$ as follows : $s_0$ to $s_1$ as 1, 
the loop on $s_1$ as 2, the transition from $s_1$ to $s_2$ as 3, and the transition from $s_2$ to $s_3$ as 4.  
The initial state of the reversed automaton is $(s_3,0000)$ where $s_3$ is the final state of the 
$\ta{}$ we started with, and 0000 signifies that no transition has fired yet.
We reverse the transition (numbered 4) from $s_2$ to $s_3$ involving guard $x=5$  as follows. 
We add the transition $(s_3,0000)$ to $(s_2,0001)$ resetting $x_4^u$ and $x_4^{l}$. 
The fourth bit  1 in $(s_2,0001)$ signifies that the original transition (numbered 4) 
from $s_2$ to $s_3$ involved a guard on $x$. The idea is that 
a transition outgoing from $(s_2,0001)$ will check a guard on  $x_4^u$ and $x_4^{l}$
if the corresponding actual transition entering $s_2$ in $\ta{}$ had a reset on $x$. 
In our case, the $\ta{}$ has a reset on $x$ on the transition entering $s_2$
from $s_1$, while it has the guard $x=5$ on the transition from
$s_2$ to $s_3$. When moving out from $(s_2,0001)$ to $(s_1,0010)$, we check  $x_4^u \leq 5$ and $x_4^{l} \geq 5$; recall that these clocks were  reset 
on entering $(s_2,0001)$. This results in the guards $x_4^{u} \leq 5 \wedge x_4^{l} \geq 5$
decorating the transition from $(s_2,0001)$ to $(s_1, 0010)$. The third bit  1
in $(s_1, 0010)$ signifies the guard $x=4$ in the 3rd transition (from $s_1$ to $s_2$)  
in the original $\ta{}$. The clocks $x_3^u$ and $x_3^{l}$ 
are reset while going from $(s_2,0001)$ to $(s_1, 0010)$; the transition from $(s_1,0010)$ to $(s_0,1000)$ checks 
$x_3^{u} \leq 4 \wedge x_3^{l} \geq 4$ since $x$ was reset on the transition from $s_0$ to $s_1$ in $\ta{}$.

Note that guards are checked only on the reverse transitions where there was a clock reset 
in the original transition; likewise, $x_{\delta}^l, x_{\delta}^u$ are reset on a reverse transition when the original transition numbered $\delta$ 
has a guard on $x$.  
For instance, the loop on $s_1$ in $\ta{}$ does not reset $x$, but has a guard on $x$. Correspondingly, 
in the reversed $\ta{}$, the transition from $(s_1, 0010)$ to $(s_1, 0110)$ has no guards but resets $x_2^{u}$ and $x_2^{l}$. 
The second and third bits 1 in the target $(s_1, 0110)$ represent guards on $x$ in the second and third transitions in $\ta{}$ outgoing from $s_1$;
notice that correspondingly we have the two clock resets (on $x_3^u, x_3^l$ and  $x_2^u, x_2^l$) on which time is accumulating. 
Likewise, the loop at $(s_1, 0110)$  has no guards, and only resets $x_2^{l}$. 
From $(s_1, 0110)$, the transition to $(s_0, 1000)$ checks guards on both $x_2^u, x_2^l$ and $x_3^u, x_3^l$, since 
the original transition (numbered 1 from $s_0$ to $s_1$) has a reset on $x$. Likewise, the transition from 
$(s_1, 0010)$  to $(s_0, 1000)$ has guards on $x_3^u, x_3^l$.
Notice the resets of  $x_1^u, x_1^l$ on both transitions entering $(s_0,1000)$
corresponding to the guard on $x$ on the transition (numbered 1) from $s_0$ to $s_1$.

We now prove that $\sem{\mathcal{A}}^\textit{R}$ is the reverse finite-word timed language of 
$\sem{\mathcal{A}}$. 
Consider
a finite timed word $\rho = (\sigma_1,\tau_1)(\sigma_2,\tau_2)\cdots(\sigma_n, \tau_n) \in \sem{\mathcal{A}}$ and
a finite accepting run 
\[
r = (s_0,\nu_0) \xrightarrow{\delta_1 = (s_0,\sigma_{1},g_1,\lambda_1,s_{1})}(s_1,\nu_1)
\xrightarrow{\cdots}
\cdots
\xrightarrow{\cdots}
(s_{n-1},\nu_{n-1}) \xrightarrow{\delta_n = (s_{n-1},\sigma_{n},g_n,\lambda_n,s_{n})}(s_n,\nu_n)
\]
of $\mathcal{A}$ on $\rho$
where $s_n \in F$ (note that $x \in \lambda_1$ for every $x \in X$, by assumption).
Given the run $r$, for each $x \in X$ and $i \in \{1, \dots, n\}$, we define 
\[
\textsf{LastReset}^r_x(i) =
\begin{cases}
0 & \text{if } i = 1 \,, \\
j_\textit{max} & \text{if } i > 1 \text{ and } j_\textit{max} \text{ is the largest } j, 1 \leq j < i \text{ with } x \in \lambda_j \,.
\end{cases}
\]
Intuitively, at each position $i$ of the run $r$, and for each clock $x$, $\textsf{LastReset}^r_x(i)$ is the latest position $j < i $ where $x$ has been reset.
By the definition of $r$, for any $x \in X$ and $i \in \{1, \dots, n\}$ with $x \in g_i$ and $\textsf{LastReset}^r_x(i) = j$,
we have $\tau_i - \tau_j \in g_i(x)$ (let $\tau_0 = 0$).
Now for each $x \in X$ and $j \in \{1, \dots, n\}$, we define
\[
\textsf{AllUses}^r_x(j) =
\begin{cases}
\emptyset & \text{if } x \notin \lambda_j \,, \\
\{\, i \mid j < i \leq n \text{ and } x \in g_i \text{ and }  
\textsf{LastReset}^r_x(i) = j
\,\} & \text{if } x \in \lambda_j \,.
\end{cases}
\]
$\textsf{AllUses}^r_x(j)$  captures all positions $i > j$ of the run $r$ where the clock $x$ is used in the guard 
$g_i$, provided $x$ was last reset at $j < i$. That is, $x \in g_i$, $x$ is reset at position $j$, and 
for all $j < k < i$, $x$ has not been reset.

Now consider $\rho^R=(\sigma_n,0)(\sigma_{n-1}, \tau_n-\tau_{n-1}) \dots 
(\sigma_2, \tau_n-\tau_2)(\sigma_1, \tau_n-\tau_1)$. We argue that there is a finite accepting run
\[
r^\textit{R} = ((s_n,\mathbf{b_0}), \mu_0) \xrightarrow{\delta^R_{n, \mathbf{b_0}}} ((s_{n-1}, \mathbf{b_1}),\mu_1)
\xrightarrow{\cdots}
\cdots
\xrightarrow{\cdots}
((s_{1}, \mathbf{b_{n-1}}),\mu_{n-1}) \xrightarrow{\delta^R_{1, \mathbf{b_{n-1}}}}((s_0, \mathbf{b_n}),\mu_n)
\]
of $\mathcal{A}^R$ on $\rho^R$ 
where $\mathbf{b_0}$ contains all zeroes and $\mu_0$ maps all clocks in $X^R$ to $0$.
In fact, it is clear from the construction of $\mathcal{A}^R$ that $\mathbf{b_1}, \dots, \mathbf{b_n}$ and $\mu_1, \dots, \mu_n$ are uniquely determined by $\delta^R_{n, \mathbf{b_0}}, \dots \delta^R_{1, \mathbf{b_{n-1}}}$, and thus we only need to show that $\mu_{n-j} + (\tau_{j+1} - \tau_{j}) \models g^R_{j}$ (the clock constraint on $\delta^R_{j, \mathbf{b_{n-j}}}$) for each $j \in \{1, \dots, n - 1 \}$  (by construction $g^R_n = \top$).
If, for example, $x^u_\delta \leq b \land x^l_\delta \geq a$
appears in $g^R_j$
for some $x \in X$ and $\delta = (s, \sigma, g, \lambda, t) \in \Transitions$ (i.e.~$g(x) = [a, b]$),
then the following conditions hold:
\begin{itemize}
    \item $x \in \lambda_j$ (the set of clocks reset on $\delta_{j}$); %
    \item $\mathbf{b_{n - j}}[x, \delta] = 1$;
    \item There is at least one $i \in \textsf{AllUses}^r_x(j)$ such that $\delta_i = \delta$.
    \item For any $i \in \textsf{AllUses}^r_x(j)$ such that $\delta_i = \delta$, we have $\tau_i - \tau_j \in [a, b]$.
\end{itemize}
Let 
$k = \max \{\, i \in \textsf{AllUses}^r_x(j) \mid  \delta_i = \delta \,\}$ 
and $m = \min \{\, i \in \textsf{AllUses}^r_x(j) \mid  \delta_i = \delta \,\}$. By construction, $(\mu_{n-j} + (\tau_{j+1} - \tau_{j}))(x^u_\delta) = \tau_k - \tau_j$ and
$(\mu_{n-j} + (\tau_{j+1} - \tau_{j}))(x^l_\delta) = \tau_m - \tau_j$. It follows that $\mu_{n-j} + (\tau_{j+1} - \tau_{j}) \models x^u_\delta \leq b \land x^l_\delta \geq a$, and the same argument applies to any $x \in X$ and $\delta \in \Delta$. For the other direction, assume that there is a finite accepting run $r^R$ (as above) of $\mathcal{A}^R$ on $\rho^R$. The claim holds by noting that 
$\tau_k - \tau_j \in [a, b]$ and $\tau_m - \tau_j \in [a, b]$ implies $\tau_i - \tau_j \in [a, b]$
for any $i \in \{\, i \in \textsf{AllUses}^r_x(j) \mid  \delta_i = \delta \,\}$.
 \end{proof}

\begin{remark}
    From~\citep{Wilke}, we know that the class of  
 timed regular languages (i.e.~accepted by $\ta{s}$) is precisely characterised by a monadic second-order logic $\mathcal{L}\!\!\oset[-1pt]{\leftrightarrow}{d}$ with relative distance formulae of the form $\oset[-1pt]{\leftarrow}{d}\!\!(X, x) \sim c$ and
 $\oset[-1pt]{\rightarrow}{d}\!\!(x, X) \sim c$, where
 $X$ can only appear in the outermost existential second-order quantifiers. This result carries over to the case of \emph{finite-word} timed regular languages and suggests a way to `reverse' a given $\ta{}$ with finite-word acceptance: we first compute the equivalent $\mathcal{L}\!\!\oset[-1pt]{\leftrightarrow}{d}$ formula $\vartheta$. Then, we obtain another $\mathcal{L}\!\!\oset[-1pt]{\leftrightarrow}{d}$ formula 
 $\vartheta^{\textit{R}}$
 for  the reverse language by inverting all the order predicates and relative distance formulae in $\vartheta$, e.g., $x < y$ becomes $y < x$ and 
 $\oset[-1pt]{\leftarrow}{d}\!\!(X, x) < c$ becomes 
 $\oset[-1pt]{\rightarrow}{d}\!\!(x, X) < c$.
 Finally, we convert $\vartheta^\textit{R}$ back into a $\ta{}$.
 The last step, however, involves removing all the `future' relative distance formulae $\oset[-1pt]{\rightarrow}{d}\!\!(x, X) \sim c$ (see details in~\citep{Wilke}) and may potentially result in a non-elementary blow-up.  Our construction above operates directly on $\ta{s}$ to avoid this blow-up; similar ideas have been used in~\citep{AlurH92, alur1999event}.

\end{remark}
For a given $\ta{}$, the lemma above yields a reverse $\ta{}$ with a larger number of locations and more clocks in general. 
For our purpose, however, with a bit more care we can construct tester $\ta{s}$ of roughly the same sizes as their future counterparts, sometimes even smaller due to the simpler acceptance conditions. For example, 
in~\cref{fig:until_l_infty} we have a tester $\ta{}$ for 
  $\phi_1\until_{\geq l} \phi_2$, assuming the finite-word semantics (a similar $\ta{}$ can be found in~\citep{DBLP:conf/cav/BrihayeGHM17}, but the one here is based on the more expressive strict semantics of $\until_I$). The running obligation\footnote{Recall that an obligation is an assertion which must hold in the future / past.} in this example is ``after $x$ reaches $l$, there will be a point where $\phi_2$ holds''.  
  The idea of the tester $\ta{}$ is to reset the clock $x$ and move from $s_0$ to $s_1$, when the trigger $p_\psi$ is first set to $\top$. While the $\ta{}$ is in $s_1$, the clock $x$ is reset when $p_\psi$ is set to $\top$---updating the existing obligation---and $\overline{\phi_1}$ must hold continuously, unless the existing obligation is satisfied at the same time ($p_\psi \land \overline{\phi_2} \land x \geq l$).
  Since $[l, \infty)$ is unilateral, the $\ta{}$ only needs to maintain a \emph{single} obligation; $x$ is reset 
each time $p_\psi$ is $\top$. 
    The tester $\ta{}$ for 
  $\phi_1\since_{\geq l} \phi_2$ (\cref{fig:since.l.infty}) is exactly the same with the transitions reversed and the clock constraints / resets swapped; in this case, we only have a lower-bound constraint, so a single clock $x$ is sufficient.
  There is no need to use a bit-array to record whether $x$ has been reset before, as it is clear from the graph that $x \geq l$ can only happen after $x := 0$.
  The tester $\ta{s}$ for other types of past $\mitl{}$ modalities
  can be found in~\cref{app:past.components}.

\begin{figure}
\begin{minipage}[t]{0.49\textwidth}
\centering
  \begin{tikzpicture}[->, transform shape, scale=0.7]
    \node[initial left, state, accepting](0) {$s_0$};
    \node[state,right=5.5cm of 0] (1) {$s_1$};

    \path[->] (0) edge[above,bend left=10]
    node{$p_\psi$, $x:=0$} (1)%
    (0) edge[loopbelow,below,looseness=20]
    node[align=center]{$\lnot p_\psi$} (0)%
    (1) edge[below,bend left=10]
    node{$\lnot p_\psi \land \overline{\phi_2} \land x\geq l$}
    (0)%
    (1) edge[loopbelow,below,looseness=20]
    node[align=center] {$\lnot p_\psi \land \overline{\phi_1}$\\
            $p_\psi \land \simplify{\phi_1}$, $x:=0$ \\
                $p_\psi\land \simplify{\phi_2}\land x\geq l$, $x:=0$
				} (1); %
  \end{tikzpicture}
\captionof{figure}{The finite-word tester $\ta{}$ for 
  $\phi_1\until_{\geq l} \phi_2$.}
  \label{fig:until_l_infty}
\end{minipage}
\hfill
\begin{minipage}[t]{0.49\textwidth}
\centering
  \begin{tikzpicture}[->, transform shape, scale=0.7]
   \node[initial left,state, accepting](0){$s_0$};
   \node[state, right=5.5cm of 0, accepting](1){$s_1$};
   
   \path
   (0) edge[loopbelow,below,looseness=20] node[align=center]{$\lnot p_\psi$} (0)
   (0) edge[below, bend right=10] node{$\lnot p_\psi \land \overline{\phi_2}$, $x := 0$} (1)
   (1) edge[loopbelow,below,looseness=20] node[align=center]{$\lnot p_\psi \land \overline{\phi_1}$ \\ $p_\psi \land \overline{\phi_1} \land x \geq l$ \\ $p_\psi \land \overline{\phi_2}$, $x \geq l$, $x := 0$} (1)
   (1) edge[above, bend right=10] node{$p_\psi$, $x \geq l$} (0);
 \end{tikzpicture}
\captionof{figure}{The tester $\ta{}$ for 
  $\phi_1\since_{\geq l} \phi_2$.}
\label{fig:since.l.infty}
\end{minipage}
\end{figure}

\subsection{Tester Automata for Pnueli Modalities} 
\label{subsec:tester.for.pnueli}
We now describe how to handle Pnueli modalities and their past counterparts. Recall  that by definition, the intervals in Pnueli modalities are unilateral.
Before we describe the tester  $\ta{}$ for (say) $\PnF_{< u}(\phi_1, \dots, \phi_n)$, let us go through an example. 
\begin{example}
Consider the timed word 
\[
\rho = (\{p_\psi, p_5\}, \tau_1)(\{p_1, p_4\}, \tau_2)(\{p_\psi\}, \tau_3)(\{p_1\}, \tau_4)(\{p_\psi\}, \tau_5)(\emptyset, \tau_6)(\{p_2\}, \tau_7)(\{p_1\}, \tau_8)(\{p_2\}, \tau_9)\dots
\]
A tester $\ta{}$ $\mathcal{C}_\psi$ for $\psi = \PnF_{< 10} (p_1, p_2)$ is triggered at positions labeled with $p_{\psi}$, namely, 
time points $\tau_1, \tau_3, \tau_5$. Assume  $\tau_7 < \tau_1 + 10 < \tau_9 < \tau_3 + 10$.  
To accept $\rho$, $\mathcal{C}_\psi$ may non-deterministically \emph{merge} the first two obligations (`$p_1, p_2$ occur in order in $[\tau_1, \tau_1 + 10)$'
and `$p_1, p_2$ occur in order in $[\tau_3, \tau_3 + 10)$'),
as they are both witnessed by the events $(\{p_1\}, \tau_4)$ and 
$(\{p_2\}, \tau_7)$; in this case, $\mathcal{C}_\psi$ needs to validate that $\tau_7 < \tau_1 + 10$.
The third obligation (`$p_1, p_2$ occur in order in $[\tau_5, \tau_5 + 10)$') 
must be dealt with separately, as it can only be witnessed by 
$(\{p_1\}, \tau_8)$ and $(\{p_2\}, \tau_9)$.  When this obligation at position 5 expects $p_1, p_2$ in sequence, 
the obligations at positions 1,3 expect $p_2$.
For the obligation at position 5,  $\mathcal{C}_\psi$ needs to validate that $\tau_9 < \tau_5 + 10$.
Alternately, $\mathcal{C}_\psi$ may merge the second and third obligations
and deal with the first separately.
In both cases, $\mathcal{C}_\psi$ needs $2$ clocks. In general, as we show below, a tester $\ta{}$ for $\PnF_{[0, u \rangle} (p_1, \dots, p_n)$ needs $n$ clocks.
\end{example}

 A tester $\ta{}$ for (say) $\PnF_{< u}(\phi_1, \dots, \phi_n)$  
 needs to maintain $n$ types of obligations---the first type of obligation expects $\phi_1, \dots, \phi_n$, the second type of obligation expects $\phi_2, \dots, \phi_n$, and so on. Since $[0, u)$ is unilateral, the tester $\ta{}$ only needs to keep track of \emph{at most one} obligation per \emph{each type}, so at most $n$ obligations in total. %
For instance, if the obligation that expects $\phi_2, \dots, \phi_n$ to hold in a sequence before time $\tau_i + u$ is satisfied, then 
the obligation that expects $\phi_2, \dots, \phi_n$ to hold in sequence before time $\tau_i + u'$ where $u < u'$ must be satisfied as well, i.e.~the second obligation is implied by the first one. 
We formulate the tester $\ta{}$
as the synchronous product of $n$ component $\ta{s}$, each keeps track of an obligation and uses an individual clock. 
Compared with the simpler tester $\ta{s}$ in the previous subsection, there is one additional variation: we split $p_\psi$ into individual triggers,
each handled by one of the component $\ta{s}$.
For example if $n = 4$, we introduce four triggers $p_\psi^1, \dots, p_\psi^4$, and, as an extra optimisation, we  enforce that \emph{exactly one} of $p_\psi^1, \dots, p_\psi^4$ is triggered when $\psi$ is required to hold.   
 A component $\ta{}$ of the tester $\ta{}$ 
for $\PnF_{< u}(\phi_1, \phi_2, \phi_3, \phi_4)$ is
depicted in~\cref{fig:pnueli}.
Component $\ta{s}$ are identical apart from their triggers  $p_\psi^1, \dots, p_\psi^4$, and each of the component $\ta{s}$ may be used to track one of the $n$ types of obligations---in this case, each of $s_0'$, $s_1$, $s_2$, $s_3$ corresponds to a specific type of obligation.
Intuitively, after $p^1_\psi$ is triggered for the first time, we want to wait at $s_0'$ for as long as possible (to capture as many new obligations as possible) before reading $\overline{\phi_1}$, after which the component $\ta{}$ no longer accepts $p^1_\psi$. %
\begin{lemma}\label{lem:pnueli.correctness}
For $\psi = \PnF_{<u}(\phi_1,\dots,\phi_n)$ over $\AP$ and component $\ta{}$s $\mathcal{C}_\psi^1, \dots, \mathcal{C}_\psi^n$
where $\mathcal{C}_\psi^i$ is the $i$-th component $\ta{}$  (over $\Sigma_{\AP \cup \AP_\Psi}$ where $p_\psi^1, \dots, p_\psi^n \in \AP_\Psi$) as discussed above and depicted in~\cref{fig:pnueli},  
\begin{enumerate}
\item $\sem{\mathcal{C}_\psi^1 \times \dots \times \mathcal{C}_\psi^n} \subseteq \sem{\globally(p_\psi^1 \lor \dots \lor p_\psi^n \implies \PnF_{<u}(\overline{\phi_1}, \dots, \overline{\phi_n}))}$.
\item For any $\rho \in \sem{\globally(p_\psi^1 \lor \dots \lor p_\psi^n \implies \PnF_{<u}(\overline{\phi_1}, \dots, \overline{\phi_n}))}$ there is a $\rho' \in \sem{\mathcal{C}_\psi^1 \times \dots \times \mathcal{C}_\psi^n}$
where
\begin{itemize}
\item $\rho, i \models p \iff \rho', i \models p$ for all $i > 0$, $p \notin \{p_\psi^1, \dots, p_\psi^n\}$, $p \in \AP \cup \AP_\Psi$ and
\item $\rho, i \models p_\psi^1 \lor \dots \lor p_\psi^n \iff \rho', i \models p_\psi^1 \lor \dots \lor p_\psi^n$ for all $i > 0$.
\end{itemize}
\end{enumerate}

\end{lemma}
\begin{proof}
(1) is straightforward.
For (2),
consider an infinite timed word $\rho = (\sigma_1,\tau_1)(\sigma_2,\tau_2)\cdots$ such that
\[
\rho, 1 \models
\globally(p_\psi^1 \lor \dots \lor p_\psi^n \implies \PnF_{<u}(\overline{\phi_1}, \dots, \overline{\phi_n})) \;.
\]
Let $\textsf{AllTriggers}(\rho)$ be the set of all the positions of $\rho$ where $p_\psi^1 \lor \dots \lor p_\psi^n$ holds.
For each $i \in \textsf{AllTriggers}(\rho)$, 
let $\textsf{End}(i) = j > i$ be the \emph{minimal} position of $\rho$
such that is a finite accepting run
of the automaton 
\begin{center}
      \begin{tikzpicture}[->, node distance=4cm, transform shape, scale=0.7]
        \node[initial left,state, accepting](0){$s_0$};
        \node[state, right of=0](0'){$s_0'$};
    \node[state, right of=0'](1){$s_1$};
        \node[state, right of=1](2){$s_2$};
        \node[state, right of=2](3){$s_3$};
        \path
        (0') edge[loopabove, ->] node[above=1mm, align=center]{} (0')
        (1) edge[loopabove, ->] node[above=1mm, align=center]{$\lnot \overline{\phi_2}$} (1)
    (2) edge[loopabove, ->] node[above=1mm, align=center]{$\lnot \overline{\phi_3}$} (2)
        (3) edge[loopbelow,below, looseness=20, out=-90, in=-60] node[align=center]{$\lnot \overline{\phi_4}$} (3)
        (3) edge[->, bend left=20] node[below=1mm, align=center]{$\overline{\phi_4}$} (0)
        (0) edge[->] node[above=1mm, align=center]{$p_\psi^1 \lor \dots \lor p_\psi^n$} (0')
        (0') edge[->] node[above=1mm, align=center]{$\overline{\phi_1}$} (1)
        (1) edge[->] node[above=1mm, align=center]{$\overline{\phi_2}$} (2)
        (2) edge[->] node[above=1mm, align=center]{$\overline{\phi_3}$} (3);
      \end{tikzpicture}
\end{center} 
(we take $n = 4$ as an example here)
on the finite timed word
\[
\rho[i, j] = (\sigma_i,0)(\sigma_{i+1},\tau_{i+1}-\tau_i)\cdots(\sigma_j, \tau_j - \tau_i) \;.
\]
Among all the runs on $\rho[i, j]$,
let
\[
r_i = s_0
\rightarrow
s_0'
\rightarrow
\cdots
\rightarrow
s_0 
\]
be a run with the most occurrences of $s_0'$, i.e.~$r_i$ is the one that reaches $s_1$ last among all the runs on $\rho[i, j]$. 
We argue that it is possible to define $\textsf{Allocate} \colon \textsf{AllTriggers}(\rho) \to \{ 1, \dots, n \}$ such that
the infinite timed word
$\rho' = (\sigma_1',\tau_1)(\sigma_2',\tau_2)\cdots$
has an accepting run on $\mathcal{C}_\psi^1 \times \dots \times \mathcal{C}_\psi^n$
where
\begin{itemize}
\item $p^{\textsf{Allocate}(i)}_\psi$ 
holds
 at each $i \in \textup{\textsf{AllTriggers}}(\rho)$, and
 \item None of $p_\psi^m$, $m \neq \textsf{Allocate}(i)$, $m \in \{1, \dots, n\}$ holds at each $i \in \textup{\textsf{AllTriggers}}(\rho)$, and
 \item the truth values of all the other atomic propositions are the same as $\rho$ at all the positions.
\end{itemize} 
  Let $j_1 < j_2 < \dots$ be all the positions of $\rho$ where for each $k > 1$, $j_k = \textsf{End}(i)$ for some $i \in \textup{\textsf{AllTriggers}}(\rho')$. We define $\textsf{Allocate}(i) = k \pmod{n}$ if $\textsf{End}(i) = j_k$. The following two facts together imply our claim that $\rho' \in \sem{\mathcal{C}_\psi^1 \times \dots \times \mathcal{C}_\psi^n}$.
  \begin{itemize}
      \item  If $\textsf{End}(i) = \textsf{End}(i')$ for some $i, i' \in \textup{\textsf{AllTriggers}}(\rho')$ with $i < i'$, then for any $i'' \in \textup{\textsf{AllTriggers}}(\rho')$ with $i < i'' < i'$ we must have $\textsf{End}(i'') = \textsf{End}(i) = \textsf{End}(i')$ (otherwise we have a contradiction with
      the fact that each of $i, i'', i'$ starts a shortest accepting run). 
      \item  For $j_k < j_{k+1} < \dots < j_{k + n}$ (where $k > 0$) and $i_k < i_{k+1} < \dots < i_{k + n}$ (where $i_{k + \ell}$ is the minimal position with $\textsf{End}(i_{k + \ell}) = j_{k + \ell}$ for $\ell \in \{0, \dots, n\}$), we must have $i_{k+n} \geq j_k$. Suppose to the contrary that $i_{k+n} < j_k$.  Then,  since
      \[
      r_{i_k}(1 + i_{k+n} - i_k), r_{i_{k+1}}(1 + i_{k+n} - i_{k+1}), \dots, r_{i_{k+n}}(1)
      \]
      are not all distinct (by the pigeonhole principle), we again have a contradiction with the fact that each 
      $i_k, i_{k+1}, \dots, i_{k + n}$ starts a shortest accepting run. \qedhere
  \end{itemize}
\end{proof}
Like the case of $\mitl{}$ modalities, we can 
obtain the tester $\ta{}$ for 
$\PnO_{< u}(\phi_1,\ldots, \phi_n)$ by simply reversing the arrows and swapping the clock constraints and resets, thanks to the simple structure of the tester $\ta{}$ and the fact that $[0, u)$ is unilateral.
Together with the known tester $\ta{s}$ from~\citep{DBLP:conf/cav/BrihayeGHM17}, %
we have the following lemmas.

\begin{lemma}
\label{lem:uconstruct}
For a subformula of the form $\phi_1 \until_I \phi_2$ with unilateral $I$, we can construct a one-clock tester $\ta{}$ with at most $3$ locations. 
\end{lemma}
\begin{proof}
We consider the following cases.
\begin{itemize}
\item $I = [0, \infty)$: 
The tester $\ta{}$ is similar to~\citep{DBLP:conf/cav/BrihayeGHM17} with minor changes for the strict semantics, and it has an additional third location
to capture the case where $\overline{\varphi_2}$ and $p_\psi$ hold simultaneously
infinitely often.
\item $I = [0, u \rangle$: The tester $\ta{}$ is again similar to~\citep{DBLP:conf/cav/BrihayeGHM17} with minor changes for the strict semantics and an additional third location.
\item $I = \langle l, \infty)$:
The tester $\ta{}$ is
again similar to~\citep{DBLP:conf/cav/BrihayeGHM17},
but we note that under the non-Zeno assumption (which can be enforced with an extra clock, e.g.,~\citep{DBLP:journals/fmsd/TripakisYB05}), we only need $3$ locations. \qedhere
\end{itemize}
\end{proof}

\begin{lemma}
\label{lem:rconstruct}
For a subformula of the form $\phi_1 \release_I \phi_2$ with unilateral $I$, we can construct a one-clock tester $\ta{}$ with at most $2$ locations. 
\end{lemma}
\begin{proof}
Again similar to~\citep{DBLP:conf/cav/BrihayeGHM17} with minor changes for the strict semantics. As the acceptance conditions are trivial (all accepting), in all the cases we need only $2$ locations.
\end{proof}

\begin{lemma}
\label{lem:pastmitl}
For a subformula of the form $\phi_1 \since_I \phi_2$ or $\phi_1 \trigger_I \phi_2$ with unilateral $I$, we can construct a one-clock tester $\ta{}$ with at most $3$ locations. 
\end{lemma}
\begin{proof}
Follows from~\cref{lem:uconstruct,lem:rconstruct},  and~\cref{lem:reversal}  (see~\cref{app:past.components} for the tester $\ta{s}$ for $\phi_1 \since_I \phi_2$ and $\phi_1 \trigger_I \phi_2$).    
\end{proof}

\begin{lemma}
\label{lem:pnconstruct}
For a subformula of the form $\PnF_J(\phi_1,\ldots, \phi_n)$ with $J = [0, u \rangle$, we can construct $n$ one-clock component $\ta{s}$ (of the tester $\ta{}$), each containing at most $n + 2$ locations.
\end{lemma}
\begin{proof}
Follows directly from~\cref{lem:pnueli.correctness}.   
\end{proof}

\begin{lemma}
\label{lem:pndualconstruct}
For a subformula of the form $\dualPnF_J(\phi_1,\ldots, \phi_n)$ with $J = [0, u \rangle$, we can construct $n$ one-clock component $\ta{s}$ (of the tester $\ta{}$), each containing at most $n + 1$ locations.
\end{lemma}
\begin{proof}
Follows from~\citep{ho2025metric} (we use the same component $\ta{s}$).    
\end{proof}

\begin{lemma}
\label{lem:pastpnueli}
For a subformula of the form $\PnO_J(\phi_1,\ldots, \phi_n)$ or $\dualPnO_J(\phi_1,\ldots, \phi_n)$ with $J = [0, u \rangle$, we can construct $n$ one-clock component $\ta{s}$ (of the tester $\ta{}$), each containing at most $n + 1$ locations.
\end{lemma}

\begin{proof}
Follows from~\cref{lem:pnconstruct,lem:pndualconstruct}, and~\cref{lem:reversal}  (see~\cref{app:past.components} for the component $\ta{s}$ for $\dualPnO_J(\phi_1,\ldots, \phi_n)$).    
\end{proof}

Based on the lemmas above and the fact that the product of tester $\ta{s}$, component $\ta{s}$ of tester $\ta{s}$, and a system $\mathcal{M}$ (modelled as a $\ta{}$) can be checked for emptiness on-the-fly in $\pspace{}$~\citep{AD94}, we can state the following theorem.
 
 \begin{theorem}\label{thm:pspace}
The satisfiability and model-checking problems for unilateral $\mitlpp{}$ 
are $\pspace{}$-complete.
 \end{theorem}
In the next section, we will see that this fragment is already expressively complete for full  $\mitlpp{}$.

  \begin{figure}[!htbp]
    \centering
      \begin{tikzpicture}[->, node distance=4cm, transform shape, scale=0.7]
        \node[initial left,state, accepting](0){$s_0$};
        \node[state, right of=0](0'){$s_0'$};
    \node[state, right of=0'](1){$s_1$};
        \node[state, right of=1](2){$s_2$};
        \node[state, right of=2](3){$s_3$};
        \path
        (0) edge[loopbelow, below, looseness=20, out=-120, in=-90] node[align=center]{$\neg p_\psi^1$} (0)
        (0') edge[loopabove, ->] node[above=1mm, align=center]{} (0')
        (1) edge[loopabove, ->] node[above=1mm, align=center]{$\neg p_\psi^1 \land \lnot \overline{\phi_2}$} (1)
    (2) edge[loopabove, ->] node[above=1mm, align=center]{$\neg p_\psi^1 \land \lnot \overline{\phi_3}$} (2)
        (3) edge[loopbelow,below, looseness=20, out=-90, in=-60] node[align=center]{$\neg p_\psi^1 \land \lnot \overline{\phi_4}$} (3)
        (3) edge[->, bend left=20] node[below=1mm, align=center]{$p_\psi^1 \wedge \overline{\phi_4} \land x_1 < u$, $x_1 := 0$} (0')
        (0) edge[->] node[above=1mm, align=center]{$p_\psi^1$, $x_1 := 0$} (0')
        (0') edge[->] node[above=1mm, align=center]{$\neg p_\psi^1 \wedge \overline{\phi_1}$} (1)
        (1) edge[->] node[above=1mm, align=center]{$\neg p_\psi^1 \wedge \overline{\phi_2}$} (2)
        (2) edge[->] node[above=1mm, align=center]{$\neg p_\psi^1 \wedge \overline{\phi_3}$} (3)
        (3) edge[->, bend right=40] node[above=1mm, align=center]{$\neg p_\psi^1 \wedge \overline{\phi_4} \land x_1 < u$} (0);
      \end{tikzpicture}
    \caption{A component $\ta{}$ of the finite-word tester $\ta{}$ for $\PnF_{< u}(\phi_1, \phi_2, \phi_3, \phi_4)$.} %
    \label{fig:pnueli}
  \end{figure}
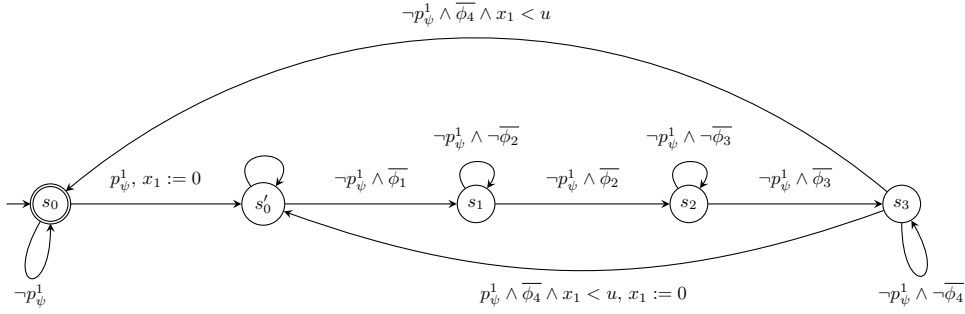
  
  \begin{figure}[!htbp]
    \centering
      \begin{tikzpicture}[->, node distance=4cm, transform shape, scale=0.7]
        \node[initial left,state, accepting](0){$s_0$};
        \node[state, right of=0, accepting](0'){$s_0'$};
    \node[state, right of=0', accepting](1){$s_1$};
        \node[state, right of=1, accepting](2){$s_2$};
        \node[state, right of=2, accepting](3){$s_3$};
        \path
        (0) edge[loopbelow, below, looseness=20, out=-120, in=-90] node[align=center]{$\neg p_\psi^1$} (0)
        (0') edge[loopabove, ->] node[above=1mm, align=center]{$x_1 < u$} (0')
        (1) edge[loopabove, ->] node[above=1mm, align=center]{$\neg p_\psi^1 \land \lnot \overline{\phi_2}$, $x_1 < u$} (1)
    (2) edge[loopabove, ->] node[above=1mm, align=center]{$\neg p_\psi^1 \land \lnot \overline{\phi_3}$, $x_1 < u$} (2)
        (3) edge[loopbelow,below, looseness=20, out=-90, in=-60] node[align=center]{$\neg p_\psi^1 \land \lnot \overline{\phi_4}$, $x_1 < u$} (3)
        (0') edge[->] node[above=1mm, align=center]{$p_\psi^1$, $x_1 < u$} (0)
        (1) edge[->] node[above=1mm, align=center]{$\neg p_\psi^1 \wedge \overline{\phi_1}$, $x_1 < u$} (0')
        (2) edge[->] node[above=1mm, align=center]{$\neg p_\psi^1 \wedge \overline{\phi_2}$, $x_1 < u$} (1)
        (3) edge[->] node[above=1mm, align=center]{$\neg p_\psi^1 \wedge \overline{\phi_3}$, $x_1 < u$} (2)
        (0') edge[->, bend right=20] node[below=1mm, align=center]{$p_\psi^1 \wedge \overline{\phi_4}$, $x_1 < u$, $x_1 := 0$} (3)
        (0) edge[->, bend left=40] node[above=1mm, align=center]{$\neg p_\psi^1 \wedge \overline{\phi_4}$, $x_1 := 0$} (3);
      \end{tikzpicture}
    \caption{A component $\ta{}$ of the tester $\ta{}$ for $\PnO_{< u}(\phi_1,\ldots, \phi_n)$.} %
    \label{fig:pnueli.past}
    
  \end{figure}

\section{General Time Intervals and Sequentialisation}
\label{sec:general}

In this section, we focus on constructing tester $\ta{s}$ for $\mitl{}$ formulae with general time intervals, e.g.,~$\phi_1 \until_{[9, 10]} \phi_2$. The key idea is similar to how we handled Pnueli modalities in the last section: each obligation can be handled by an identical component $\ta{}$, and there is a bound on the number of obligations that can simultaneously exist. It follows that we can formulate tester $\ta{s}$ as products of component $\ta{s}$, which is conceptually much simpler and easier to implement; this is in stark contrast with~\citep{DBLP:conf/cav/BrihayeGHM17, akshay2024mitl} where the tester $\ta{}$ is constructed as a rather sophisticated monolithic $\ta{}$. To avoid unnecessary interleavings of obligations stored on component $\ta{s}$, we also introduce a novel sequentialisation technique which leads to an exponential improvement in the size (number of reachable locations) of the constructed tester $\ta{s}$.

\subsection{Expressing $\mitl{}$ modalities with $\langle l, u \rangle$ in Pnueli modalities}
It is known that, when interpreted over timed words (i.e.~the pointwise semantics is adopted),  $\mitl{}$ modalities with general time intervals $\langle l, u \rangle$ cannot be
expressed in unilateral $\mitl{}$~\citep{raskin1999logic, raskin-thesis}.
This is rather unfortunate from a practical point of view, since if $\mitl{}$ is interpreted over \emph{signals} (i.e.~in the \emph{continuous} or \emph{state-based} semantics) one can use
equalities like $\eventually_{(l + m, u + m)} \varphi \iff 
\eventually_{(0, m)} \globally_{(0, m)} \eventually_{(l, u)} \varphi$ (where $0 < m \leq u-l$)
to rewrite any $\mitl{}$ formula into one that uses only modalities with unilateral intervals (cf.~e.g.,~\citep{icalp-raskin, raskin-thesis}), for which
the tester $\ta{s}$ are fairly simple (as discussed in the last section).\footnote{This approach circumvents the construction of monolithic tester $\ta{s}$ and achieves compositionality, but from a computational complexity point of view, there is no obvious benefit in rewriting a $\mitl{}$ formula into a unilateral one with such rules: it still needs roughly the same number of clocks and introduces an exponential blow-up in the number of locations, like~\citep{DBLP:conf/cav/BrihayeGHM17}.}

We now present equivalence rules (similar to~\citep{H19, concur23} but simpler and \emph{non-inductive}) that 
express 
$\until_{\langle l, u \rangle}$ and 
$\release_{\langle l, u \rangle}$
directly in terms of Pnueli modalities. %

Beyond their practical roles in our implementation, which are detailed later in this section, they also imply that the unilateral fragment of $\mitlpp{}$ (\cref{thm:pspace}) is already expressively complete for full $\mitlpp{}$, without going through $\qtwomlo{}$ as in~\citep{ho2025metric}.
Recall from~\cref{sec:prelim}  that
$\overline{\phi}$ is obtained from $\phi$ by replacing its temporal subformulae by the corresponding triggers. 

\begin{lemma}\label{lem:untilinpnueli}
Let $\phi_1 \until_{(k, k+1)} \phi_2$ be a subformula of the $\mitlpp{}$ formula $\varphi$, and $\Psi$ be the set of all temporal subformulae of $\varphi$. For all timed words $\rho$ over $\Sigma_{\AP \cup \AP_{\Psi}}$, positions $i \geq 1$,  and $\phi^{\geq 1} \equiv \overline{\phi_2} \land (\lnot \overline{\phi_2} \until_{\geq 1} \overline{\phi_2})$, we have
\medskip
\smallskip

$$
\rho, i \models \overline{\phi_1} \until_{(k, k+1)} \overline{\phi_2}  
\Leftrightarrow
 \rho, i \models \overline{\phi_1} \until_{> k} \overline{\phi_2}
 \land
\bigvee_{\ell \in \{0, \dots, k + 1\}} \left (\begin{array}{c}\PnF_{< k + 1}(\underbrace{\phi^{\geq 1}, \dots, \phi^{\geq 1}}_{\ell}, \overline{\phi_2})
\bigwedge
\lnot \PnF_{\leq k}(\underbrace{\phi^{\geq 1}, \dots}_{\ell}, \phi^{\geq 1})\end{array} \right ) \;.$$
\medskip
\smallskip

\end{lemma}
\begin{proof}
We prove the two directions of the implication.
\begin{itemize}
\item[($\Leftarrow$):]
Since $\rho, i \models \overline{\phi_1} \until_{> k} \overline{\phi_2}$, all the positions in $\tau_i + [0, k]$ (but not necessarily $i$) must satisfy $\overline{\phi_1}$. Assume the disjunct holds for some $\ell \in \{0, \dots, k + 1\}$.  Consider the following cases, where the first case covers the case where the disjunct holds by $\ell = 0$. 
\begin{itemize}
    \item There are at least $\ell$ occurrences of $\phi^{\geq 1}$ in $\tau_i + [0, k]$: %
    Let $j > i$ be the last position in $\tau_i + [0, k]$ where $\overline{\phi_2}$ holds.
    There are three possibilities:
    \begin{itemize}
    \item There is no such $j$: This is only possible when $\ell = 0$. Since we have $\rho, i \models \PnF_{< k + 1}(\underbrace{\phi^{\geq 1}, \dots, \phi^{\geq 1}}_{\ell}, \overline{\phi_2})$, it is immediate that
    $\rho, i \models \overline{\phi_1} \until_{(k, k+1)} \overline{\phi_2}$.
    \item $\rho, j \models \phi^{\geq 1}$: Obviously $j$ is also the last position in $\tau_i + [0, k]$ where $\phi^{\geq 1}$ holds, and since
    $\rho, i \models \lnot \PnF_{\leq k}(\underbrace{\phi^{\geq 1}, \dots}_{\ell}, \phi^{\geq 1})$, it must be the case that $\ell > 0$. Since $\rho, i \models \PnF_{< k + 1}(\underbrace{\phi^{\geq 1}, \dots, \phi^{\geq 1}}_{\ell}, \overline{\phi_2})$, the witnessing position $j'$ for  
    $\overline{\phi_2}$ must be greater than $j$, which implies
    that $\rho, i \models \overline{\phi_1} \until_{(k, k+1)} \overline{\phi_2}$. 
    \item $\rho, j \not \models \phi^{\geq 1}$:
    Since $\rho, i \models \overline{\phi_1} \until_{> k} \overline{\phi_2}$, there exists a minimal $j' > j$ such that 
    $\rho, j' \models \overline{\phi_2}$, and it is clear that $j'$
    must be in $\tau_i + (k, k + 1)$ (otherwise we have
    $\rho, j \models \phi^{\geq 1}$, a contradiction).
    \end{itemize}
    \item There are fewer than $\ell$ occurrences of $\phi^{\geq 1}$ in $\tau_i + [0, k]$:
    We must have $\ell > 0$.
    Since $\PnF_{< k + 1}(\underbrace{\phi^{\geq 1}, \dots, \phi^{\geq 1}}_{\ell}, \overline{\phi_2})$, there exists a position $j$ in $\tau_i + (k, k + 1)$ where $\phi^{\geq 1}$, and thus $\overline{\phi_2}$, holds.
    We have $\rho, i \models \overline{\phi_1} \until_{(k, k+1)} \overline{\phi_2}$.

\end{itemize}
\item[($\Rightarrow$):] Let the number of times 
$\phi^{\geq 1}$
is satisfied by the positions in $\tau_i + [0, k]$ (excluding $i$)
be $\ell$. It is clear that $\ell$ is at most $k + 1$, and the 
corresponding disjunct in the RHS holds at $i$. \qedhere 
\end{itemize}
\end{proof}

\begin{lemma}\label{lem:releaseinpnueli}
Let $\phi_1 \release_{(k, k+1)} \phi_2$ be a subformula of the $\mitlpp{}$ formula $\varphi$, and $\Psi$ be the set of all temporal subformulae of $\varphi$. For all timed words $\rho$ over $\Sigma_{\AP \cup \AP_{\Psi}}$, positions $i \geq 1$, and $\phi^{\geq 1} \equiv \lnot \overline{\phi_2} \land ( \overline{\phi_2} \until_{\geq 1} \lnot \overline{\phi_2})$, we have
\medskip
\smallskip
$$
\rho, i \models \overline{\phi_1} \release_{(k, k+1)} \overline{\phi_2}  
\Leftrightarrow
 \rho, i \models \overline{\phi_1} \release_{> k} \overline{\phi_2}
 \lor
\bigvee_{\ell \in \{0, \dots, k + 1\}} \left (\begin{array}{c}\lnot \PnF_{< k + 1}(\underbrace{\phi^{\geq 1}, \dots, \phi^{\geq 1}}_{\ell}, \lnot \overline{\phi_2})
\bigwedge
\PnF_{\leq k}(\underbrace{\phi^{\geq 1}, \dots, \phi^{\geq 1}}_{\ell})\end{array} \right )\;.$$

\medskip
\smallskip
\end{lemma}
\begin{proof}[Proof sketch]
Substituting $\neg \overline{\phi_1}, \neg \overline{\phi_2}$ 
for $\overline{\phi_1}, \overline{\phi_2}$ 
in the formula in~\cref{lem:untilinpnueli}
and taking negation
gives the following formula:
$$
\overline{\phi_1} \release_{> k} \overline{\phi_2}
 \lor
\bigwedge_{\ell \in \{0, \dots, k + 1\}} \left (\begin{array}{c}\lnot \PnF_{< k + 1}(\underbrace{\phi^{\geq 1}, \dots, \phi^{\geq 1}}_{\ell}, \lnot \overline{\phi_2})
\bigvee
\PnF_{\leq k}(\underbrace{\phi^{\geq 1}, \dots}_{\ell}, \phi^{\geq 1})\end{array} \right )$$
where $\phi^{\geq 1} \equiv \lnot \overline{\phi_2} \land ( \overline{\phi_2} \until_{\geq 1} \lnot \overline{\phi_2})$.
Observe that
\begin{itemize}
\item $\lnot \PnF_{< k + 1}(\underbrace{\phi^{\geq 1}, \dots, \phi^{\geq 1}}_{\ell}, \lnot \overline{\phi_2})$ implies
the same formulae with larger $\ell$'s. That is, for any $\rho, i$, if $\rho, i \models \lnot \PnF_{< k + 1}(\underbrace{\phi^{\geq 1}, \dots, \phi^{\geq 1}}_{x}, \lnot \overline{\phi_2})$, then $\rho, i \models \lnot \PnF_{< k + 1}(\underbrace{\phi^{\geq 1}, \dots, \phi^{\geq 1}}_{y}, \lnot \overline{\phi_2})$ for any $y > x$;
\item Symmetrically, $\PnF_{\leq k}(\underbrace{\phi^{\geq 1}, \dots}_{\ell}, \phi^{\geq 1})$ implies the same formulae with smaller $\ell$'s.
\end{itemize}
Also note that 
$\lnot \PnF_{< k + 1}(\underbrace{\phi^{\geq 1}, \dots, \phi^{\geq 1}}_{\ell}, \lnot \overline{\phi_2})$
and $\PnF_{\leq k}(\underbrace{\phi^{\geq 1}, \dots}_{\ell}, \phi^{\geq 1})$
cannot hold simultaneously for the same $\ell$.
It is then easy to adjust the indices to obtain the desired formula.
\end{proof}

We also state the following lemmas for past $\mitl{}$ modalities.

\begin{lemma}\label{lem:since.in.pnueli}
Let $\phi_1 \since_{(k, k+1)} \phi_2$ be a subformula of the $\mitlpp{}$ formula $\varphi$, and $\Psi$ be the set of all temporal subformulae of $\varphi$. For all timed words $\rho$ over $\Sigma_{\AP \cup \AP_{\Psi}}$, positions $i \geq 1$, and $\phi^{\geq 1} \equiv \overline{\phi_2} \land (\lnot \overline{\phi_2} \since_{\geq 1} \overline{\phi_2})$, we have
\medskip
\smallskip
$$
\rho, i \models \overline{\phi_1} \since_{(k, k+1)} \overline{\phi_2}  
\Leftrightarrow
 \rho, i \models \overline{\phi_1} \since_{> k} \overline{\phi_2}
 \land
\bigvee_{\ell \in \{0, \dots, k + 1\}} \left (\begin{array}{c}\PnO_{< k + 1}(\underbrace{\phi^{\geq 1}, \dots, \phi^{\geq 1}}_{\ell}, \overline{\phi_2})
\bigwedge
\lnot \PnO_{\leq k}(\underbrace{\phi^{\geq 1}, \dots}_{\ell}, \phi^{\geq 1})\end{array} \right ) \;.$$
\medskip
\smallskip
\end{lemma}

\begin{lemma}\label{lem:trigger.in.pnueli}
Let $\phi_1 \trigger_{(k, k+1)} \phi_2$ be a subformula of the $\mitlpp{}$ formula $\varphi$, and $\Psi$ be the set of all temporal subformulae of $\varphi$. For all timed words $\rho$ over $\Sigma_{\AP \cup \AP_{\Psi}}$, positions $i \geq 1$, and $\phi^{\geq 1} \equiv \lnot \overline{\phi_2} \land ( \overline{\phi_2} \since_{\geq 1} \lnot \overline{\phi_2})$, we have
\medskip
\smallskip
$$
\rho, i \models \overline{\phi_1} \trigger_{(k, k+1)} \overline{\phi_2}  
\Leftrightarrow
 \rho, i \models \overline{\phi_1} \trigger_{> k} \overline{\phi_2}
 \lor
\bigvee_{\ell \in \{0, \dots, k + 1\}} \left (\begin{array}{c}\lnot \PnO_{< k + 1}(\underbrace{\phi^{\geq 1}, \dots, \phi^{\geq 1}}_{\ell}, \lnot \overline{\phi_2})
\bigwedge
\PnO_{\leq k}(\underbrace{\phi^{\geq 1}, \dots, \phi^{\geq 1}}_{\ell})\end{array} \right ) \;.$$
\medskip\smallskip
\end{lemma}
These lemmas naturally generalise to arbitrary intervals \(I = \langle l, u \rangle\) beyond the specific case \((k, k+1)\). In such cases, the relevant index \(\ell\) belongs to the range \(\{0, \dots, \lfloor \frac{l}{u - l} \rfloor + 1\}\).

\subsection{Tester automata for $\mitl{}$ modalities with $\langle l, u \rangle$}\label{subsec:general}
At this point we may, of course, just use~\cref{lem:untilinpnueli,lem:releaseinpnueli,lem:since.in.pnueli,lem:trigger.in.pnueli}
to rewrite all $\mitl{}$ modalities with $\langle l, u \rangle$
and use the tester $\ta{}$ constructions described in the previous section and~\citep{ho2025metric}. This is not ideal, as there are roughly $2 \cdot \lfloor \frac{l}{u - l} \rfloor$ occurrences of Pnueli or dual Pnueli modalities in each of these lemmas. Each occurrence of Pnueli or dual Pnueli modalities with $n$ arguments needs $n$ component $\ta{s}$ and $n$ clocks, giving $2 \cdot \lfloor \frac{l}{u - l} \rfloor \cdot n$ clocks in total (where $n$ can be as large as $\lfloor \frac{l}{u - l} \rfloor + 1$). 
Significant simplifications can be achieved, however, by working directly with component $\ta{s}$ instead of with formulae.
We now present two such simplifications below.
\paragraph{Reducing the number of component $\ta{s}$ and clocks}
We note the following (take~\cref{lem:untilinpnueli} as an example):
\begin{itemize}
\item The structures of the component $\ta{s}$ for $\PnF_{< k + 1}(\underbrace{\phi^{\geq 1}, \dots, \phi^{\geq 1}}_{\ell}, \overline{\phi_2})$
and $\lnot \PnF_{\leq k}(\underbrace{\phi^{\geq 1}, \dots}_{\ell}, \phi^{\geq 1})$ are similar and the conjunction, in effect, specifies that $\phi^{\geq 1}$ occurs exactly $\ell$ times before $\tau_i + k$.

\item The first $\ell$ arguments are the same, and thus we only need to maintain $k + 2$ obligations (recall that $\ell \in \{0, \dots, k+1\}$) at any time across all the disjuncts.

\end{itemize}
These observations suggest that, while we do not see a way to simplify the formulae used in~\cref{lem:untilinpnueli,lem:releaseinpnueli,lem:since.in.pnueli,lem:trigger.in.pnueli}, we only need $k + 2$ component $\ta{s}$ to `implement' the whole disjunct;
each such component $\ta{}$ uses two clocks, $x$ and $y$ ($x$ is reset when $p_\psi$ first holds; $y$ is reset when $p_\psi$ holds later each time) to keep track of an obligation, which may correspond to any of the disjuncts.

\begin{figure}
\begin{minipage}[t]{0.49\textwidth}
\centering
\begin{tikzpicture}[transform shape, scale=0.55]
\begin{scope}

\draw[-, loosely dashed] (-70pt,0pt) -- (200pt,0pt);

\draw[-, very thick, loosely dotted] (-110pt,0pt) -- (-75pt,0pt);

\draw[loosely dashed] (-120pt,-30pt) -- (-120pt,10pt) node[at start, below=2mm] {\huge $\tau_i$};

\draw[loosely dashed] (-10pt,-30pt) -- (-10pt,10pt) node[at start, below=2mm] {\huge $\tau_i + k-1$};

\draw[loosely dashed] (90pt,-30pt) -- (90pt,10pt) node[at start, below=2mm] {\huge $\tau_i + k$};

\draw[loosely dashed] (79pt,0pt) -- (79pt,25pt);
\draw[loosely dashed] (179pt,0pt) -- (179pt,25pt);

\draw[loosely dashed] (190pt,-30pt) -- (190pt,10pt) node[at start, below=2mm] {\huge $\tau_i + k+1$};

\draw[draw=black, fill=white] (-121pt, -8pt) rectangle (-119pt, 8pt);

\draw[draw=black, fill=white] (-61pt, -8pt) rectangle (-59pt, 8pt);

\draw[draw=black, fill=white] (1pt, -8pt) rectangle (-1pt, 8pt);

\draw[draw=black, fill=blue] (53pt, -8pt) rectangle (51pt, 8pt);

\draw[draw=black, fill=white] (105pt, -8pt) rectangle (103pt, 8pt);

\draw[draw=black, fill=white] (157pt, -8pt) rectangle (155pt, 8pt);

\draw[draw=black, fill=blue] (-21pt, -8pt) rectangle (-23pt, 8pt);

\draw[draw=black, fill=white] (31pt, -8pt) rectangle (29pt, 8pt);

\draw[draw=black, fill=white] (80pt, -8pt) rectangle (78pt, 8pt);

\draw[draw=black, fill=white] (87pt, -8pt) rectangle (85pt, 8pt);

\draw[draw=black, fill=white] (135pt, -8pt) rectangle (133pt, 8pt);

\draw[draw=black, fill=blue] (180pt, -8pt) rectangle (178pt, 8pt);

\end{scope}

\begin{scope}
\draw[|<->|][dotted] (52pt,25pt)  -- (179pt,25pt) node[midway,above] {\huge $\geq 1$};
\end{scope}

\end{tikzpicture}
\captionof{figure}{Case (1). Blue boxes are $\overline{\phi_2}$-events.}
  \label{fig:isolated}
\end{minipage}
\hfill
\begin{minipage}[t]{0.49\textwidth}
\centering
\begin{tikzpicture}[transform shape, scale=0.55]
\begin{scope}

\draw[-, loosely dashed] (-70pt,0pt) -- (200pt,0pt);

\draw[-, very thick, loosely dotted] (-110pt,0pt) -- (-75pt,0pt);

\draw[loosely dashed] (-120pt,-30pt) -- (-120pt,10pt) node[at start, below=2mm] {\huge $\tau_i$};

\draw[loosely dashed] (-10pt,-30pt) -- (-10pt,10pt) node[at start, below=2mm] {\huge $\tau_i + k-1$};

\draw[loosely dashed] (90pt,-30pt) -- (90pt,10pt) node[at start, below=2mm] {\huge $\tau_i + k$};

\draw[loosely dashed] (79pt,0pt) -- (79pt,25pt);
\draw[loosely dashed] (179pt,0pt) -- (179pt,25pt);

\draw[loosely dashed] (190pt,-30pt) -- (190pt,10pt) node[at start, below=2mm] {\huge $\tau_i + k+1$};

\draw[draw=black, fill=white] (-121pt, -8pt) rectangle (-119pt, 8pt);

\draw[draw=black, fill=white] (-55pt, -8pt) rectangle (-53pt, 8pt);

\draw[draw=black, fill=white] (-11pt, -8pt) rectangle (-13pt, 8pt);

\draw[draw=black, fill=blue] (32pt, -8pt) rectangle (30pt, 8pt);

\draw[draw=black, fill=white] (46pt, -8pt) rectangle (44pt, 8pt);

\draw[draw=black, fill=white] (141pt, -8pt) rectangle (139pt, 8pt);

\draw[draw=black, fill=blue] (-21pt, -8pt) rectangle (-23pt, 8pt);

\draw[draw=black, fill=blue] (80pt, -8pt) rectangle (78pt, 8pt);

\draw[draw=black, fill=blue] (135pt, -8pt) rectangle (133pt, 8pt);

\draw[draw=black, fill=blue] (150pt, -8pt) rectangle (148pt, 8pt);

\end{scope}

\begin{scope}
\draw[|<->|][dotted] (-22pt,25pt)  -- (31pt,25pt) node[midway,above] {\huge $< 1$};
\draw[|<->|][dotted] (31pt,25pt)  -- (79pt,25pt) node[midway,above] {\huge $< 1$};
\draw[|<->|][dotted] (79pt,25pt)  -- (134pt,25pt) node[midway,above] {\huge $< 1$};
\end{scope}

\end{tikzpicture}
\captionof{figure}{Case (2). Blue boxes are $\overline{\phi_2}$-events.}
\label{fig:strip}
\end{minipage}
\end{figure}

\paragraph{Simplifying individual component automata.}
The whole purpose of the large disjunction 
in~\cref{lem:untilinpnueli}
is to `locate' the last event where
$\phi^{\geq 1}$ holds in $\tau_i + [0, k]$
and ensure either of the following is true
(a similar observation holds for $\phi_1 \release_{\langle l, u \rangle} \phi_2$  and~\cref{lem:releaseinpnueli}):
\begin{enumerate}
\item The next $\overline{\phi_2}$ (which is $\geq 1$ away) is just in $\tau_i + (k, k + 1)$ (see~\cref{fig:isolated}); %
\item The next $\overline{\phi_2}$ is still in $\tau_i + [0, k]$,
but since $\rho, i \models \lnot \PnF_{\leq k}(\underbrace{\phi^{\geq 1}, \dots}_{\ell}, \phi^{\geq 1})$, it is $< 1$ away from the next $\overline{\phi_2}$ further in the future, which in turn 
is also $< 1$ away from the next $\overline{\phi_2}$ further in the future, and so on;  
eventually there is a $\overline{\phi_2}$ in $(k, k + 1)$ (see~\cref{fig:strip}). %
\end{enumerate}
In fact, using clock constraints on both $x$ and $y$ allows these conditions to be enforced directly. A component $\ta{}$ need not read a sequence of $\phi^{\geq 1}$-events, but can instead guess a single event where $x < k+1$ and $\overline{\phi_2}$ holds.
Then we have two cases, and they correspond to the two cases above:
\begin{enumerate}
\item[(1')] If $y > k$ is also true at this event, then this $\overline{\phi_2}$ satisfies all the obligations tracked by this component $\ta{}$ since $y \leq x$.
\item[(2')] If $y \leq k$ at this event, then this $\overline{\phi_2}$ satisfies the first and  possibly some of the earlier obligations but not all of them, in particular not the latest one. We enforce $\lnot \phi^{\geq 1}$ until a later event where $y > k$ and $\overline{\phi_2}$ both hold; this sequence of $\overline{\phi_2}$-events satisfies all the intermediate obligations, and in particular the last one satisfies the latest obligation.
\end{enumerate}
In both cases, we would have achieved the same effect of validating all the obligations if we simply ignored the sequence of $\phi^{\geq 1}$'s.
This leads to drastically simpler component $\ta{s}$
whose number of locations are independent of 
$\lfloor \frac{l}{u - l} \rfloor$.
A component $\ta{}$ for $\eventually_{(l, u)} \phi$ is depicted in~\cref{fig:simplified}
(the role of $q^1_\psi$ will be explained later); intuitively, it looks for either \begin{inparaenum} \item an event where $\overline{\phi}$ holds; or \item a sequence of events where $\overline{\phi}$ hold, and every two occurrences of $\overline{\phi}$ (with only $\neg \overline{\phi}$ in between) are separated by less than $u - l$. \end{inparaenum}
When $p^1_\psi$ is triggered, the component $\ta{}$ moves to $s_0'$, where it waits before guessing a point where $\overline{\phi}$ and $x < u$ both hold and moves to $s_1$. Between $s_1$ and $s_2$, it enforces that occurrences of $\overline{\phi}$ are separated by less than $u - l$,
until heading back to $s_0$ or $s_0'$, again on reading $\overline{\phi}$.
Finally, we notice that there can only be a bounded number (no more than $\lfloor \frac{u}{u - l} \rfloor$) of pairs of events that both satisfy $\overline{\phi}$ but are more than $u - l$ apart in $\tau_i + [0, u)$.
The tester $\ta{}$ for $\phi_1 \until_{(l, u)} \phi_2$ is the product of these component $\ta{s}$ and the tester $\ta{}$ for 
$\phi_1 \until_{> k} \phi_2$. 

\begin{figure}
\centering
\begin{tikzpicture}[->, node distance=7.5cm, transform shape, scale=0.6]
   \node[initial left,state, accepting](0){$s_0$};
   \node[state, right of=0](1){$s_0'$};
   \node[state, right of=1](2){$s_1$};
   \node[state, right of=2](3){$s_2$};
   
   \path
   (0) edge[loopabove] node[above=1mm, align=center, looseness=20, out=120, in=90,xshift=-5mm]{$\neg p^1_\psi \land \neg q^{1}_\psi$, $x := 0, y := 0$} (0)
   (0) edge[->, bend left=10] node[above=1mm, align=center,xshift=2mm]{$p^1_\psi \land \neg q^{1}_\psi$, $x := 0, y := 0$} (1)
   (1) edge[loopabove] node[above=1mm,align=center]{$\neg p^1_\psi \land \neg q^{1}_\psi$, $x < u, y \leq l$} (1)
   (1) edge[loopbelow,looseness=10] node[below=1mm, align=center,xshift=-8mm]{$p^1_\psi \land \neg q^{1}_\psi$, $x < u, y \leq l$, $y := 0$ \\ $p^1_\psi \land q^{1}_\psi \land \overline{\phi}$, $x < u, y > l$, $x := 0, y := 0$} (1)
   (1) edge[->, bend left=10] node[below=1mm, align=center]{$\neg p^1_\psi \land q^{1}_\psi \land \overline{\phi}$, $x < u, y > l$, $x := 0, y := 0$} (0)
   
   (1) edge[->] node[above=1mm, align=center,xshift=-3mm]{$\neg p^1_\psi \land q^{1}_\psi \land \overline{\phi}$, $x < u, y \leq l$, $x := 0$} (2)
   
   (2) edge[->, bend left=10] node[above=1mm, align=center,xshift=-4mm,yshift=-2mm]{$\neg p^1_\psi \land \neg q^{1}_\psi \land \neg \overline{\phi}$, $x < (u - l), y \leq l$} (3)
   
   (3) edge[->, bend left=10] node[below=1mm, align=center,xshift=-5mm,yshift=1mm]{$\neg p^1_\psi \land q^{1}_\psi \land \overline{\phi}$, $x < (u - l), y \leq l$, $x := 0$} (2)
   
   (2) edge[loopabove,looseness=10] node[above=1mm,align=center,xshift=10mm,yshift=0mm]{$\neg p^1_\psi \land q^{1}_\psi \land \overline{\phi}$, $x < (u - l), y \leq l$, $x := 0$} (2)
   
   (3) edge[loopabove] node[above=1mm,align=center,xshift=12mm]{$\neg p^1_\psi \land \neg q^{1}_\psi \land \neg \overline{\phi}$, $x < (u - l), y \leq l$} (3)

   (2) edge[->, bend left = 30] node[below right=1mm and 2mm,align=center]{$p^1_\psi \land q^{1}_\psi \land \overline{\phi}$, $x < (u - l), y > l$, $x := 0, y := 0$} (1)
   
   (3) edge[->, bend left = 40] node[below=1mm,align=center]{$p^1_\psi \land q^{1}_\psi \land \overline{\phi}$, $x < (u - l), y > l$, $x := 0, y := 0$} (1)

   (2) edge[->, bend right = 30] node[above=1mm,align=center, xshift=5mm]{$\neg p^1_\psi \land q^{1}_\psi \land \overline{\phi}$, $x < (u - l), y > l$, $x := 0, y := 0$} (0)
   
   (3) edge[->, bend right = 40] node[above=1mm,align=center]{$\neg p^1_\psi \land q^{1}_\psi \land \overline{\phi}$, $x < (u - l), y > l$, $x := 0, y := 0$} (0)
   
   ;
 \end{tikzpicture}
\caption{A simplified component $\ta{}$ of the tester $\ta{}$ for $\eventually_{(l, u)} \phi$.}
\label{fig:simplified}
\end{figure}

\subsection{Sequentialisation}
The simplified tester $\ta{}$ construction described above is correct by~\cref{lem:reversal,lem:pnueli.correctness,lem:uconstruct,lem:rconstruct,lem:pastmitl,lem:pnconstruct,lem:pndualconstruct,lem:pastpnueli}.
 However, it allows for arbitrary interleavings of obligations, which can cause an exponential blow-up in the number of locations (the same is also true for the tester $\ta{}$ construction for Pnueli modalities in~\cref{subsec:tester.for.pnueli}). For example, if there are $n$ component $\ta{s}$ and $n$ obligations, then there are $n!$ ways to allocate these obligations to the component $\ta{s}$. 
To ensure that the product of $n$ component $\ta{s}$  only has a polynomial (in $n$) number of reachable locations, we \emph{sequentialise} the obligations  by introducing new atomic propositions; our approach is inspired by~\citep{lal2009reducing, fischer2013cseq, chaki2011time} and shares conceptual similarities with \emph{partial-order methods}~\citep{valmari1990stubborn, godefroid1990using, holzmann1992coverage, godefroid1994partial, valmari1993fly, peled1993all}.
For example, for the $n$ component $\ta{s}$ of the tester $\ta{}$ for some $\phi_1 \until_{\langle l, u \rangle} \phi_2$, in addition to the triggers $p_\psi^{1}, \dots, p_\psi^{n}$ we also introduce $q_\psi^{1}, \dots, q_\psi^{n}$.
We label ~$q^i_\psi$, on all the transitions   that are also labelled with $\overline{\phi}$ (see~\cref{fig:simplified}) in the $i^\textit{th}$ component $\ta{}$. 
Then we use two additional (untimed) component automata $\mathcal{A}^\textit{in}$ and $\mathcal{A}^\textit{out}$ (see~\cref{fig:seq.in,fig:seq.out} where $n = 4$, $\textit{in}^\textit{null} = \neg (p^1_\psi \lor p^2_\psi \lor p^3_\psi \lor p^4_\psi)$, and $\textit{out}^\textit{null} = \neg (q^1_\psi \lor q^2_\psi \lor q^3_\psi \lor q^4_\psi)$) to ensure that \begin{inparaenum} \item the obligations are allocated to the component $\ta{s}$ in a sequential order, \item the sequences of triggers (e.g.,~$p^1_\psi$ and $p^2_\psi$) are never overlapped, and \item the sequences of $\overline{\phi_2}$ are never overlapped (see~\cref{fig:sequentialised.obligations}).
\end{inparaenum}
This implies the following lemmas.
\begin{lemma}\label{lem:until.sequentialisation}
For a subformula of the form $\phi_1 \until_I \phi_2$ with $I = \langle l, u \rangle$ where $n = \lfloor \frac{u}{u - l} \rfloor + 2$, we can construct a tester $\ta{}$ with at most $\mathcal{O}(n^2)$  locations and $2n$ clocks. 
\end{lemma}
\begin{proof}[Proof sketch]
Each location of the tester $\ta{}$ corresponds to a configuration of the component $\ta{s}$.
The product of the $n$ component $\ta{s}$ from \cref{fig:simplified} (without  $\mathcal{A}^\textit{in}$ and $\mathcal{A}^\textit{out}$)
has $4^n$ locations, as each component $\ta{}$ may be in any of its four locations.
With $\mathcal{A}^\textit{in}$ and $\mathcal{A}^\textit{out}$,
we can only reach configurations where obligations are sequentialised on the $n$ component $\ta{}$s $C_\psi^1, \dots, C_\psi^n$:
\begin{itemize}
    \item For all component $\ta{s}$ not in location $s_0$, their indices form a (circular) consecutive sequence of numbers, e.g., $3, 4, 5$ or $6, 7, 1, 2$ (where $n = 7$).
    \item At most one component $\ta{}$ (at the beginning of the sequence) may be in location $s_1$ or $s_2$.
\end{itemize}
The number of reachable locations is at most $1 + 3n^2$. By the discussion in~\cref{subsec:general}, 
we need $n$ component $\ta{s}$, each uses $2$ clocks.
\end{proof}
\begin{lemma}
For a subformula of the form $\phi_1 \release_I \phi_2$ with $I = \langle l, u \rangle$ where $n = \lfloor \frac{u}{u - l} \rfloor + 2$, we can construct a tester $\ta{}$ with at most $\mathcal{O}(n^2)$ locations and $2n$ clocks.
\end{lemma}
\begin{proof}
Same as the proof of~\cref{lem:until.sequentialisation}.
\end{proof}
Through~\cref{lem:reversal} and with a bit of extra care, we can obtain the tester $\ta{s}$ also for the past $\mitl{}$ modalities with general time intervals.
\begin{lemma}
For a subformula of the form $\phi_1 \since_I \phi_2$ with $I = \langle l, u \rangle$ where $n = \lfloor \frac{u}{u - l} \rfloor + 2$, we can construct a tester $\ta{}$ with at most $\mathcal{O}(n^2)$  locations and $2n$ clocks. 
\end{lemma}
\begin{proof}
Same as the proof of~\cref{lem:until.sequentialisation}.
\end{proof}
\begin{lemma}
For a subformula of the form $\phi_1 \trigger_I \phi_2$ with $I = \langle l, u \rangle$ where $n = \lfloor \frac{u}{u - l} \rfloor + 2$, we can construct a tester $\ta{}$ with at most $\mathcal{O}(n^2)$ locations and $2n$ clocks.
\end{lemma}
\begin{proof}
Same as the proof of~\cref{lem:until.sequentialisation}.
\end{proof}

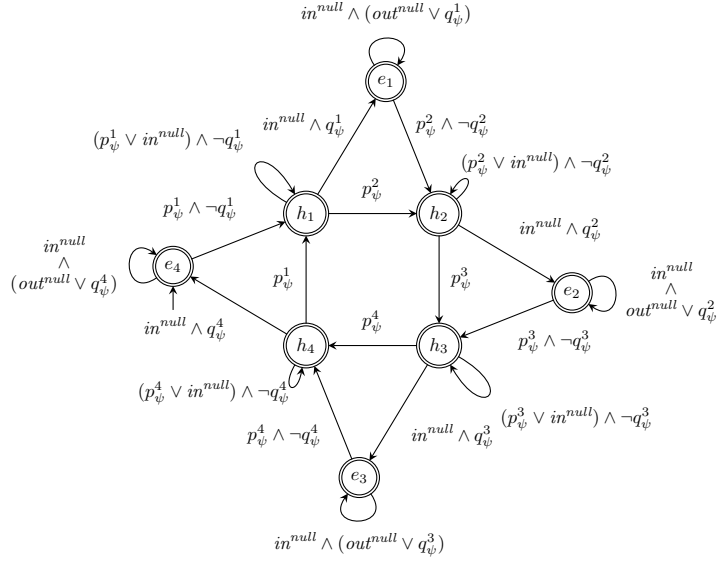
\begin{figure}[!htbp]
\centering
    \begin{tikzpicture}[->, node distance=2.5cm, transform shape, scale=0.7]
        \node[accepting, state](h1){$h_1$};
        \node[accepting, state, above of=h1, xshift=1.5cm](e1){$e_1$};
        \node[accepting, state, right of=h1](h2){$h_2$};
        \node[accepting, state, right of=h2, yshift=-1.5cm](e2){$e_2$};
        \node[accepting, state, below of=h2](h3){$h_3$};
        \node[accepting, state, below of=h3, xshift=-1.5cm](e3){$e_3$};
        \node[accepting, state, left of=h3](h4){$h_4$};
        \node[accepting, initial below,state, left of=h4, yshift=1.5cm](e4){$e_4$};
        \path
        (h1) edge[loopabove, out=150, in = 120, looseness=15] node[above left=1mm and 1mm, align=center]{$(p^1_\psi \lor \textit{in}^\textit{null}) \land \neg q^1_\psi$} (h1)
        (h1) edge[->] node[above left=.5mm and -1mm, align=center]{$\textit{in}^\textit{null} \land q^1_\psi$} (e1)
        (h1) edge[->] node[above=1mm, align=center]{$p^2_\psi$} (h2)
        (e1) edge[loopabove] node[above=1mm, align=center]{$\textit{in}^\textit{null} \land (\textit{out}^\textit{null} \lor q^1_\psi)$} (e1)
        (e1) edge[->] node[above right=.5mm and -.5mm, align=center]{$p^2_\psi \land \neg q^2_\psi$} (h2);
        
        \path
        (h2) edge[loopabove, out=40, in = 60, looseness=10] node[above right=-.5mm and -2.5mm, align=center]{$(p^2_\psi \lor \textit{in}^\textit{null}) \land \neg q^2_\psi$} (h2)
        (h2) edge[->] node[above right=1mm and 1mm, align=center]{$\textit{in}^\textit{null} \land q^2_\psi$} (e2)
        (h2) edge[->] node[right=1mm, align=center]{$p^3_\psi$} (h3)
        (e2) edge[loopright] node[right=.5mm, align=center]{\shortstack{$\textit{in}^\textit{null}$\\ $\land$ \\$\textit{out}^\textit{null} \lor q^2_\psi$}} (e2)
        (e2) edge[->] node[below right=1mm and 1mm, align=center]{$p^3_\psi \land \neg q^3_\psi$} (h3);
        
        \path
        (h3) edge[loopabove, out=-30, in = -60, looseness=15] node[below right=1mm and 1mm, align=center]{$(p^3_\psi \lor \textit{in}^\textit{null}) \land \neg q^3_\psi$} (h3)
        (h3) edge[->] node[below right=1mm and 1mm, align=center]{$\textit{in}^\textit{null} \land q^3_\psi$} (e3)
        (h3) edge[->] node[above=1mm, align=center]{$p^4_\psi$} (h4)
        (e3) edge[loopbelow] node[below=1mm, align=center]{$\textit{in}^\textit{null} \land (\textit{out}^\textit{null} \lor q^3_\psi)$} (e3)
        (e3) edge[->] node[below left=1mm and 1mm, align=center]{$p^4_\psi \land \neg q^4_\psi$} (h4);
        
        \path
        (h4) edge[loopabove, out=-120, in = -100, looseness=10] node[below left=-2.5mm and -1mm, align=center]{$(p^4_\psi \lor \textit{in}^\textit{null}) \land \neg q^4_\psi$} (h4)
        (h4) edge[->] node[below left=1mm and 1mm, align=center]{$\textit{in}^\textit{null} \land q^4_\psi$} (e4)
        (h4) edge[->] node[left=1mm, align=center]{$p^1_\psi$} (h1)
        (e4) edge[loopleft] node[left=1mm, align=center]{\shortstack{$\textit{in}^\textit{null}$\\ $\land$\\ $(\textit{out}^\textit{null} \lor q^4_\psi)$}} (e4)
        (e4) edge[->] node[above left=2.5mm and -1mm, align=center]{$p^1_\psi \land \neg q^1_\psi$} (h1);
    \end{tikzpicture}
\captionof{figure}{$\mathcal{A}^\textit{in}$ for $n = 4$.}
    \label{fig:seq.in}
\end{figure}    
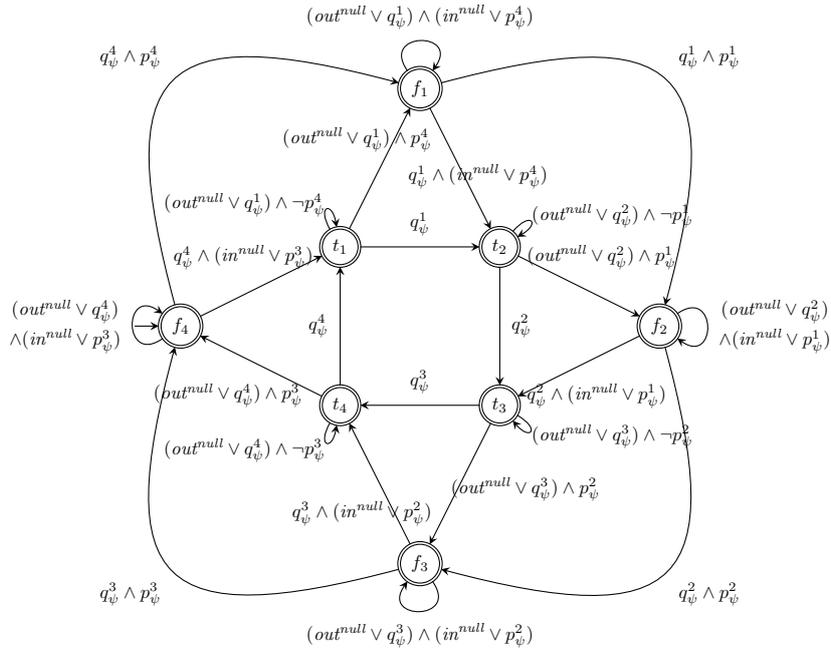
\begin{figure}[!htbp]
\centering
    \begin{tikzpicture}[->, node distance=3cm, transform shape, scale=0.7]
        \node[accepting, state](t1){$t_1$};
        \node[accepting, state, above of=t1, xshift=1.5cm](f1){$f_1$};
        \node[accepting, state, right of=t1](t2){$t_2$};
        \node[accepting, state, right of=t2, yshift=-1.5cm](f2){$f_2$};
        \node[accepting, state, below of=t2](t3){$t_3$};
        \node[accepting, state, below of=t3, xshift=-1.5cm](f3){$f_3$};
        \node[accepting, state, left of=t3](t4){$t_4$};
        \node[accepting, initial left,state, left of=t4, yshift=1.5cm](f4){$f_4$};
        \path
        (t1) edge[loopabove, out=120, in = 100, looseness=10] node[above left=-2.5mm and -1mm, align=center]{$(\textit{out}^\textit{null} \lor q^1_\psi) \land \neg p^4_\psi$} (t1)
        (t1) edge[->] node[above left=2mm and -11mm, align=center]{$(\textit{out}^\textit{null} \lor q^1_\psi) \land p^4_\psi$} (f1)
        (t1) edge[->] node[above=1mm, align=center]{$q^1_\psi$} (t2)
        (f1) edge[loopabove] node[above=1mm, align=center]{$(\textit{out}^\textit{null} \lor q^1_\psi) \land (\textit{in}^\textit{null} \lor p^4_\psi)$} (f1)
        (f1) edge[->] node[above right=-5mm and -11mm, align=center]{$q^1_\psi \land (\textit{in}^\textit{null} \lor p^4_\psi)$} (t2)
        (f1) edge[->, bend left=60, looseness=2] node[above right=1mm and 1mm, align=center]{$q^1_\psi \land p^1_\psi$} (f2);
        
        \path
        (t2) edge[loopabove, out=50, in = 30, looseness=10] node[above right=-2.5mm and -1mm, align=center]{$(\textit{out}^\textit{null} \lor q^2_\psi) \land \neg p^1_\psi$} (t2)
        (t2) edge[->] node[above right=2mm and -11mm, align=center]{$(\textit{out}^\textit{null} \lor q^2_\psi) \land p^1_\psi$} (f2)
        (t2) edge[->] node[right=1mm, align=center]{$q^2_\psi$} (t3)
        (f2) edge[loopright] node[right=1mm, align=center]{\shortstack{$(\textit{out}^\textit{null} \lor q^2_\psi)$ \\ $\land (\textit{in}^\textit{null} \lor p^1_\psi)$}} (f2)
        (f2) edge[->] node[below right=2mm and -11mm, align=center]{$q^2_\psi \land (\textit{in}^\textit{null} \lor p^1_\psi)$} (t3)
        (f2) edge[->, bend left=60, looseness=2] node[below right=1mm and 1mm, align=center]{$q^2_\psi \land p^2_\psi$} (f3);
        
        \path
        (t3) edge[loopabove, out=-30, in = -50, looseness=10] node[below right=-2.5mm and -1mm, align=center]{$(\textit{out}^\textit{null} \lor q^3_\psi) \land \neg p^2_\psi$} (t3)
        (t3) edge[->] node[below right=-2.5mm and -3mm, align=center]{$(\textit{out}^\textit{null} \lor q^3_\psi) \land p^2_\psi$} (f3)
        (t3) edge[->] node[above=1mm, align=center]{$q^3_\psi$} (t4)
        (f3) edge[loopbelow] node[below=1mm, align=center]{$(\textit{out}^\textit{null} \lor q^3_\psi) \land (\textit{in}^\textit{null} \lor p^2_\psi)$} (f3)
        (f3) edge[->] node[below left=2mm and -11mm, align=center]{$q^3_\psi \land (\textit{in}^\textit{null} \lor p^2_\psi)$} (t4)
        (f3) edge[->, bend left=60, looseness=2] node[below left=1mm and 1mm, align=center]{$q^3_\psi \land p^3_\psi$} (f4);
        
        \path
        (t4) edge[loopabove, out=-120, in = -100, looseness=10] node[below left=-2mm and -1mm, align=center]{$(\textit{out}^\textit{null} \lor q^4_\psi) \land \neg p^3_\psi$} (t4)
        (t4) edge[->] node[below left=2mm and -9mm, align=center]{$(\textit{out}^\textit{null} \lor q^4_\psi) \land p^3_\psi$} (f4)
        (t4) edge[->] node[left=1mm, align=center]{$q^4_\psi$} (t1)
        (f4) edge[loopleft] node[left=1mm, align=center]{\shortstack{$(\textit{out}^\textit{null} \lor q^4_\psi)$ \\$ \land (\textit{in}^\textit{null} \lor p^3_\psi)$}} (f4)
        (f4) edge[->] node[above left=2mm and -11mm, align=center]{$q^4_\psi \land (\textit{in}^\textit{null} \lor p^3_\psi)$} (t1)
        (f4) edge[->, bend left=60, looseness=2] node[above left=1mm and 1mm, align=center]{$q^4_\psi \land p^4_\psi$} (f1);
        
    \end{tikzpicture}
\captionof{figure}{$\mathcal{A}^\textit{out}$ for $n = 4$.}
    \label{fig:seq.out}
\end{figure}

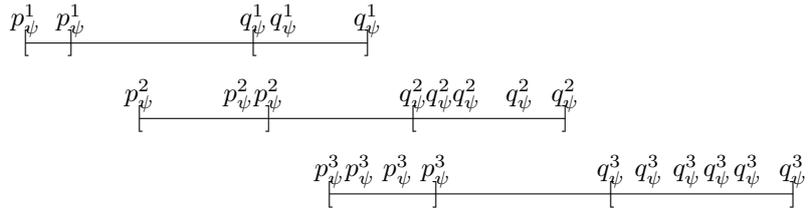
\begin{figure}[!htbp]
  \centering
  \scalebox{1}{\begin{tikzpicture}[node distance=3mm]
    \node(1) {$[$};
    \node[above of=1] () {$p^1_\psi$};
    \node[right of=1,node distance=0.6cm](1') {$]$};
    \node[above of=1'] () {$p^1_\psi$};

    \node[right of=1,node distance=3cm](1a) {$[$};
    \node[above of=1a] () {$q^1_\psi$};
    \node[right of=1a,node distance=2cm,xshift=-0.5cm](1b) {$]$};
    \node[above of=1b, xshift=-1.1cm] () {$q^1_\psi$};
    \node[above of=1b] () {$q^1_\psi$};

    \draw (1.center) -- (1b.center);

    \node[below of=1,node distance=1cm,xshift=1.5cm](2) {$[$}; 
    \node[above of=2] () {$p^2_\psi$};
    \node[right of=2,node distance=1.7cm](2') {$]$};
    \node[above of=2', xshift=-0.4cm] () {$p^2_\psi$};
    \node[above of=2'] () {$p^2_\psi$};

    \node[right of=2,node distance=3.6cm](2a) {$[$}; 
    \node[above of=2a] () {$q^2_\psi$};
    \node[right of=2a,node distance=2cm](2b)  {$]$}; 
    \node[above of=2b, xshift=-1.65cm] () {$q^2_\psi$};
    \node[above of=2b, xshift=-1.3cm] () {$q^2_\psi$};
    \node[above of=2b, xshift=-0.6cm] () {$q^2_\psi$};
    \node[above of=2b] () {$q^2_\psi$};

    \draw (2.center) -- (2b.center);

    \node[below of=2,node distance=1cm,xshift=2.5cm](3) {$[$}; 
    \node[above of=3] () {$p^3_\psi$};
    \node[right of=3,node distance=1.4cm](3') {$]$};
    \node[above of=3', xshift=-1.0cm] () {$p^3_\psi$};
    \node[above of=3', xshift=-0.5cm] () {$p^3_\psi$};
    \node[above of=3'] () {$p^3_\psi$};

    \node[right of=3,node distance=3.7cm](3a) {$[$}; 
    \node[above of=3a] () {$q^3_\psi$};
    \node[right of=3a,node distance=2cm, xshift=0.4cm](3b)  {$]$}; 
    \node[above of=3b,xshift=-1.9cm] () {$q^3_\psi$};
    \node[above of=3b,xshift=-1.4cm] () {$q^3_\psi$};
    \node[above of=3b,xshift=-1.0cm] () {$q^3_\psi$};
    \node[above of=3b,xshift=-0.6cm] () {$q^3_\psi$};
    \node[above of=3b] () {$q^3_\psi$};

    \draw (3.center) -- (3b.center);

  \end{tikzpicture}}
  \caption{An example of three obligations sequentialised.}
  \label{fig:sequentialised.obligations}
\end{figure}

\paragraph{Comparison with $\mightyl{}$.}
The construction for $\phi_1 \until_{\langle l, u \rangle} \phi_2$ in \citep{brihaye2013mitl, brihaye2014mitl} merges obligations and uses  clock values to represent them. It does not distinguish 
the types of obligations, but the number of clocks it uses is roughly $4n$
(where $n = \lfloor \frac{u}{u - l} \rfloor + 2$).  
$\mightyl{}$~\citep{DBLP:conf/cav/BrihayeGHM17} uses only roughly $2n$
clocks for $\phi_1 \until_{\langle l, u \rangle} \phi_2$, but it distinguishes
different types of obligations (based on how the original obligations overlap with each other) and stores this information in the locations, resulting in an exponential blow-up in $n$.
Both of these constructions support only the future fragment, but a more crucial technical difference with our construction is that in these constructions, \emph{each obligation is satisfied by a single $\overline{\phi_2}$-event}, where in our construction \emph{each obligation is satisfied by a single $\overline{\phi_2}$-event, or a sequence of $\overline{\phi_2}$ events}; this additional generality enables us to handle the obligations independently.
Along with sequentialisation, our construction uses the same number of clocks but with exponentially less number of reachable locations. %

\paragraph{Comparison with other constructions.}
The original construction for $\phi_1 \until_{\langle l, u \rangle} \phi_2$~\citep{AFH96}, and
a later construction formulated in terms of \emph{timed signal transducers}~\citep{maler2006mitl},
are both based on the idea of \emph{guessing the exact instants when $\overline{\phi_2}$
starts (or ends) in the future}.
This capability, unfortunately, is not supported in standard $\ta{s}$ over timed words (without $\epsilon$-transitions~\citep{berard1998characterization}) as $\ta{s}$ can only reset clocks on events. 
The construction of \citep{akshay2024mitl} is based on essentially the same idea as~\citep{maler2006mitl} but formulated in terms of \emph{generalized timed automata} ($\gta{s}$)~\citep{akshay2023unified}, which support the missing capability discussed above through \emph{future clocks}.
By contrast, our construction is based on standard $\ta{s}$ and  independent of this peculiar capability (in light of~\cref{lem:reversal}, our construction can be seen as reminiscent of a compositional `reverse' of~\citep{maler2005real}). In fact, given the results of this section, $\gta{s}$ do not appear to offer significant advantages over standard $\ta{s}$ for $\mitl$ satisfiability  and model checking, as suggested in~\citep{akshay2024mitl}, especially considering the more mature and widely adopted tool support available for the latter.
Finally, to the best of our knowledge, these constructions have never been implemented.

\label{sec:e}
\section{Implementation and Experiments}
\label{sec:exp}
We give a detailed description of our experimental evaluation in this section. 

\subsection{Implementation}
\label{subsec:imp}

\paragraph{Overview.} The tool $\mightyppl{}$ is a command-line program written in C\texttt{++}17.
Given an $\mitlpp{}$ formula, $\mightyppl{}$ can output tester $\ta{s}$
(or component $\ta{s}$, if a tester $\ta{}$ comprises of multiple component $\ta{s}$) individually, or generate the synchronous product of all tester $\ta{s}$
(the `flattened' model) as a single monolithic $\ta{}$. In both cases, $\mightyppl{}$ can output  
in \textsc{TChecker} or \textsc{Uppaal} format, 
enabling satisfiability and model checking of $\mitlpp{}$ using these tools as the
back-ends
(\textsc{Uppaal} models can be used with \textsc{LTSmin} for multi-core model checking~\citep{laarman2013multi}), or the user can use a built-in implementation of the standard backward (nested) fixpoint algorithm~\citep{henzinger1994symbolic}
(based on \textsc{MoniTAal}~\citep{Monitaal} and PARDIBAAL~\citep{Pardibaal})
for satisfiability and model checking.
The system is depicted in~\cref{fig:system}.
The source code repository is hosted on GitHub at
\begin{center}
\url{https://github.com/hsimho/MightyPPL}
\end{center}
We compiled $\mightyppl{}$ with commit \href{https://github.com/DEIS-Tools/MoniTAal/commit/4bbf3cc1b199c881e99e0d7f71dffefd969c13a2}{\texttt{4bbf3cc}} of \textsc{MoniTAal},
commit \href{https://github.com/DEIS-Tools/PARDIBAAL/commit/1c8f7c2f64cd2febe0d7874862c768d627998d31}{\texttt{1c8f7c2}} of PARDIBAAL, and
commit \href{https://github.com/ssoelvsten/buddy/commit/5aca063a4b2e90352480f3dd24daeb6dbefa2d33}{\texttt{5aca063}} of BuDDy.
\paragraph{Symbolic alphabets and synchronisation.}
Our construction from $\mitlpp{}$ to $\ta{s}$
is formulated in terms of $\ta{s}$ over Boolean combinations of atomic propositions;
however, this representation is not directly supported by our target back-end tools.
While tools like  Spot~\citep{duret2022spot} and \texttt{Owl}~\citep{kvretinsky2018owl} support symbolic alphabets, 
\textsc{Uppaal} and \textsc{TChecker} only
support \emph{synchronisers}, which are inadequate for our purpose.
In $\mightyl{}$, this issue is circumvented by using Boolean variables that are non-determistically set to $\top$ or $\bot$ as atomic propositions in a round-robin `cascade product', where a cyclic counter is incremented for each transition taken in a tester $\ta{}$, i.e.~a synchronised transition is realised as a sequence of steps, one for each individual tester $\ta{}$.
$\mightyppl{}$ works similarly in the individual tester / component mode, but it also uses wildcard `$\ast$' values to effectively force the back-end tools to do `symbolic execution' (rather than `exhaustive testing' as in $\mightyl{}$).
Alternatively, $\mightyppl{}$ uses BDDs to synchronise transitions of tester / component $\ta{s}$ into joint transitions in monolithic $\ta{s}$.

\paragraph{Forward and backward reachability analysis of tester $\ta{s}$.}
Symbolic letters are particularly useful
when explicitly constructing a flattened model: 
$\mightyppl{}$ considers all combinations of transitions in component $\ta{s}$ and eliminates \emph{conflicting} joint transitions based on symbolic letters.
Specifically, $\mightyppl{}$ constructs only the locations and transitions that are forward reachable (without considering timing constraints) from the initial location of the synchronous product. 
$\mightyppl{}$ also uses a backward reachability
analysis to rule out locations with no B\"uchi accepting runs.
As demonstrated in the experiments below,
these simple optimisations often results in monolithic $\ta{s}$ 
of manageable size that can be analysed by back-end tools.

\paragraph{Using the tool.}
$\mightyppl{}$ can be invoked with the following command:
\begin{center}
\verb/mitppl <in_spec_file> --{fin|inf} [out_file --{tck|xml}]/
\end{center}
where \verb/in_spec_file/ is the $\mitlpp{}$ formula and \texttt{fin} / \texttt{inf} specifies finite / infinite-word acceptance.
If \verb/out_file/ is specified, the format (\texttt{tck} / \texttt{xml}) must be too.
If the option \texttt{-{}-noflatten} is set, the tool will output individual tester / component $\ta{s}$ instead of a flattened model. The option \texttt{-{}-debug} can also be set to see some diagnostic messages. %
For model checking, one can edit the generated files (or the file \texttt{MightyPPL.cpp}) to `plug in' the system model $\mathcal{M}$ (see the `timed lamp' benchmark in the source code repository for an example).

\subsection{Experiments}

   \begin{figure}[t]
   \begin{center}
   \includegraphics[width=14cm]{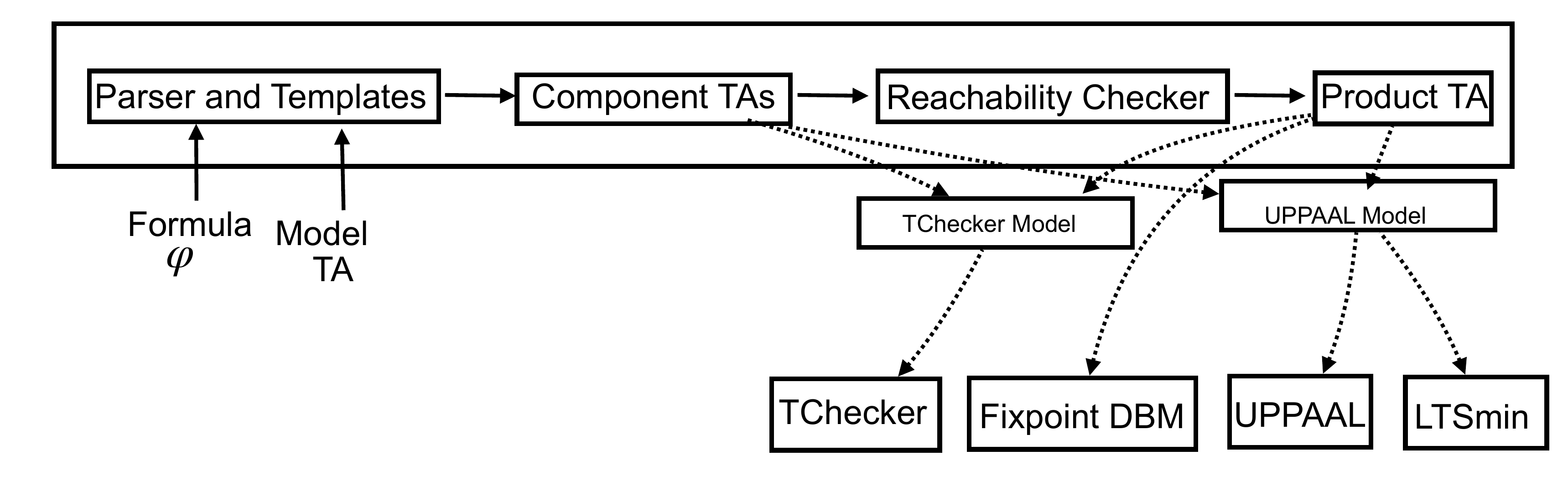}
   \end{center}
    \caption{$\mightyppl{}$ architecture. An arrow from one box to another implies that the output of the source can be used as input for the destination.}
    \label{fig:system}
\end{figure}

We focus on highlighting the main technical advances in \textsc{MightyPPL}, i.e.~the \emph{support for past modalities and Pnueli modalities}, the \emph{improved construction for $\mitl{}$ modalities with general time intervals}, and \emph{symbolic encoding of atomic propositions}---to the best of our knowledge, these features are not currently supported by any other tool for $\mitl{}$ or similar logics.
For benchmarks that involve only future $\mitl{}$ formulae,
we also compare with \textsc{MightyL}. 
 We use the following benchmarks on \emph{satisfiability} and \emph{model checking}:
\begin{enumerate}

    \item Simple parametrised formulae with unilateral intervals, adapted from \citep{DBLP:conf/cav/BrihayeGHM17}. 

\item A set of benchmarks  from Acacia-Bonsai \citep{DBLP:conf/tacas/CadilhacP23} which are also part of the SYNTCOMP 2021 LTL competition suite, adapted by adding time intervals.

  \item A set of benchmarks used in  \citep{dokhanchi2017formal}   
for debugging  $\mitl$ specifications. %

  \item A set of real-world specifications used in robotic missions \citep{menghi1, menghi2, menghi3}. 
  
  \item A food delivery timed path-planning problem, formulated as model checking.   
  
    \item A model checking example on a `timed lamp' $\ta{}$ model, adapted from~\citep{bersani2016tool}. %
  
    \item A model checking example on Fischer's mutual exclusion protocol~\citep{Behrmann2006}. %
    
  \item The pinwheel scheduling problem~\citep{Holte1989pinwheel}, formulated as model checking.
  
\end{enumerate}
All the experiments were conducted on a desktop machine with an Intel i9-13900K CPU and 64GB memory.
To demonstrate the versatility of $\mightyppl{}$, we make use of various back-ends including 
\textsc{Uppaal} (v4.1.19 rev.~5648), \textsc{LTSmin} (\texttt{v3.0.2-147-g07f9bf}), and \textsc{TChecker} (commit \href{https://github.com/ticktac-project/tchecker/commit/43a5763cdf339c4c5c4d1bf55388bb697bc4c13d}{\texttt{43a5763}}). 
Excluding the pinwheel scheduling benchmarks, which specifically evaluate multi-core model checking, we report the total runtime of each toolchain:
each reported runtime includes the time for any necessary model construction, compilation, or flattening.
For experiments that involve \textsc{LTSmin}, the tool was executed on a single thread to ensure a consistent comparison with the other back-end tools (\textsc{Uppaal} and \textsc{TChecker}), which are single-threaded.
In the tables, we use the following notations:
\begin{itemize}
\item `\texttt{comp}' (components)
means outputting tester / component $\ta{s}$ individually;
\item `\texttt{flat}' (flattening)
means outputting monolithic $\ta{s}$ as described in~\cref{subsec:imp};
\item `\texttt{verifyta}' means using \textsc{Uppaal} as the back-end tool (exclusively for finite-word satisfiability); `\texttt{-b}' or `$\texttt{-d}$' means using breadth-first or depth-first search orders, respectively;
\item `\texttt{opaal\_ltsmin}' means using \textsc{LTSmin} (with the \texttt{opaal} front-end to support the \textsc{Uppaal} format) as the back-end tool (exclusively for infinite-word satisfiability);
\item `\texttt{fp}' means using the built-in backward fixpoint algorithm implementation;
\item `\texttt{tck -{}-inf}' or `\texttt{tck -{}-fin}' means
using  \textsc{TChecker} as the back-end tool for infinite- and finite-word satisfiability, respectively.
\end{itemize}

\subsubsection{$\mightyl{}$ benchmarks}
We consider the finite- and infinite-word satisfiability of formulae:
\begin{IEEEeqnarray*}{rClrCl}
F(k,I) & = & \textstyle{\bigwedge_{i=1}^k}\eventually_I p_i \;, \quad & U(k,I) & = & (\cdots (p_1 \until_I p_2)\until_I \cdots)\until_I p_k \;, \\
G(k,I) & = & \textstyle{\bigwedge_{i=1}^k} \globally_I p_i \;, \quad & R(k, I) & = & (\cdots (p_1 \release_I p_2)\release_I \cdots)\release_I p_k \;,
\end{IEEEeqnarray*}
\vspace*{-.8cm}
\begin{IEEEeqnarray*}{rCl}
\mu(k) & = & \textstyle{\bigwedge_{i=1}^k} \eventually_{[3(i - 1), 3i]}t_i \land \globally (\neg p) \;, \\
\theta(k) & = & \neg \textstyle{(\bigwedge_{i=1}^k} \globally \eventually p_i \implies \globally (q \implies \eventually_{[100, 1000]} r))  \;.
\end{IEEEeqnarray*}
Following~\citep{DBLP:conf/cav/BrihayeGHM17}, for $\mu(k)$ we only check finite-word satisfiability  and similarly for $\theta(k)$ we only check infinite-word satisfiability,
due to the nature of these formulae.
To ensure an unbiased comparison with \textsc{MightyL} on this set of benchmarks, we evaluated both tools  using the same back-ends on an identical execution environment: the virtual machine provided at~\citep{MightyLVM} (executed natively through Docker).
The results are in~\cref{tab:parametrisedpart} and plotted in~\cref{fig:plot_1,fig:plot_2,fig:plot_3} (all satisfiable).
Using a log-log scale, each plot compares the $\mightyl{}$ and $\mightyppl{}$ front-ends for a specific back-end tool. Points below the diagonal line show that $\mightyppl{}$ yielded faster runtimes, and points on the edge of the graph represent instances that timed out. %
Thanks to the improved symbolic encoding,  $\mightyppl{}$ performs much better on most testcases regardless of the back-end (\textsc{LTSmin} or \textsc{Uppaal}), sometimes by more than two orders of magnitude (notable exceptions are $F\big(5,[1, 2]\big), F\big(5,[3, 8]\big)$, etc., where \textsc{LTSmin} performs worse due to large branching factors).

\begin{table}[!htbp]                    
\caption{Execution times on the $\mightyl{}$ benchmarks.
 $\mightyppl{}$ is used with \texttt{-{}-noflatten}.
Times are in seconds.
`TO' indicates timeouts (300s) and `ERR' means out-of-memory errors.
} %
\label{tab:parametrisedpart}
\centering
\scalebox{0.8}{
\def\arraystretch{1.2}
\setlength\tabcolsep{2mm}
\begin{tabular}{lrrrrrr}
\toprule %
 & \multirowcell{2}{$\mightyl{}$ \\ \small \texttt{opaal\_ltsmin}}
& \multirowcell{2}{$\mightyl{}$ \\ \small \texttt{verifyta -b}}
& \multirowcell{2}{$\mightyl{}$ \\ \small \texttt{verifyta -d}}
 & \multirowcell{2}{$\mightyppl{}$ \\ \small \texttt{opaal\_ltsmin}}
 & \multirowcell{2}{$\mightyppl{}$ \\ \small \texttt{verifyta -b}}
 & \multirowcell{2}{$\mightyppl{}$ \\ \small \texttt{verifyta -d}} \\
Formula \\
\otoprule

$F\big(5,[0, \infty)\big)$                                      & 2.058    & 0.057    & 0.043          & 1.428         & 0.004              & 0.004           \\
$F\big(5,[0, 2]\big)$                                           & 1.847    & 0.053    & 0.043          & 1.438         & 0.005              & 0.004           \\
$F\big(5,[2, \infty)\big)$                                      & 2.143    & 0.069    & 0.057          & 1.554         & 0.010              & 0.007           \\
$F\big(2,[1, 2]\big)$                                           & 76.575   & 0.949    & 0.929          & 1.474         & 0.022              & 0.020           \\
$F\big(3,[1, 2]\big)$                                           & 122.174  & 1.754    & 1.502          & 26.984        & 0.039              & 0.030           \\
$F\big(5,[1, 2]\big)$                                           & 225.575  & 11.001   & 5.604          & TO            & 0.203              & 0.055           \\
$F\big(5,[3, 8]\big)$                                           & 220.238  & 10.577   & 5.548          & TO            & 0.217              & 0.058          \\
$G\big(5,[0, \infty)\big)$                                      & 1.753    & 0.054    & 0.041          & 0.301         & 0.003              & 0.003           \\
$G\big(5,[0, 2]\big)$                                           & 1.855    & 0.051    & 0.042          & 0.325         & 0.004              & 0.004           \\
$G\big(5,[2, \infty)\big)$                                      & 1.862    & 0.054    & 0.043          & 0.335         & 0.004              & 0.004           \\
$G\big(2,[1, 2]\big)$                                           & 4.442    & 0.015    & 0.014          & 0.911         & 0.013              & 0.012           \\
$G\big(3,[1, 2]\big)$                                           & 1.196    & 0.023    & 0.021          & 1.431         & 0.023              & 0.019            \\
$G\big(5,[1, 2]\big)$                                           & 3.460    & 0.096    & 0.071          & 6.207         & 0.176              & 0.067           \\
$G\big(5,[3, 8]\big)$                                           & 3.426    & 0.091    & 0.070          & 8.860         & 0.291              & 0.083           \\
$U\big(5,[0, \infty)\big)$                                      & 1.006    & 0.025    & 0.029          & 0.384         & 0.004              & 0.004           \\
$U\big(5,[0, 2]\big)$                                           & 0.994    & 0.024    & 0.031          & 0.400         & 0.004              & 0.004           \\
$U\big(5,[2, \infty)\big)$                                      & 1.135    & 0.034    & 0.193          & 0.573         & 0.006              & 0.006           \\
$U\big(2,[1, 2]\big)$                                           & 37.204   & 0.481    & 0.471          & 0.820         & 0.012              & 0.011           \\
$U\big(3,[1, 2]\big)$                                           & 80.216   & 0.986    & 2.242          & 1.458         & 0.022              & 0.020           \\
$U\big(5,[1, 2]\big)$                                           & 171.445  & 7.440    & TO             & TO            & 0.045              & 0.043            \\
$U\big(5,[3, 8]\big)$                                           & 169.950  & 6.773    & TO             & TO            & 0.045              & 0.040          \\
$R\big(5,[0, \infty)\big)$                                      & 0.995    & 0.025    & 0.023          & 0.319         & 0.003              & 0.003           \\
$R\big(5,[0, 2]\big)$                                           & 1.055    & 0.026    & 0.022          & 0.348         & 0.004              & 0.004           \\
$R\big(5,[2, \infty)\big)$                                      & 1.045    & 0.027    & 0.022          & 0.379         & 0.005              & 0.004           \\
$R\big(2,[1, 2]\big)$                                           & 0.545    & 0.009    & 0.008          & 0.589         & 0.008              & 0.007           \\
$R\big(3,[1, 2]\big)$                                           & 0.879    & 0.015    & 0.015          & 1.053         & 0.014              & 0.013             \\
$R\big(5,[1, 2]\big)$                                           & 2.249    & 0.064    & 0.043          & 1.833         & 0.026              & 0.026           \\
$R\big(5,[3, 8]\big)$                                           & 2.223    & 0.061    & 0.044          & 2.454         & 0.038              & 0.037           \\
$\mu\big(1\big)$                                                & -        & 0.003    & 0.003          & -             & 0.003              & 0.003           \\
$\mu\big(2\big)$                                                & -        & 0.518    & 0.479          & -             & 0.012              & 0.012           \\
$\mu\big(3\big)$                                                & -        & 2.261    & 1.649          & -             & 0.101              & 0.027           \\
$\mu\big(4\big)$                                                & -        & 55.616   & 11.338         & -             & 5.469              & 0.069           \\
$\theta\big(1,[100, 1000]\big)$                                 & 0.807    & -        & -              & 0.814         & -                  & -           \\
$\theta\big(2,[100, 1000]\big)$                                 & 2.215    & -        & -              & 0.841         & -                  & -           \\
$\theta\big(3,[100, 1000]\big)$                                 & ERR      & -        & -              & 0.918         & -                  & -           \\
$\theta\big(4,[100, 1000]\big)$                                 & ERR      & -        & -              & 1.989         & -                  & -           \\

\bottomrule     
\end{tabular}     
}     
\end{table}

\begin{figure}
\begin{minipage}[t]{0.32\textwidth}
\centering
   \includegraphics[width=4.9cm]{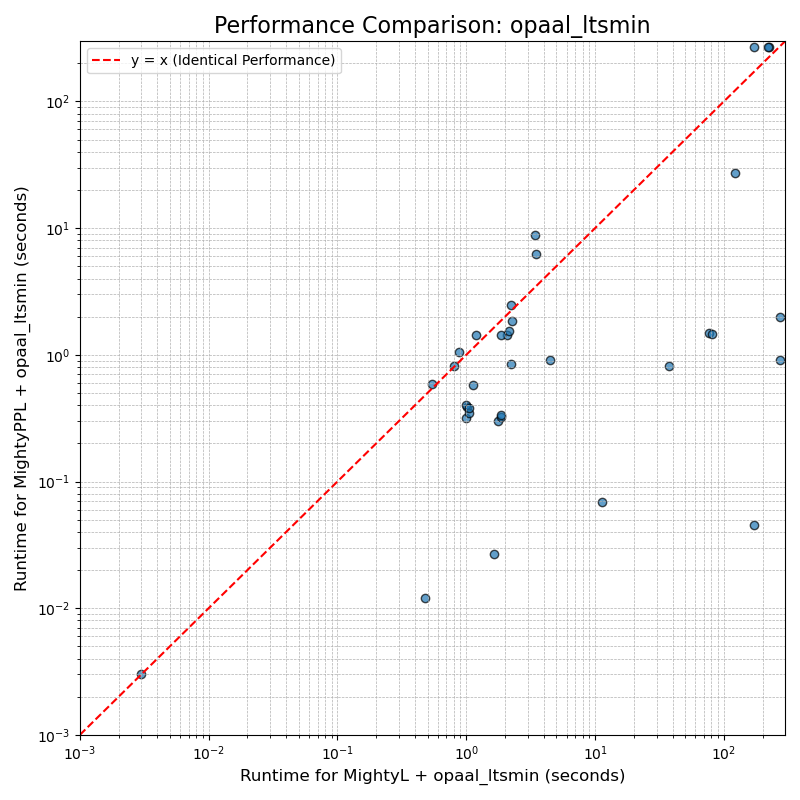}
\captionof{figure}{$\mightyl{}$ benchmarks with \texttt{opaal\_ltsmin} as the back-end.}
\label{fig:plot_1}
\end{minipage}
\hfill
\begin{minipage}[t]{0.32\textwidth}
\centering
   \includegraphics[width=4.9cm]{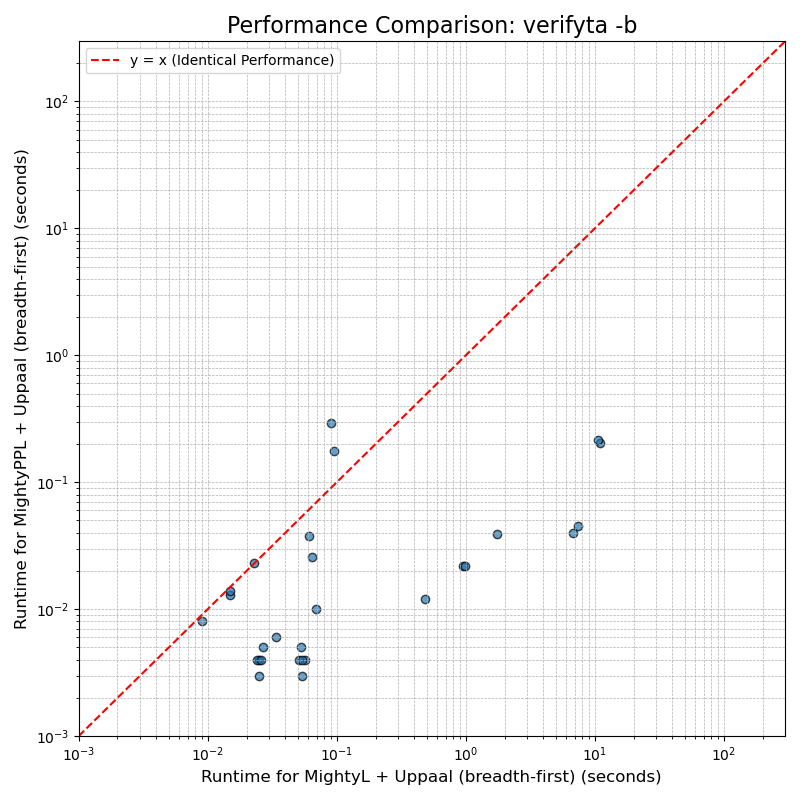}
\captionof{figure}{$\mightyl{}$ benchmarks with \texttt{verifyta -b} as the back-end.}
\label{fig:plot_2}
\end{minipage}
\hfill
\begin{minipage}[t]{0.32\textwidth}
\centering
   \includegraphics[width=4.9cm]{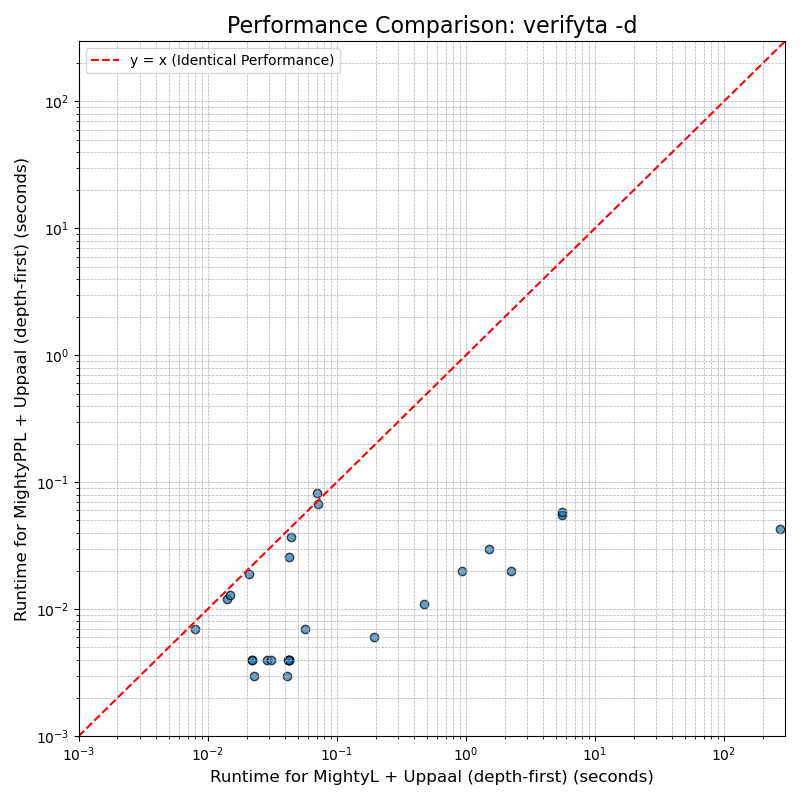}
\captionof{figure}{$\mightyl{}$ benchmarks with \texttt{verifyta -d} as the back-end.}
\label{fig:plot_3}
\end{minipage}
\end{figure}

\subsubsection{Acacia-Bonsai benchmarks}
This set of benchmarks is taken from Acacia-Bonsai \citep{DBLP:conf/tacas/CadilhacP23} which are also part of the SYNTCOMP 2021 LTL competition suite. We consider
the infinite-word satisfiability of $10$ formulae $\varphi_1 \dots \varphi_{10}$ chosen from this suite and modified by adding time intervals:
\begin{IEEEeqnarray*}{rCl}
    \varphi_1&=&((\eventually_{[0, 2]}(\globally p)) \wedge  (\globally(\eventually a))) \vee (\neg(\eventually_{[0, 2]}(\globally p)) \wedge  \neg (\globally(\eventually a))) \;, \\
\varphi_2&=&((\neg(\globally_{(2, \infty)} p)\vee (\eventually_{[0, 2]} q)) \wedge  (\globally (\eventually a))) \vee 
(\neg(\neg(\globally_{(2,\infty)} p)\vee (\eventually_{[0, 2]} q)) \wedge  \neg(\globally (\eventually a))) \;, \\
\varphi_3 &=& (((\globally(\eventually_{[1, 2]}p)) \vee (\globally(\eventually_{[1, 2]}q)) \vee (\globally(\eventually_{[1, 2]}r))) \wedge (\globally(\eventually(a))))\vee\\ 
&&(\neg((\globally(\eventually_{[1, 2]}p))\vee (\globally(\eventually_{[1, 2]}q)) \vee (\globally(\eventually_{[1, 2]}r))) \wedge  \neg (\globally(\eventually(a)))) \;, \\
\varphi_4&=&(\globally(( (p \until_{[1, 2]} q) \until_{[1, 2]} (\neg p) ) \until_{[1, 2]} (\neg r) ) \wedge  \globally( \eventually (a))) \vee\\
&&(\neg (\globally(( (p \until_{[1, 2]} q) \until_{[1, 2]} (\neg p) ) \until_{[1, 2]} (\neg r) )) \wedge  \neg (\globally( \eventually (a)))) \;, \\
\varphi_5&=&(((\globally(\neg p \vee (\eventually_{[1, 4]} q))) \wedge  (\globally(p \vee (\eventually_{[1, 4]} (\neg q))))) \wedge  (\globally(\eventually a))) \wedge  \\&&(\neg((\globally(\neg p \vee (\eventually_{[1, 4]} q))) \wedge  (\globally(p \vee (\eventually_{[1, 4]} (\neg q))))) \wedge  \neg (\globally(\eventually a))) \;, \\
\varphi_6&=&(( (\neg(\globally_{(2, \infty)} p) \vee (\eventually_{[0, 2]} q)) \vee (\neg(\globally_{(2, \infty)} (\neg p)) \vee (\eventually_{[0, 2]} (\neg q)) ) ) \wedge  (\globally(\eventually(a)))) \vee\\
&&(\neg((\neg(\globally_{(2, \infty)} p) \vee (\eventually_{[0, 2]} q)) \wedge  (\neg(\globally_{(2, \infty)} (\neg p)) \vee (\eventually_{[0, 2]} (\neg q)) )) \wedge \neg(\globally(\eventually(a)))) \;, \\
\varphi_7 & = & \resizebox{0.91\hsize}{!}{%
$(((\globally(\eventually_{[2, 4]} p)) \wedge (\globally(\eventually_{[2, 4]} q)) \wedge(\globally(\eventually_{[2, 4]} r)) \wedge (\globally(\eventually_{[2, 4]} s)) \wedge (\globally(\eventually_{[2, 4]} u))) \wedge (\globally(\eventually(a)))) \vee$ }  \\
& & \resizebox{0.91\hsize}{!}{%
$(\neg((\globally(\eventually_{[2, 4]} p)) \wedge (\globally(\eventually_{[2, 4]} q)) \wedge (\globally(\eventually_{[2, 4]} r)) \wedge (\globally(\eventually_{[2, 4]} s)) \wedge (\globally(\eventually_{[2, 4]} u))) \wedge\neg(\globally(\eventually (a))))$ } \;, \\
\varphi_8&=&((\globally_{[0, 2]} r) \wedge (\globally_{[0, 2]}(\neg p \vee (\eventually_{[1, 2]} q))) \wedge (\globally_{[0, 2]}(\neg q \vee (\eventually_{[1, 2]} r)))) \wedge  (\globally(\eventually(a)))) \vee \\
&&(\neg((\globally_{[0, 2]} r) \wedge  (\globally_{[0, 2]}(\neg p || (\eventually_{[1, 2]} q))) \wedge 
(\globally_{[0, 2]}(\neg q \vee (\eventually_{[1, 2]} r)))) \wedge  \neg(\globally(\eventually(a)))) \;, \\
\varphi_9&=& \resizebox{0.91\hsize}{!}{%
$(p \until_{[0, 2]} (q \wedge  (\nex_{[2, 4]}(r \wedge (\eventually_{[2, \infty)}(s \wedge (\nex_{[2, 4]}(\eventually_{[2, \infty)}(u \wedge (\nex_{[2, 4]}(\eventually_{[2, \infty)}(v \wedge  (\nex_{[2, 3]}(\eventually_{[2, \infty)} w)))))))))))))) \vee$ } \\
&& \resizebox{0.91\hsize}{!}{%
$(\neg (p \until_{[0, 2]} (q \wedge  (\nex_{[2, 4]}(r \wedge  (\eventually_{[2, \infty)}(s \wedge (\nex_{[2, 4]}(\eventually_{[2, \infty)}(u \wedge  (\nex_{[2, 4]}(\eventually_{[2, \infty)}(v \wedge  (\nex_{[2, 3]}(\eventually_{[2, \infty)} w)))))))))))))))$ } \;, \\
\varphi_{10}&=& ((\eventually_{[1, 10]}(p \wedge  (\nex_{[2, 4]}(q \wedge (\nex \eventually_{[2, \infty)} r))))) \wedge  (\globally(\eventually(a)))) \vee\\
&&(\neg(\eventually_{[1, 10]}(p \wedge  (\nex_{[2, 4]}(q \wedge  (\nex \eventually_{[2,\infty)} r))))) \wedge  \neg(\globally(\eventually(a)))) \;.
\end{IEEEeqnarray*}
The results are in~\cref{tab:acacia}. For satisfiable formulae clearly $\mightyppl{}$ in the `\texttt{comp}' mode is the faster approach.
For the `\texttt{flat}' mode, the backward fixpoint algorithm performs similarly as \textsc{TChecker}, only notably better on $\varphi_3$.
All the approaches, however, suffer for unsatisfiable formula $\varphi_5$
(a simple remedy is to implement an abstraction that replaces subformulae with atomic propositions; we chose not to as such is outside of the scope of this article). 

\begin{table}[!htb]
 \caption{Execution times on the Acacia benchmarks.  
Times are in seconds. `TO' indicates timeouts (300s).}
\label{tab:acacia}
\centering
\scalebox{0.8}{
\def\arraystretch{1.2}
\setlength\tabcolsep{2mm}
\begin{tabular}{l@{\hspace{2mm}}rrrr}
\toprule 
& 
 & \multirowcell{2}{$\mightyppl{}$ \\ \small \texttt{tck -{}-inf comp}}
 & \multirowcell{2}{$\mightyppl{}$ \\ \small \texttt{tck -{}-inf flat}}
 & \multirowcell{2}{$\mightyppl{}$ \\ \small \texttt{fp -{}-inf}} \\
$\varphi$  & Sat? & \\
\otoprule
$\varphi_1$        &  \cmark      & 0.013      & 0.035      & 0.027      \\ 
$\varphi_2$        &  \cmark      & 0.012      & 0.056      & 0.041      \\ 
$\varphi_3$        &  \cmark      & 0.507      & TO         & 254.192    \\ 
$\varphi_4$        &  \cmark      & 0.369      & TO         & TO         \\
$\varphi_5$        &  \xmark      & TO         & TO         & TO         \\
$\varphi_6$        &  \cmark      & 0.016      & 0.146      & 0.122      \\
$\varphi_7$        &  \cmark      & 2.951      & TO         & TO      \\ 
$\varphi_8$        &  \cmark      & 0.203      & TO         & TO       \\ 
$\varphi_9$        &  \cmark      & 17.015     & TO         & TO      \\ 
$\varphi_{10}$     &  \cmark      & 16.614     & TO         & TO        \\ 
\bottomrule
\end{tabular}
}
\end{table}

\smallskip

\subsubsection{Debugging specifications for cyber-physical systems}
\citep{dokhanchi2017formal} considers the \emph{debugging} problem of formal specifications written as $\mitl{}$ formulae. Namely, a subformula $\phi$ of a given formula $\varphi$ may be
\begin{inparaenum}[(1)]
\item \emph{Trivial}: $\phi$ is unsatisfiable or is a tautology; 
\item  \emph{Redundant}: $\phi$ is implied by the other conjuncts in the same conjunction;
\item \emph{Vacuous}: $\varphi$ implies $\varphi[\bot/\phi]$, where the latter is
    obtained from $\varphi$ by replacing $\phi$ with $\bot$.
\end{inparaenum}
Checking for these can be reduced to verifying the (un-)satisfiability
of some $\varphi'$ derived from $\varphi$ and $\phi$.  
In this set of benchmarks we consider the finite- and infinite-word satisfiability of the following formulae: 
\begin{IEEEeqnarray*}{rCl}
\varphi_1' & =  & \eventually_{[0, 20]} p_1 \land \lnot (\eventually_{[0, 30]} p_1) \;, \\
\varphi_2' & = & \neg \big((p_1 \implies p_1 \land \globally_{[0, 20]} p_1) 
 \lor \eventually_{[0, 30]} (p_1 \implies p_1 \land \globally_{[0, 20]} p_1)\big) \;,  \\
\varphi_3' & = & (p_1 \land \globally_{[0, 40]} p_1) \land \lnot \big( \globally_{[0, 40]} (p_1 \lor \eventually_{[0, 10]} p_1) \big) \;, \\
\varphi_4' & = & (p_2 \lor \eventually_{[0, 40]} p_2) \land \big(\eventually_{[0, 40]} (p_1 \land \globally_{[0, 30]} p_1)\big) 
 \land \neg \big((p_1 \lor p_3) \lor \eventually_{[0, 40]} (p_1 \lor p_3)\big) \;, \\
\varphi_5' & = & \eventually_{[10, 40]} \big( ( p_1 \lor p_3 \implies \eventually_{[0, 20]} p_2) \land \globally_{[0, 30]} p_1 \big) \\
 & &     {}  \land  \neg \eventually_{[10, 40]} \big( ( p_1 \lor \bot \implies \eventually_{[0, 20]} p_2) \land \globally_{[0, 30]} p_1 \big) \;, \\
\varphi_6' & = & \neg \Big((p \land \globally_{[0,40]} p \land \eventually_{[20,40]} \top) \implies \big((p \lor \eventually_{[0, 20]} p) \land \globally_{[0,20]} (p \lor \eventually_{[0,20]} p)\big)\Big) \;.
\end{IEEEeqnarray*}
For each of them we also consider a
modified version, where $\textit{past}(\varphi)$ is obtained by replacing all the future modalities with their past counterparts in $\varphi$.
All these formulae are unsatisfiable, and the results are in~\cref{tab:debugging1} (where we use the identical back-ends and execution environment through the virtual machine) and \cref{tab:debugging2} (where we run $\mightyppl{}$ and \textsc{TChecker} natively).
In comparison with $\mightyl{}$, the models generated by $\mightyppl{}$ in the `\texttt{comp}' mode perform better in general on both \textsc{LTSmin}
and \textsc{Uppaal}, 
with the exception of the infinite-word satisfiability of $\varphi_5'$ where the `\texttt{flat}' mode 
gives better performance.
For the back-ends, \textsc{TChecker} performs better than both \textsc{LTSmin}
and \textsc{Uppaal} on this set of benchmarks, most nobably with the harder cases
$\eventually \textit{past}(\varphi_3')$ 
and
$\eventually \textit{past}(\varphi_6')$. 
The backward fixpoint algorithm sometimes performs even better than \textsc{TChecker}, e.g.,~$\eventually \textit{past}(\varphi_6')$ in the infinite-word case.

\begin{table}[!htb]
\caption{Execution times on the specification debugging benchmarks.
$\mightyppl{}$ is used with \texttt{-{}-noflatten}.
Times are in seconds.
`TO' indicates timeouts (300s) and `-' means unsupported.}
\label{tab:debugging1}
\centering
\scalebox{0.8}{
\def\arraystretch{1.2}
\setlength\tabcolsep{2mm}
\begin{tabular}{lrrrrrr}
\toprule %
 & \multirowcell{2}{$\mightyl{}$ \\ \small \texttt{opaal\_ltsmin}}
& \multirowcell{2}{$\mightyl{}$ \\ \small \texttt{verifyta -b}}
& \multirowcell{2}{$\mightyl{}$ \\ \small \texttt{verifyta -d}}
 & \multirowcell{2}{$\mightyppl{}$ \\ \small \texttt{opaal\_ltsmin}}
 & \multirowcell{2}{$\mightyppl{}$ \\ \small \texttt{verifyta -b}} 
 & \multirowcell{2}{$\mightyppl{}$ \\ \small \texttt{verifyta -d}} \\
Formula \\
\otoprule
$\varphi_1'$                                  & 0.824   & 0.027     & 0.004     &  0.338  & 0.004   & 0.003      \\
$\varphi_2'$                                  & 0.361   & 0.004     & 0.004     &  0.338  & 0.004   & 0.003    \\
$\varphi_3'$                                  & 0.366   & 0.005     & 0.004     &  0.361  & 0.004   & 0.004    \\
$\varphi_4'$                                  & 0.551   & 0.010     & 0.009     &  0.429  & 0.005   & 0.004     \\
$\varphi_5'$                                  & 48.313  & 14.924    & 13.710    &  TO     & 0.165   & 0.230         \\
$\varphi_6'$                                  & 36.697  & 0.572     & 0.555     &  1.037  & 0.019   & 0.020     \\
$\eventually \textit{past}(\varphi_1')$       & -       & -         & -         &  0.895  & 0.005   & 0.004    \\
$\eventually \textit{past}(\varphi_2')$       & -       & -         & -         &  0.562  & 0.005   & 0.004    \\
$\eventually \textit{past}(\varphi_3')$       & -       & -         & -         &  15.590 & 0.008   & 0.007    \\
$\eventually \textit{past}(\varphi_4')$       & -       & -         & -         &  TO     & 0.017   & 0.022      \\
$\eventually \textit{past}(\varphi_5')$       & -       & -         & -         &  TO     & TO      & TO       \\
$\eventually \textit{past}(\varphi_6')$       & -       & -         & -         &  TO     & 13.438  & 19.770     \\
\bottomrule
\end{tabular}
}
\end{table}

\begin{table}[!htb]
\caption{Execution times on the specification debugging benchmarks.
Times are in seconds.
`TO' indicates timeouts (300s).}
\label{tab:debugging2}
\centering
\scalebox{0.8}{
\def\arraystretch{1.2}
\setlength\tabcolsep{2mm}
\begin{tabular}{lrrrrrr}
\toprule %
 & \multirowcell{2}{$\mightyppl{}$ \\ \small \texttt{tck -{}-inf comp}}
 & \multirowcell{2}{$\mightyppl{}$ \\ \small \texttt{tck -{}-inf flat}}
 & \multirowcell{2}{$\mightyppl{}$ \\ \small \texttt{fp -{}-inf}}
 & \multirowcell{2}{$\mightyppl{}$ \\ \small \texttt{tck -{}-fin comp}}
 & \multirowcell{2}{$\mightyppl{}$ \\ \small \texttt{tck -{}-fin flat}}
 & \multirowcell{2}{$\mightyppl{}$ \\ \small \texttt{fp -{}-fin}} \\
Formula \\
\otoprule
$\varphi_1'$                                  & 0.015   & 0.009     & 0.005     & 0.013     & 0.008  & 0.003          \\
$\varphi_2'$                                  & 0.012   & 0.012     & 0.008     & 0.008     & 0.008  & 0.004          \\
$\varphi_3'$                                  & 0.016   & 0.011     & 0.006     & 0.008     & 0.008  & 0.004        \\
$\varphi_4'$                                  & 0.043   & 0.023     & 0.015     & 0.010     & 0.013  & 0.006          \\
$\varphi_5'$                                  & TO      & 15.833    & 15.249    & 0.301     & 0.461  & 0.350          \\
$\varphi_6'$                                  & 0.284   & 0.585     & 1.450     & 0.046     & 0.227  & 0.255          \\

$\eventually \textit{past}(\varphi_1')$       & 0.074   & 0.026     & 0.019     & 0.010     & 0.012  & 0.006            \\
$\eventually \textit{past}(\varphi_2')$       & 0.078   & 0.029     & 0.021     & 0.010     & 0.013  & 0.006            \\
$\eventually \textit{past}(\varphi_3')$       & 1.536   & 0.094     & 0.058     & 0.016     & 0.023  & 0.012           \\
$\eventually \textit{past}(\varphi_4')$       & TO      & 213.522   & 0.209     & 0.036     & 0.060  & 0.032           \\
$\eventually \textit{past}(\varphi_5')$       & TO      & TO        & TO        & TO        & TO     & TO              \\
$\eventually \textit{past}(\varphi_6')$       & TO      & TO        & 26.774    & 7.746     & 7.854  & 1.703             \\
\bottomrule
\end{tabular}
}
\end{table}

\subsubsection{Robotic missions}
This set of benchmarks is adapted from \citep{menghi1, menghi2, menghi3}.
We consider the following specification patterns:
\begin{itemize}
\item \emph{Timed Sequenced Visit / Patrolling \emph{(TSV / TSP)}}: A robot must visit locations $l_1, \ldots l_n$ in this order at least once in the next $m$ time units / in every window of $m$ time units. We denote these patterns by 
$\textup{TSV} = \{ \PnF_{[0,m]}(l_1, l_2, \ldots l_n) \}$
and $\textup{TSP} = \{\globally \PnF_{[0,m]}(l_1, l_2, \ldots l_n)\}$. 
\item \emph{Timed Past Avoidance \emph{(TPA)}}: A robot can visit a certain location $l_1$ only if it has \emph{not} visited a location $l_2$ within $m$ time units in the past. We denote this by $\textup{TPA} = \globally(l_1 \rightarrow \past_{[0,m]} l_2)$. 
\end{itemize}
We consider the following formulae (which are instances of these patterns) and check their infinite-word satisfiability:
\begin{IEEEeqnarray*}{rCl}
\varphi_1'' & = &\globally (l_1 \implies \once_{[0, 2]} l_2) \land \globally (l_3 \implies \once_{[0, 2]} l_4) \;, \\
\varphi_2'' & = & \PnF_{[0, 2]}(p_1, p_2, p_3, p_4, p_5, p_6, p_7) \;, \\
\varphi_3'' & = & \globally \PnF_{[0, 1]}(p_1, p_2, p_3) \;, \\
\varphi_4'' & = &\PnF_{[0, 1]}(p_1, p_2, p_3)
           \land \globally (p_2 \implies \neg \once_{[0, 1]}p_1) \land \globally (p_3 \implies \lnot \once_{[0, 1]} p_2) \;.
\end{IEEEeqnarray*}
The results are in~\cref{tab:scheduling}. 
On this set of benchmarks, \textsc{TChecker}'s performance does not significantly benefit from flattening the model.
The backward fixpoint algorithm also performs well.
\begin{table}[!htb]
 \caption{Execution times on the robotic missions benchmarks.  
Times are in seconds. `TO' indicates timeouts (300s).}
\label{tab:scheduling}
\centering
\scalebox{0.8}{
\def\arraystretch{1.2}
\setlength\tabcolsep{2mm}
\begin{tabular}{ll@{\hspace{2mm}}rrrr}
\toprule 
&   & 
& \multirowcell{2}{$\mightyppl{}$ \\ \small \texttt{tck -{}-inf  comp}}
& \multirowcell{2}{$\mightyppl{}$ \\ \small \texttt{tck -{}-inf flat}}
 & \multirowcell{2}{$\mightyppl{}$ \\ \small \texttt{fp -{}-inf}} \\
Pattern & Formula  & Sat? &  \\
\otoprule
TPA &      $\varphi_1''$       & \cmark       & 0.008   & 0.012   & 0.007          \\
TSV   &   $\varphi_2''$       & \cmark       & 0.091   & 0.018   & 0.012           \\
TSP  &   $\varphi_3''$       & \cmark       & 0.024   & 0.054    & 0.042         \\
TSV, TPA    &   $\varphi_4''$     & \xmark  & 0.164    & 0.131    & 0.067        \\
\bottomrule
\end{tabular}
}
\end{table}

\subsubsection{Food delivery}
Recall the food delivery scenario from~\cref{ex:food}.
In this set of benchmarks, we consider a more sophisticated property where in addition to the 
requirement for delivery to be quick (within $15$ minutes of order), we also require that the food must be fresh (delivery is within $10$ minutes of pickup). 
These constraints are most naturally expressed in unilateral fragment of $\tptl{}$~\citep{AH94}, known as $\uptl{}$~\citep{concur23,concur25}: 
\[
\phi(K,L) = x.\fut(K \wedge y.\fut(L \wedge x < 15 \wedge y<10)) \;.
\]
By~\citep{concur25}, there is an equivalent $\mitlpp{}$ formula for any $\uptl{}$ formula.
In this case $\phi(K,L)$ is equivalent to 
\[
\phi'(K, L) = \PnF_{[0,15]}(K \wedge \eventually_{[0,10]} L, L)  \;.
\]
We can now pose the planning problem---given a map of the city (\cref{fig:icon-delivery-map}) and a pattern of orders, is it possible for the delivery driver to plan a route that satisfies the constraints?---as a model-checking problem.
We model the city as a $\ta{}$ $\mathcal{M}$ with the following atomic propositions:
\[
\{P, B, C, L0, L1, L2, P{:}L1, B{:}L1, C{:}L1, P{:}L2, B{:}L2, C{:}L2\} \;.
\]
The locations and transitions of $\mathcal{M}$ are as depicted~\cref{fig:icon-delivery-map}, with an additional self-loop (with no guards) on each location.
On each transition exactly one of $P, B, C, 
L0, L1, L2$ holds (which indicates the target location of the transition).
As before, $K{:}L$ (where $K \in \{P, B, C\}$ and $L \in \{L1, L2\}$) indicates an order from location $L$ for item $K$. %
We model-check $\mathcal{M}$ against the formula
\[
\pi(K, L, (l, u)) \implies  \eventually (K{:}L \land \neg \phi'(K, L))
\]
where the pattern of $K{:}L$ is specified by
\[
\pi(K, L, (l, u))  =  \eventually_{[0, 1]} (K{:}L) \land \globally_{[0, 50]} (K{:}L \implies ((\neg K{:}L) \until_{[l, \infty)} K{:}L) \land ((\neg K{:}L) \until_{[0, u]} K{:}L))
\]
for each $(l, u) \in \{ (3, 5), (8, 10), (10, 12), (15, 20) \}$,
$K \in \{P, B, C\}$, and $L \in \{ L1, L2 \}$
 (over finite timed words).
The results are in~\cref{tab:food}. 
For this set of benchmarks, \textsc{TChecker} performs similarly
on the models generated by $\mightyppl{}$ in both the `\texttt{comp}' and `\texttt{flat}' modes, significantly outperforming the backward fixpoint algorithm.

\begin{table}[!htb]
 \caption{Execution times on the food delivery benchmarks.  
Times are in seconds. `TO' indicates timeouts (300s).}
\label{tab:food}
\centering
\scalebox{0.8}{
\def\arraystretch{1.2}
\setlength\tabcolsep{2mm}
\begin{tabular}{l@{\hspace{2mm}}l@{\hspace{4mm}}l@{\hspace{4mm}}rrrr}
\toprule 
& & &
& \multirowcell{2}{ $\mightyppl{}$ \\ \small \texttt{tck -{}-fin comp}}
& \multirowcell{2}{ $\mightyppl{}$ \\ \small \texttt{tck -{}-fin flat}}
 & \multirowcell{2}{ $\mightyppl{}$ \\ \small \texttt{fp -{}-fin}} \\
$K$ & $L$ & $(l, u)$ & Hold? &  \\
\otoprule
  $P$  & $L1$    & $(3, 5)$       & \xmark        & 6.558      & 2.440   & 15.751           \\
  $P$  & $L1$    & $(8, 10)$      & \xmark        & 1.489      & 2.244   & 147.426           \\
  $P$  & $L1$    & $(10, 12)$     & \xmark        & 1.415      & 2.214   & 53.940           \\
  $P$  & $L1$    & $(15, 20)$     & \xmark        & 1.096      & 2.166   & 102.728           \\
  $P$  & $L2$    & $(3, 5)$       & \cmark        & 1.241      & 2.354   & 0.863            \\
  $P$  & $L2$    & $(8, 10)$      & \xmark        & 1.812      & 2.384   & 43.537            \\
  $P$  & $L2$    & $(10, 12)$     & \cmark        & 0.791      & 4.982   & 1.725            \\
  $P$  & $L2$    & $(15, 20)$     & \xmark        & 1.127      & 2.312   & 70.684            \\
  $B$  & $L1$    & $(3, 5)$       & \xmark        & 6.082      & 2.437   & TO            \\
  $B$  & $L1$    & $(8, 10)$      & \xmark        & 1.634      & 2.152   & 173.988            \\
  $B$  & $L1$    & $(10, 12)$     & \xmark        & 1.982      & 2.194   & 298.544            \\
  $B$  & $L1$    & $(15, 20)$     & \xmark        & 0.982      & 2.013   & 115.108            \\
  $B$  & $L2$    & $(3, 5)$       & \xmark        & 8.077      & 2.803   & 5.999          \\
  $B$  & $L2$    & $(8, 10)$      & \xmark        & 1.295      & 2.444   & 46.012          \\
  $B$  & $L2$    & $(10, 12)$     & \xmark        & 1.652      & 2.441   & 57.116           \\
  $B$  & $L2$    & $(15, 20)$     & \xmark        & 1.125      & 2.420   & 128.690           \\
  $C$  & $L1$    & $(3, 5)$       & \xmark        & 3.660      & 2.040   & 85.348             \\
  $C$  & $L1$    & $(8, 10)$      & \xmark        & 0.903      & 1.833   & 136.496             \\
  $C$  & $L1$    & $(10, 12)$     & \xmark        & 1.277      & 1.869   & 211.248             \\
  $C$  & $L1$    & $(15, 20)$     & \xmark        & 0.821      & 1.856   & TO             \\
  $C$  & $L2$    & $(3, 5)$       & \xmark        & 3.396      & 1.995   & TO           \\
  $C$  & $L2$    & $(8, 10)$      & \xmark        & 0.935      & 1.681   & 299.449           \\
  $C$  & $L2$    & $(10, 12)$     & \xmark        & 1.223      & 1.730   & TO           \\
  $C$  & $L2$    & $(15, 20)$     & \xmark        & 0.743      & 1.700   & TO           \\
\bottomrule
\end{tabular}
}
\end{table}

\subsubsection{Timed lamp}
This is a case study adapted from~\citep{bersani2016tool}.
We consider a lamp controlled by a single button that can be pushed by the user at any time. Whenever the button is pushed, the lamp blinks at the same instant and enters a state, where it blinks every $1$ time unit for three more times. The simple system can be modelled by the $\ta{}$ $\mathcal{M}_\textit{lamp}$ of~\cref{fig:timedlamp}, which is model-checked against the following formulae (over infinite timed words):
\begin{itemize}
\item $\alpha_1 = \globally (\textit{push} \implies \eventually_{\geq 3} \textit{blink})$: Each time the button is pushed, regardless of whether it is pushed again later, the lamp must blink at least once after (exactly) $3$ time units. This clearly holds in $\mathcal{M}_\textit{lamp}$.
\item $\alpha_2 = \globally\big(\textit{push} \land \textit{blink} \implies \eventually_{\leq 5} (\neg \textit{blink})\big)$: Each time the button is pushed, the lamp goes off (i.e.~$\neg \textit{blink}$ occurs) at least once in the next $5$ time units. This does not hold as the user can push  repeatedly every, e.g.,~$2$ time units.
\item $\alpha_3 = \globally\big(\textit{blink} \land \eventually_{\leq 1} \textit{blink} \land \globally_{\leq 9} (\eventually_{\leq 1} \textit{blink}) \implies \PnF_{\leq 9} ( \textit{push},\textit{push})\big)$:
    If the lamp blinks continuously for $9$ time units, then the button must have been pressed at least twice in this period.
    This holds in $\mathcal{M}_\textit{lamp}$.
\item $\alpha_4 = \globally\big(\textit{blink} \land \eventually_{\leq 1} \textit{blink} \land \globally_{\leq 9} (\eventually_{\leq 1} \textit{blink}) \implies \PnF_{\leq 9} ( \textit{push},\textit{push},\textit{push})\big)$:
    This does not hold as one may, e.g.,~push the button every $3.9$ time units.
\end{itemize}    
The results are in~\cref{tab:timedlamp}. All the approaches perform well on the first two (simpler) formulae.
For the other two formulae with Pnueli modalities, \textsc{TChecker} performs much better than the backward fixpoint algorithm.

\begin{figure}
\begin{minipage}[t]{0.49\textwidth}
\centering
\begin{tikzpicture}[->, node distance=5cm, transform shape, scale=0.7]
   \node[initial left,state, accepting](0){$s_0$};
   \node[state, right of=0](1){$s_1$};
   
   \path
   (0) edge[loopabove] node[above=1mm, align=center, looseness=20, out=120, in=90]{$\neg \textit{push} \land \neg \textit{blink}$ \\ $x = 1; x := 0$} (0)
   (0) edge[->, bend left=10] node[above=1mm, align=center]{$\textit{push} \land \textit{blink}$ \\ $x := 0, y := 0$} (1)
   (1) edge[loopabove] node[above=1mm,align=center]{$\textit{push} \land \textit{blink}$ \\ $x := 0, y := 0$} (1)
   (1) edge[loopbelow] node[below=1mm, align=center]{$\neg \textit{push} \land \textit{blink}$ \\ $x = 1 \land y \leq 2; x := 0$} (1)
   (1) edge[->, bend left=10] node[below=1mm, align=center]{$\neg \textit{push} \land \textit{blink}$ \\ $x = 1, y = 3; x := 0$} (0);
 \end{tikzpicture}
\captionof{figure}{The model $\ta{}$ $\mathcal{M}_\textit{lamp}$.}
\label{fig:timedlamp}
\end{minipage}
\hfill
\begin{minipage}[b]{0.49\textwidth}
\centering
 \captionof{table}{Execution times on the timed lamp benchmarks. Times are in seconds.
 `TO' indicates timeouts (300s).}
\scalebox{0.65}{
\def\arraystretch{1.2}
\setlength\tabcolsep{2mm}
\begin{tabular}{l@{\hspace{2mm}}rrrr}
\toprule 
& 
& \multirowcell{2}{ $\mightyppl{}$ \\ \small \texttt{tck -{}-inf comp}}
& \multirowcell{2}{ $\mightyppl{}$ \\ \small \texttt{tck -{}-inf flat}}
 & \multirowcell{2}{ $\mightyppl{}$ \\ \small \texttt{fp -{}-inf}} \\
$\varphi$  & Hold? &  \\
\otoprule
   $\alpha_1$       & \cmark       & 0.027      & 0.012   & 0.009           \\
  $\alpha_2$       & \xmark        & 0.008      & 0.017   & 0.013            \\
  $\alpha_3$       & \cmark        & 0.070      & 0.267   & 152.755            \\
  $\alpha_4$     & \xmark          & 0.045      & 0.430   & TO               \\
\bottomrule
\end{tabular}
}
\label{tab:timedlamp}
\end{minipage}
\end{figure}

\subsubsection{Fischer's mutual exclusion protocol}
This is a classic mutual exclusion algorithm (from~\citep{lamport1987fast})
that has been used as a standard example for $\ta{}$-based
verification tools such as \textsc{Uppaal}.
We use the file \texttt{fischer.sh} from \textsc{TChecker}'s source-code repository~\citep{TChecker} to generate the system models (asynchronous networks of $\ta{s}$---one $\ta{}$ $\mathcal{P}_i$ for each process $i$) with a fixed delay $K = 5$; then we add the atomic propositions $\textit{req}_i, \textit{wait}_i, \textit{cs}_i, \textit{idle}_i$ to signify which location is entered. For example, $\textit{cs}_i$ holds on the transition from $\ell_\textit{req}$ to $\ell_\textit{cs}$ in  $\mathcal{P}_i$, but at this point $\textit{req}_i, \textit{wait}_i, \textit{idle}_i$ and $\{\, \textit{req}_j, \textit{wait}_j, \textit{cs}_j, \textit{idle}_j \mid j \neq i  \,\}$ must all be $\bot$. The system models are model-checked against the following formulae (over infinite timed words):
\begin{itemize}

    \item $\alpha_1' = \bigwedge_{i} \big( \globally(\textit{cs}_i \implies \past_{[0,10]} \textit{req}_i) \big)$:  
    Whenever process $i$ enters the critical section, it must have issued a request at most 10 time units earlier. This does not hold as any process can wait in $\ell_\textit{wait}$ for arbitrarily long before entering $\ell_\textit{cs}$.
    
    \item $\alpha_2' = \bigwedge_{i} \big( \globally(\textit{cs}_i \implies \past_{[5,10]} \textit{req}_i) \big)$:  
    Same as $\theta_1$ but process $i$ must have issued a request between 5 and 10 time units earlier. This does not hold for the same reason as $\theta_1$.
    
    \item $\alpha_3' = \bigwedge_{i} \big(\globally(\textit{cs}_i \implies \neg \PnF_{[0,11]} (\textit{cs}_i, \textit{cs}_i))\big)$:  
    No process $i$ can enter the critical section more than three times within any $11$ time units window. This does not hold as a process can delay for $0$ time units in $\ell_\textit{idle}$ and $\ell_\textit{req}$, for exactly $5.1$ time units in $\ell_\textit{wait}$, and repeat this pattern. 
    
    \item $\alpha_4' = \bigwedge_{i} \big(\globally(\textit{cs}_i \implies \neg \PnF_{[0,10]} (\textit{cs}_i, \textit{cs}_i))\big)$:  
    Same as $\theta_3$ but with $10$ time units windows. This holds as any process must take $> 5$ time units from $\ell_\textit{wait}$ to $\ell_\textit{cs}$.
    
    \item $\alpha_5' = \bigwedge_{i} \big(\globally(\textit{req}_i \implies \eventually_{[0,10]} \textit{cs}_i) \big)$:  
    Every request by process $i$ must be followed by process $i$ entering the critical section within $10$ time units.
    It is clear that this does not hold.
\end{itemize}
The results are in~\cref{tab:fischer}. As one may expect, it takes longer to verify the formulae that hold in the system model (i.e.~the negated formulae are not satisfiable by the system model). 

\begin{table}[!htb]
 \caption{Execution times on the Fischer benchmarks.  
Times are in seconds. `TO' indicates timeouts (300s).}
\label{tab:fischer}
\centering
\scalebox{0.8}{
\def\arraystretch{1.2}
\setlength\tabcolsep{2mm}
\begin{tabular}{l@{\hspace{2mm}}rrrrr}
\toprule 
&  &
 & \multirowcell{2}{$\mightyppl{}$ \\ \small \texttt{tck -{}-inf comp}} \\
$\varphi$  & $N$ & Hold? & \\
\otoprule
$\alpha_1'$  & 2     &  \xmark      & 0.033          \\ 
$\alpha_1'$  & 3     &  \xmark      & 0.212         \\ 
$\alpha_1'$  & 4     &  \xmark      & 30.095         \\ 
$\alpha_2'$  & 2     &  \xmark      & 77.571           \\ 
$\alpha_2'$  & 3     &  \xmark      & TO           \\ 
$\alpha_2'$  & 4     &  \xmark      & TO          \\ 
$\alpha_3'$  & 2     &  \xmark      & 0.098          \\ 
$\alpha_3'$  & 3     &  \xmark      & 1.641          \\ 
$\alpha_3'$  & 4     &  \xmark      & 62.947         \\ 
$\alpha_4'$  & 2     &  \cmark      & 0.747          \\ 
$\alpha_4'$  & 3     &  \cmark      & 57.668          \\ 
$\alpha_4'$  & 4     &  \cmark      & TO         \\ 
$\alpha_5'$  & 2     &  \xmark      & 0.012          \\ 
$\alpha_5'$  & 3     &  \xmark      & 0.021          \\ 
$\alpha_5'$  & 4     &  \xmark      & 0.127         \\ 
\bottomrule
\end{tabular}
}
\end{table}

\subsubsection{Pinwheel scheduling}

 The \emph{pinwheel scheduling} problem~\citep{Holte1989pinwheel} asks
 whether there is an \emph{infinite} schedule for tasks $\{i \mid 1 \leq i \leq k\}$ (each with a \emph{relative deadline     } $a_i \geq 2$),
e.g.,~on each day we can schedule a task, and
 each task is scheduled at least once in every $a_i$ days.
 We use the model $\ta{}$ $\mathcal{M}_\textit{pin}^k$ to enforce that exactly one of $\{p_1, \dots, p_k\}$ holds on any event, and any two events must be separated by at least $1$ time unit. We then model-check $\mathcal{M}_\textit{pin}^k$ against the corresponding formula
 for an instance $(a_1, \dots, a_k)$ of the problem:
 \[
\varphi_{(a_1, \dots, a_k)}  = (\bigwedge_{i \in \{1, \dots, k\}} \eventually p_i) 
    \implies \neg \bigwedge_{i \in \{1, \dots, k\}} \globally (p_i \implies \eventually_{[0, a_i]} p_i) \;.
\]
The results are in Table~\ref{tab:pinwheel}. 
For these experiments, the \texttt{-u0} option was necessary to disable the subsumption abstraction in \textsc{LTSmin}, as the default configuration caused a memory bottleneck that led to incorrect results.
To highlight the effect of using multiple cores, we exclude the time used for constructing, flattening, and compiling the models. The single-threaded performance of \textsc{TChecker} in the `\texttt{flat}' mode is generally comparable to that of \textsc{LTSmin}, yet it performs significantly better on more challenging instances.
 The latter, however, can run on multiple threads and we can see that a speedup is generally observed as the number of threads increases, and this improvement is sometimes linear with respect to the thread count.  Furthermore, we notice that using more than $16$ threads does not always yield significant performance gains, likely due to hardware limitations and increased synchronisation overhead.

 \begin{table}[!htb]
 \caption{Execution times on the pinwheel scheduling benchmarks. 
 Times are in seconds. Numbers in the heading are numbers of parallel threads. `TO' indicates timeouts (300s).}
 \label{tab:pinwheel}
 \centering
 \scalebox{0.7}{
 \def\arraystretch{1.2}
 \setlength\tabcolsep{2mm}
 \begin{tabular}{l@{\hspace{2mm}}rrrrrrrrrr}
 \toprule %
 & \multirowcell{2}{ \\ \\ \\ Hold?} & \multirowcell{2}{$\mightyppl{}$ \\ \small \texttt{tck -{}-inf comp}} & \multirowcell{2}{$\mightyppl{}$ \\ \small \texttt{tck -{}-inf flat}} & \multirowcell{2}{$\mightyppl{}$ \\ \small \texttt{fp -{}-inf}} & \multicolumn{6}{c}{\multirowcell{2}{$\mightyppl{}$ \\ \small \texttt{opaal\_ltsmin flat}}}  \\
 & & & & \\
 \cmidrule[\heavyrulewidth](l{2mm}r{2mm}){6-11}
 $\varphi$ & & & & & 1 & 2 & 4 & 8 & 16 & 32 \\
 \otoprule
 $(3, 4, 5, 7)$             & \cmark       & 1.400    & 0.112   & 0.821          &   0.048  & 0.033  & 0.040   & 0.066     & 0.152   & 0.407       \\
 $(3, 4, 5, 8)$             & \xmark       & 0.190    & 0.093   & 1.745          &   0.020  & 0.041  & 0.038   & 0.070     & 0.158   & 0.405          \\
 $(2, 7, 9, 9, 16)$         & \cmark       & TO       & 1.956   & 13.870         &   1.303  & 0.706  & 0.407   & 0.293     & 0.388   & 0.583             \\
 $(2, 7, 9, 10, 16)$        & \xmark       & 5.548    & 0.555   & 66.768         &   0.110  & 0.056  & 0.050   & 0.090     & 0.164   & 0.362  \\
 $(3, 7, 8, 8, 9, 10)$      & \cmark       & TO       & 14.561  & TO             &   36.676 & 19.495 & 8.990   & 5.352     & 6.930   & 4.504  \\
 $(3, 7, 8, 8, 9, 11)$      & \xmark       & TO       & 3.806   & TO             &   1.496  & 1.048  & 0.827   & 0.599     & 2.809   & 0.726  \\
 $(4, 6, 8, 9, 9, 9, 15)$   & \cmark       & TO       & TO      & TO             &   TO     & TO     & TO      & TO        & TO      & TO       \\
 $(4, 6, 8, 9, 9, 9, 16)$   & \xmark       & TO       & 28.066  & TO             &  14.659  & 15.363 & 4.953   & 3.935     & 3.618   & 2.800       \\
 $(5, 5, 8, 9, 9, 9, 11)$   & \cmark       & TO       & 247.438 & TO             &   TO     & TO     & TO      & 275.269   & 191.635 & 133.457      \\
 $(5, 5, 8, 9, 9, 9, 12)$   & \xmark       & TO       & 28.166  & TO             & 12.099   & 24.472 &  3.587  & 10.727    & 5.154   & 5.477     \\
 $(6, 7, 7, 7, 8, 8, 10)$   & \cmark       & TO       & 147.104 & TO             &   TO     & TO     & 299.481 & 166.373   & 106.729 & 71.496    \\
 $(6, 7, 7, 7, 8, 8, 11)$   & \xmark       & TO       & 27.746  & TO             &   15.170 & 21.481 & 15.673  & 2.845     & 7.506   & 10.742       \\
 \bottomrule
 \end{tabular}
 }
 \end{table}

\section{Conclusion}
\label{sec:cn}
We have presented a fully compositional reduction from $\mitlpp{}$ to timed automata, addressing one of the most expressive decidable logics for specifying real-time behaviors. Our approach supports both past and Pnueli modalities and handles $\mitl{}$ modalities with arbitrary non-singular intervals using novel abstraction and sequentialisation techniques. These methods allow us to avoid the complex monolithic constructions of prior work, enabling both theoretical transparency and implementation modularity. Extensive experimental results on benchmarks affirm the feasibility and scalability of our approach.

As part of ongoing work, we aim to verify the correctness of the construction via formal methods tools, such as interactive proof assistants. Additionally, tighter integration with model checkers that natively support symbolic alphabets (e.g., \textsc{LTSmin} and Spot~\citep{duret2022spot}) may lead to further performance gains. We also plan to generalise the framework to handle freeze-quantified fragments of real-time logic, including the recently studied decidable variants of $\tptl{}$. We conjecture that such multi-clock freeze logics---such as $\uptl{}$~\citep{concur23, concur25}---may admit exponentially more succinct representations, opening new avenues for efficient specification and verification.

\bibliographystyle{ACM-Reference-Format}
\bibliography{biblio}

\appendix
\clearpage

\section{Tester $\ta{}$ for $\phi_1 \until_{\geq l} \phi_2$}
\label{app:until.l.infty.tester}

\begin{figure}[!htbp]
\centering
\begin{tikzpicture}[->, node distance=5cm, transform shape, scale=0.7]
    \node[initial left, state, accepting](0) {$s_0$};
    \node[state,right=5.5cm of 0] (1) {$s_1$};
    \node[state,accepting,right=5.5cm of 1] (2) {$s_2$};

    \path[->] (0) edge[above,bend left=10]
    node{$p_\psi$, $x:=0$} (1)%
    (0) edge[loopbelow,below,looseness=20]
    node[align=center]{$\lnot p_\psi$, $x:=0$} (0)%
    (1) edge[below,bend left=10]
    node{$\lnot p_\psi \land \overline{\phi_2} \land x\geq l$, $x:=0$}
    (0)%
    (1) edge[loopbelow,below,looseness=20]
    node[align=center] {$\lnot p_\psi \land \overline{\phi_1} \land \newneg{\phi_2}$\\
            $\lnot p_\psi \land \simplify{\phi_1} \land x< l$\\
            $p_\psi \land \simplify{\phi_1} \land \newneg{\phi_2}$, $x:=0$ \\
				$p_\psi \land \simplify{\phi_1} \land x<l$, $x:=0$
				} (1)%

    (1) edge[above,bend left=10]
    node[align=center]{$\lnot p_\psi \land \simplify{\phi_1} \land \simplify{\phi_2}$ \\
    $p_\psi \land \simplify{\phi_1} \land \simplify{\phi_2}$, $x := 0$ \\
    $p_\psi\land \simplify{\phi_2}\land x\geq l$, $x:=0$
    } (2)%
    (2) edge[above,bend right=55]
    node{$\lnot p_\psi \land \simplify{\phi_2}\land x\geq l$, $x:=0$}
    (0)%
    (2) edge[below,bend left=10]
    node[align=center] {$\lnot p_\psi \land \overline{\phi_1} \land \newneg{\phi_2}$\\
            $\lnot p_\psi \land \simplify{\phi_1} \land x< l$\\
            $p_\psi \land \simplify{\phi_1} \land \newneg{\phi_2}$, $x:=0$ \\
				$p_\psi \land \simplify{\phi_1} \land x<l$, $x:=0$
				}
	(1)%
    (2) edge[loopbelow,below,looseness=20]
    node[align=center]{$\lnot p_\psi \land \simplify{\phi_1} \land \simplify{\phi_2}$ \\
    $p_\psi \land \simplify{\phi_1} \land \simplify{\phi_2}$, $x := 0$ \\
    $p_\psi\land \simplify{\phi_2}\land x\geq l$, $x:=0$
    } (2);
  \end{tikzpicture}
\caption{The tester $\ta{}$ for $\phi_1 \until_{[l, \infty)} \phi_2$.}
  \label{fig:until_a_infty}
\end{figure}
  
\section{Component $\ta{s}$ for past modalities}
\label{app:past.components}

\begin{figure}[!htbp]
\centering
\begin{tikzpicture}[->, node distance=5cm, transform shape, scale=0.7]
   \node[initial left,state, accepting](0){$s_0$};
   \node[state, right of=0, accepting](1){$s_1$};
   
   \path
   (0) edge[loopabove] node[above=1mm, align=center, looseness=20, out=120, in=90]{$\lnot p_\psi$, $x := 0$} (0)
   (0) edge[->, bend left=10] node[above=1mm, align=center]{$\lnot p_\psi \land \overline{\phi_2}$, $x := 0$} (1)
   (1) edge[loopabove] node[above=1mm,align=center]{$\overline{\phi_1} \land \lnot \overline{\phi_2}$, $x < u$ \\ $p_\psi \land \overline{\phi_2}$, $x < u$, $x := 0$} (1)
   (1) edge[->, bend left=10] node[below=1mm, align=center]{$p_\psi$ \\ $x < u$, $x := 0$} (0);
 \end{tikzpicture}
\caption{The tester $\ta{}$ for $\phi_1 \since_{[0, u)} \phi_2$.}
\label{fig:since.0.u}
\end{figure}

\begin{figure}[!htbp]
\centering
\begin{tikzpicture}[->, node distance=7cm, transform shape, scale=0.7]
   \node[initial left,state, accepting](0){$s_0$};
   \node[initial below, state, right of=0, accepting](1){$s_1$};
   \node[initial right, state, right of=1, accepting](2){$s_2$};
   
   \path
   (0) edge[loopabove] node[above=1mm, align=center, looseness=20, out=120, in=90]{$\lnot p_\psi$, $x := 0$} (0)
   (0) edge[->, bend left=10] node[above=1mm, align=center]{$\lnot p_\psi \land \overline{\phi_1} \land \overline{\phi_2}$, $x := 0$} (1)
   (0) edge[->, bend right=40] node[below=1mm, align=center]{$\lnot p_\psi$, $x := 0$} (2)
   (1) edge[loopabove] node[above=1mm,align=center]{$\neg \overline{\phi_1} \land  \overline{\phi_2}$, $x < u$ \\ $p_\psi \land \overline{\phi_1} \land \overline{\phi_2}$, $x < u$, $x := 0$} (1)
   
   (1) edge[->, bend left=10] node[above=1mm,align=center]{$p_\psi \land \neg \overline{\phi_1} \land  \overline{\phi_2}$, $x < u$ \\ $p_\psi$, $x < u$, $x := 0$} (2)
   (2) edge[loopabove] node[above=1mm, align=center]{$p_\psi \land \lnot \overline{\phi_1} \land \overline{\phi_2}$, $x \geq u$ \\ $\lnot p_\psi \land \lnot \overline{\phi_1} \land \overline{\phi_2}$ \\ $p_\psi$, $x \geq u$, $x := 0$} (2)
   (2) edge[->, bend right = 40] node[above=1mm, align=center]{$p_\psi$, $x \geq u$, $x := 0$} (0)
   (2) edge[->, bend left =10] node[below=1mm, align=center]{$p_\psi \land \overline{\phi_1} \land \overline{\phi_2}$, $x \geq u$, $x := 0$} (1)
   (1) edge[->, bend left=10] node[below=1mm, align=center]{$p_\psi$, $x < u$, $x := 0$} (0);
 \end{tikzpicture}
\caption{The tester $\ta{}$ for $\phi_1 \trigger_{[0, u)} \phi_2$.}
\label{fig:trigger.0.u}
\end{figure}

\begin{figure}[!htbp]
\centering
\begin{tikzpicture}[->, node distance=7cm, transform shape, scale=0.7]
   \node[initial left,state, accepting](0){$s_0$};
   \node[initial below, state, right of=0, accepting](1){$s_1$};
   \node[initial right, state, right of=1, accepting](2){$s_2$};
   
   \path
   (0) edge[loopabove] node[above=1mm, align=center, looseness=20, out=120, in=90]{$\lnot p_\psi$, $x := 0$} (0)
   (0) edge[->, bend left=10] node[above=1mm, align=center]{$\lnot p_\psi \land \overline{\phi_1} \land \overline{\phi_2}$, $x := 0$} (1)
   (0) edge[->, bend right=40] node[below=1mm, align=center]{$\lnot p_\psi \land \overline{\phi_1}$, $x := 0$} (2)
   
   (1) edge[loopabove] node[above=1mm,align=center]{$\neg \overline{\phi_1} \land \overline{\phi_2}$ \\ $p_\psi \land \overline{\phi_1} \land \overline{\phi_2}$, $x > l$, $x := 0$} (1)
   
   (1) edge[->, bend left=10] node[above=1mm,align=center]{$\lnot \overline{\phi_1}$, $x := 0$ \\ $p_\psi \land \overline{\phi_1}$, $x > l$, $x := 0$} (2)
   (2) edge[loopabove] node[above=1mm, align=center]{$\lnot \overline{\phi_1}$, $x \leq l$ \\ $p_\psi \land \land \overline{\phi_1}$, $x \leq l$, $x := 0$} (2)
   (2) edge[->, bend right = 40] node[above=1mm, align=center]{$p_\psi$, $x \leq l$, $x := 0$} (0)
   (2) edge[->, bend left =10] node[below=1mm, align=center]{$p_\psi \land \overline{\phi_1} \land \overline{\phi_2}$, $x \leq l$, $x := 0$} (1)
   (1) edge[->, bend left=10] node[below=1mm, align=center]{$p_\psi$, $x > l$, $x := 0$} (0);
 \end{tikzpicture}
\caption{The tester $\ta{}$ for $\phi_1 \trigger_{(l, \infty)} \phi_2$.}
\label{fig:trigger.l.infty}
\end{figure}

  \begin{figure}[!htbp]
    \centering
    \begin{tikzpicture}[->, node distance=6cm, transform shape, scale=0.55]
        \node[initial left,state, accepting](0){$s_0$};
        \node[initial below, state, below right = 2cm and 2cm of 0, accepting](0'){$s_0'$};
        \node[initial below, state, right of=0', accepting](1){$s_1$};
        \node[initial below, state, right of=1, accepting](2){$s_2$};
        \node[initial below, state, right of=2, accepting](3){$s_3$};
        
        \path
        (0) edge[loopabove, ->, looseness=20, out=135, in=90] node[above=1mm, align=center]{$\neg p_\psi^1$, $x_1 := 0$} (0)
        (0') edge[->] node[left=1mm, align=center]{$p_\psi^1$, $x_1 \geq u$, $x_1 := 0$} (0)                    
        (0') edge[loopabove, ->, out=80, in=50, looseness=10] node[above right=1mm and 1mm, align=center]{$\neg p^1_\psi \wedge  \overline{\chi_1}$ \\ $p^1_\psi \wedge \overline{\chi_1}$, $x_1 \geq u$} (0')
        (1) edge[loopabove, ->, out=90, in=60, looseness=10] node[above right=1mm and 1mm, align=center]{$\neg  p^1_\psi \wedge  \overline{\chi_2}$ \\ $p^1_\psi \wedge \overline{\chi_2}$, $x_1 \geq u$} (1)
        (2) edge[loopabove, ->, out=90, in=60, looseness=10] node[above right=1mm and 1mm, align=center]{$\neg  p^1_\psi \wedge  \overline{\chi_3}$ \\ $p^1_\psi \wedge \overline{\chi_3}$, $x_1 \geq u$} (2)
        (3) edge[loopabove, ->, out=90, in=60, looseness=10] node[above right=1mm and 1mm, align=center]{$\neg  p^1_\psi \wedge  \overline{\chi_4}$ \\ $p^1_\psi \wedge \overline{\chi_4}$, $x_1 \geq u$} (3)
        
        (1) edge[->] node[below=1mm, align=center]{ $\neg  p^1_\psi \land \neg \overline{\chi_1}$ \\ $p^1_\psi \land \neg \overline{\chi_1}$, $x_1 \geq u$} (0')                    
        (2) edge[->] node[below=1mm, align=center]{ $\neg  p^1_\psi \land \neg \overline{\chi_2}$ \\ $p^1_\psi \land \neg \overline{\chi_2}$, $x_1 \geq u$} (1)                    
        (3) edge[->] node[below=1mm, align=center]{ $\neg  p^1_\psi \land \neg \overline{\chi_3}$ \\ $p^1_\psi \land \neg \overline{\chi_3}$, $x_1 \geq u$} (2)                    
        (0) edge[loopabove,->,looseness=1.5,out=60,in=110] node[above=1mm, align=center]{$\neg p_\psi^1$, $x_1 := 0$} (3)
        (0) edge[loopabove,->,looseness=1.5,out=60,in=110] node[above=1mm, align=center]{$\neg p_\psi^1$, $x_1 := 0$} (2)
        (0) edge[loopabove,->,looseness=1.5,out=60,in=110] node[above=1mm, align=center]{$\neg p_\psi^1$, $x_1 := 0$} (1)
        (0) edge[loopabove,->,looseness=1.5,out=60,in=110] node[right=1mm, align=center]{$\neg p_\psi^1$, $x_1 := 0$} (0')
        
        (0') edge[loopleft,->,looseness=20,out=150,in=170] node[below left=1mm and 1mm, align=center]{$p_\psi^1$, $x_1 \geq u$, $x_1 := 0$} (0')
        (0') edge[loopbelow,->,looseness=1.5,out=-115,in=-120] node[below right=1mm and -2mm, align=center]{$p_\psi^1$, $x_1 \geq u$, $x_1 := 0$} (1)
        (0') edge[loopbelow,->,looseness=1.5,out=-115,in=-120] node[below right=1mm and 1mm, align=center]{$p_\psi^1$, $x_1 \geq u$, $x_1 := 0$} (2)
        (0') edge[loopbelow,->,looseness=1.5,out=-115,in=-120] node[below right=1mm and 1mm, align=center]{$p_\psi^1$, $x_1 \geq u$, $x_1 := 0$} (3);
      \end{tikzpicture}                    
    \caption{A component $\ta{}$ of the tester $\ta{}$ for  $\dualPnO_{< u}(\phi_1,\ldots, \phi_n)$.}
    \label{fig:pndual.past}                    
  \end{figure}

\end{document}